\documentclass[a4paper, 11pt]{amsart}
\usepackage[active]{srcltx}
\usepackage[all]{xy}
\usepackage{subfig}
\usepackage{soul}
\usepackage{hyperref}
\usepackage{mathrsfs}
\usepackage{esint}
\usepackage{amsmath}
\usepackage{amssymb,latexsym}
\usepackage{mathrsfs}
\usepackage{graphics}
\usepackage{latexsym}
\usepackage{psfrag}
\usepackage{import}
\usepackage{verbatim}
\usepackage{graphicx}
\usepackage[usenames]{color}
\usepackage{pifont,marvosym}
\usepackage[normalem]{ulem}
\usepackage{mathtools}

\theoremstyle{plain}
\newtheorem{lemma}{Lemma}[section]
\newtheorem{theorem}[lemma]{Theorem}
\newtheorem{proposition}[lemma]{Proposition}
\newtheorem{corollary}[lemma]{Corollary}

\theoremstyle{definition}

\newtheorem{definition}[lemma]{Definition}
\newtheorem{remark}[lemma]{Remark}

\numberwithin{equation}{section}

\newcommand{\R}{\mathbb{R}}
\newcommand{\N}{\mathbb{N}}

\newcommand{\dist}{\text{\rm dist}}
\newcommand{\supp}{\text{\rm supp}}

\newcommand{\id}{\mathrm{Id}}

\newcommand{\ve}{\varepsilon}

\newcommand{\mm}{\mathfrak m}

\newcommand{\sfd}{\mathsf d}
\newcommand{\cP}{\mathcal P}

\newcommand{\vol}{\mathrm{vol}}
\newcommand{\Ric}{\mathrm{Ric}}
\newcommand{\Det}{\mathrm{Det}}
\newcommand{\Hess}{\mathrm{Hess}}
\newcommand{\Ent}{\mathrm{Ent}}
\newcommand{\rG}{\mathrm{G}}
\newcommand{\TCD}{\mathsf{TCD}}
\newcommand{\wTCD}{\mathsf{wTCD}}

\newcommand{\cA}{\mathcal{A}}
\newcommand{\cB}{\mathcal{B}}
\newcommand{\cU}{\mathcal{U}}
\newcommand{\cC}{\mathcal{C}}
\newcommand{\cL}{\mathcal{L}}
\newcommand{\cK}{\mathcal{K}}

\newcommand{\cH}{\mathcal{H}}

\newcommand{\Int} {{\rm Int}}
\newcommand{\Cpl} {{\rm Cpl}}
\newcommand{\ee} {{\rm e}}
\newcommand{\tr} {{\rm tr}}
\newcommand{\Tr} {{\rm Tr}}

\newcommand{\red}{\color{red}}

\begin{document}

\title[An optimal transport formulation of the Einstein  equations]{  An optimal transport formulation of the Einstein equations of general relativity}
\author{A. Mondino}  \thanks{A.  Mondino: University of Oxford,  Mathematical Institut,  email: Andrea.Mondino@maths.ox.ac.uk. Supported by the EPSRC First Grant  EP/R004730/1 ``Optimal transport and
geometric analysis'' and by the ERC Starting Grant  802689  ``CURVATURE''}
\author{S. Suhr}  \thanks{S. Suhr: University of Bochum, email: Stefan.Suhr@ruhr-uni-bochum.de.  Supported by the SFB/TRR 191 ``Symplectic Structures in Geometry, Algebra and Dynamics'', funded by the Deutsche Forschungsgemeinschaft}


\bibliographystyle{plain}

\maketitle

\begin{abstract} 

The goal of the paper is to give an optimal transport formulation of the full Einstein equations of general relativity, linking the (Ricci) curvature of a space-time with the cosmological constant and the energy-momentum tensor. Such an optimal transport formulation is in terms of convexity/concavity properties of the Shannon-Bolzmann entropy along curves of probability measures extremizing suitable optimal transport costs. The result gives a new connection between general relativity and  optimal transport; moreover it gives a mathematical reinforcement of the strong link between general relativity and   thermodynamics/information theory that emerged in the physics literature of the  last years.

\end{abstract}

\tableofcontents

\section{Introduction}
In recent years, optimal transport revealed to be a very effective and innovative tool in several fields of mathematics and applications. By way of example, let us mention fluid mechanics (e.g. Brenier \cite{Br77}  and Benamou-Brenier \cite{BeBr00}), partial differential equations (e.g. Jordan-Kinderleher-Otto \cite{JKO98} and Otto \cite{O01}), random matrices (e.g. Figalli-Guionnet \cite{FiGu}), optimization (e.g.  Bouchitt\'e-Buttazzo \cite{BB1}), non-linear $\sigma$-models (e.g. Carfora \cite{Car}),  geometric and functional inequalities (e.g. Cordero-Erausquin-Nazaret-Villani \cite{CENV}, Figalli-Maggi-Pratelli \cite{FMP}, Klartag \cite{Kl}, Cavalletti- Mondino \cite{CaMo}) Ricci curvature in Riemannian geometry (e.g. Otto-Villani \cite{OV}, Cordero Erausquin-McCann-Schmuckenschl\"ager \cite{CMS}, Sturm-VonRenesse \cite{SVR})  and in metric measure spaces (e.g. Lott-Villani \cite{lottvillani:metric}, Sturm \cite{sturm:I,sturm:II}, Ambrosio-Gigli-Savar\'e \cite{AGS}). For more details about optimal transport and its applications in both pure and applied mathematics, we refer the reader to the many books on the topic, e.g. \cite{AGUser, AGSBook, Sant,Vil, VilTopics}.

Here let us just quote two of the many applications to partial differential equations. In the pioneering work of Jordan-Kinderleher-Otto \cite{JKO98}  it was discovered a new optimal transport formulation of the Fokker-Planck equation (and in particular of the heat equation) as a gradient flow of a suitable functional (roughly, the Boltzmann-Shannon entropy defined below in \eqref{eq:defEntintro} plus a potential)  in the Wasserstein space (i.e. the space of probability measures with finite second  moments endowed with the quadratic Kantorovich-Wasserstein distance); later, Otto \cite{O01} found a related optimal transport formulation of the porous medium equation. The impact of these works in the optimal transport community has been huge, and opened the way to  a more general theory of gradient flows (see for instance the monograph by Ambrosio-Gigli-Savar\'e \cite{AGSBook}). 

The goal of the present work is to give a new optimal transport formulation of another fundamental class of partial differential equations: the Einstein equations of general relativity. First published by Einstein in 1915, the Einstein equations describe gravitation as a result of space-time being curved by mass and energy; more precisely,   the space-time (Ricci) curvature is related to the local energy and momentum expressed by the energy-momentum tensor.
Before entering into the topic, let us first recall that the Einstein equations are \emph{hyperbolic} evolution equations (for a comprehensive treatment see the recent monograph by Klainerman-Nicol\'o \cite{KlNi}).  Instead of a gradient flow/PDE approach, we will see the evolution from a geometric/thermodynamic/information point of view.

Next we briefly recall the formulation of the  Einstein equations.  Let $M^{n}$ be an $n$-dimensional manifold ($n\geq 3$, the physical dimension being $n=4$) endowed with a Lorentzian metric $g$, i.e. $g$ is a nondegenerate symmetric bilinear form of signature $(-++\ldots+)$.
Denote with $\Ric$ and ${\rm Scal}$ the Ricci and the scalar curvatures of $(M^{n},g)$. The Einstein equations read as
\begin{equation}\label{eq:EFEintro}
\Ric-\frac{1}{2} {\rm Scal} \, g +\Lambda g=8\pi T,
\end{equation}
where  $\Lambda \in \R$ is the cosmological constant, and $T$ is the energy-momentum tensor.  Physically, the cosmological constant $\Lambda$ corresponds to the energy density of the vacuum; the energy-momentum tensor is a symmetric bilinear form on $M$ representing the density of energy and momentum, acting as  the source of the gravitational field.

\subsection{Statement of the main results}

Recall that in  
 a Lorentzian manifold $(M^{n},g)$,  a non-zero tangent vector $v\in T_{x}M$ is called \emph{time-like} if $g(v,v)<0$.   If $M$ admits a continuous no-where vanishing time-like vector field $X$, then $(M,g)$ is said to be \emph{time-oriented} and it
 is called a \emph{space-time}. The vector field $X$ induces a partition on the set of time-like vectors, into two equivalence classes: the \emph{future pointing} tangent vectors  $v$ for which $g(X,v)< 0$ and the  \emph{past pointing}  tangent vectors  $v$ for which $g(X,v)> 0$.  The closure of the set of future pointing time-like vectors is denoted 
$$
\cC={\rm Cl}(\{ v \in TM:\,\, g(v,v)<0 \textrm{ and } g(X,v)<0\})\subset TM.
$$
A physical particle moving in the space-time $(M,g,\cC)$ is represented by a  \emph{causal curve}
which is an absolutely continuous curve, $\gamma$, satisfying 
$$\dot{\gamma}_{t}\in \cC \textrm{ a.e. } t \in [0,1].$$
If the particle cannot reach the speed of light (e.g. massive particle), then it is represented by a \emph{chronological} curve
which is an absolutely continuous curve, $\gamma$, satisfying 
$$\dot{\gamma}_{t} \in \Int(\cC) \textrm{ a.e. } t\in [0,1],$$
where  $\Int(\cC)$ is the interior of the cone $\cC$ made of future pointing time-like vectors.  
The Lorentz length of a causal curve is 
$$
L_g(\gamma)= \int_0^1 \sqrt{-g(\dot{\gamma}_{t}, \dot{\gamma}_{t})} \, dt.
$$
A point $y$ is in the future of $x$, denoted $y>>x$, if there is a future oriented chronological curve from $x$ to $y$; in this case,
 the {\em Lorentz distance} or {\em proper time } between $x$ and $y$ is defined by
$$
\sup\{ L_g(\gamma):\,\, \gamma_{0}=x \textrm{ and } \gamma_{1}=y , \; \text{$\gamma$ chronological}\}>0,
$$
which is achieved by a geodesic which is called a {\em maximal geodesic}.  See for example \cite{EH, oneill,W} 

In this paper, we consider the following Lorentzian Lagrangian on $TM$ for $p\in (0,1)$:
\begin{equation}\label{def:pLangrangianLpIntro}
\cL_{p}(v):=
\begin{cases}
-\frac{1}{p}(-g(v,v))^{\frac{p}{2}} & \text{ if $v\in \cC$}\\
+\infty &  \text{otherwise.}
\end{cases}
\end{equation}
Note that if $p$ were $1$ this would be the negative of the integrand for the Lorentz length given above.   
Here we study $p\in (0,1)$
because this  Lorentzian Lagrangian $\cL_{p}$ has good convexity properties for such $p$ [see Lemma \ref{lem:Lpconvex}]. 

%

Let  ${\rm AC}([0,1], M)$ denote the space of absolutely continuous curves from $[0,1]$ to $M.$ The Lagrangian action $\cA_{p}$, corresponding to the Lagrangian $\cL_{p}$ and defined for any $\gamma\in {\rm AC}([0,1], M)$,  is given by
\begin{equation}\label{def:pActionApIntro}
\cA_{p}(\gamma):=\int_{0}^{1} \cL_{p}(\dot{\gamma}_{t}) dt \in (-\infty, 0]\cup\{+\infty\}.
\end{equation}
Observe that $ \cA_{p}(\gamma)\in (-\infty, 0]$ if and only if $\gamma$ is a causal curve.
Note that, if $p$ were $1$, this would be the negative of the Lorentz length of $\gamma$ or the proper time along $\gamma$.
Thus $-\cA_{p}(\dot{\gamma})$ can be seen as a kind of non-linear $p$-proper time along $\gamma$, enjoying better convexity properties.  The reader may note the parallel with the theory of  Riemannian  geodesics, where one often studies the energy functional $\int |\dot{\gamma}|^{2}$ in place of the length functional $\int |\dot{\gamma}|$, due to the analogous advantages.

The choice of the minus sign in (\ref{def:pActionApIntro}) is motivated by optimal transport theory, in order to have a minimization problem instead of a maximization one (as in the sup defining the Lorentz distance between points above).  It is readily checked that the critical points of $\cA_{p}$ with negative action are time-like geodesics [see Lemma \ref{lemma_minimizer}]. The advantage of $\cA_{p}$ with $p\in (0,1)$ is that it automatically selects an affine parametrization for its critical points with negative action. 

The cost function, $c_{p}: \;M\times M\to (-\infty, 0]\cup\{+\infty\} $,  relative to the $p$-action $\cA_{p}$ is
defined by
$$
c_{p}(x,y)=\inf \{\cA_{p}(\gamma)\,:\, \gamma\in {\rm AC}([0,1], M), \gamma_{0}=x, \gamma_{1}=y\}.
$$
Note that, if $p$ were $1$, then this would be the negative of the Lorentz distance between $x$ and $y$.

%
Consider a  relatively  compact  open subset 
$$
E\subset\subset  \Int(\cC) \subset TM
\textrm{ and }
r\in (0, {\rm inj}_{g}(E)),
$$
where ${\rm inj}_{g}(E)>0$ is  the injectivity radius of the exponential map of $g$ restricted to $E$.
If we take $p_{TM\to M}:TM\to M$ to be the canonical projection map then
$$
\forall x\in p_{TM\to M}(E) \textrm{ and } v \in T_xM\cap E \textrm{ with } g(v,v)=-r^2,
$$
we have a maximal geodesic $\gamma_x:[0,1]\to M$ defined by
$$
\gamma_x(t)=\exp_x((t-1/2)v) \textrm{ such that } \gamma_x(1/2)=x.
$$ 
The Ricci curvature, $\Ric_x(v,v)$ at a point $x\in M$ in the direction $v$, is a trace of the
curvature tensor so that intuitively it
measures the average way in which geodesics near $\gamma_x$ bend towards or
away from it.  See Section~\ref{Ex-FLRW}. In Riemannian geometry, the Ricci curvature influences the volumes of 
balls.  Here, instead of balls, we define
for any $x\in p_{TM\to M}(E)$ 
$$B^{g,E}_{r}(x):= \{\exp_{x}^{g}(t w): w\in T_{x}M\cap E, \, g(w,w)=-1, t\in [0,r] \}$$
precisely to avoid the null directions.

Rather than considering individual paths between a given pair of points, we will
consider distributions of paths between a given pair of distributions of points using the optimal transport
approach.

We denote by ${\cP}(M)$ the set of Borel probability measures on $M$. For any $\mu_1,\mu_{2}\in \cP(M)$, we say that  a Borel probability measure 
$$\pi\in \cP(M\times M) \textrm{ is a coupling of }\mu_{1} \textrm{ and }\mu_{2}
$$
 if $(p_{i})_{\sharp}\pi=\mu_{i}, i=1,2$, where $p_{1}, p_{2}:M\times M\to M$ are the projections onto the first and second coordinate. Recall that the push-forward $(p_{1})_{\sharp} \pi$ is defined by  
$$(p_{1})_{\sharp} \pi (A):=\pi \big(p_{1}^{-1}(A)\big)$$ for any Borel subset $A\subset M$. The set of couplings of $\mu_{1},\mu_{2}$ is denoted by $\Cpl(\mu_{1},\mu_{2})$. The \emph{$c_{p}$-cost of a coupling} $\pi$ is  given by 
$$
\int_{M\times M} c_{p}(x,y) d\pi(x,y) \in  [-\infty, 0]\cup \{+\infty\}.
$$
Denote by $C_{p}(\mu_{1}, \mu_{2})$ the \emph{minimal cost relative to $c_{p}$} among all couplings from $\mu_{1}$ to $\mu_{2}$, i.e.
$$
C_{p}(\mu_{1}, \mu_{2}):=\inf \left\{\int c_{p} d\pi \,:\, \pi \in \Cpl(\mu_{1},\mu_{2}) \right\}\in  [-\infty, 0]\cup \{+\infty\}.
$$
If $C_{p}(\mu_{1}, \mu_{2})\in\R$, a coupling achieving the infimum  is said to be $c_{p}$-\emph{optimal}. 

For $t\in [0,1]$ denote by $\ee_{t}: {\rm AC}([0,1], M)\to M$ the evaluation map 
$$
\ee_{t}(\gamma):=\gamma_{t}.
$$ 
A \emph{ $c_{p}$-optimal dynamical plan} is a probability measure $\Pi$ on ${\rm AC}([0,1], M)$ such that $(\ee_{0}, \ee_{1})_{\sharp} \Pi$ is a  $c_{p}$-optimal   coupling from  $\mu_{0}:=(\ee_{0})_{\sharp}\Pi$ to  $\mu_{1}:=(\ee_{1})_{\sharp}\Pi$. 
One can naturally associate to $\Pi$ a curve  
$$\left(\mu_{t}:=(\ee_{t})_{\sharp}\Pi \right)_{t\in [0,1]}\subset {\cP}(M)$$ of probability measures. The condition that $\Pi$ is a  $c_{p}$-optimal dynamical plan corresponds to saying that the curve $(\mu_{t})_{t\in [0,1]}\subset {\cP}(M)$ is a length minimizing geodesic with respect to $C_{p}$, i.e. 
$$C_{p}(\mu_{s}, \mu_{t})=|t-s| C_{p}(\mu_{0}, \mu_{1}) \quad \forall s,t\in [0,1].
$$

We will mainly consider a special class of $c_{p}$-optimal dynamical plans, that we call \emph{regular}: roughly, a $c_{p}$-optimal dynamical plan is said to be regular if it is obtained by exponentiating the gradient (which is assumed to be time-like) of a smooth Kantorovich potential $\phi$:
\begin{equation}\label{def_kant_pot}
\mu_{t}=(\Psi_{1/2}^{t})_{\sharp}\mu_{1/2}, \; \Psi_{1/2}^{t}(x):={\rm exp}^{g}_{x}\Big(-(t-1/2)|\nabla_{g}\phi|_{g}^{q-2} \nabla_{g}\phi(x)\Big), \; \frac{1}{p}+\frac{1}{q}=1,
\end{equation}
and moreover $\mu_{t}\ll \vol_{g}$ for all $t\in (0,1)$, where $\vol_{g}$ denotes the standard volume measure of $(M,g)$. For the precise notions, the reader is referred to Section \ref{SS:BasicsOT}.

A key role in our optimal transport formulation of the Einstein equations will be played by the (relative) Boltzmann-Shannon entropy.  Denote by $\vol_{g}$ the standard volume measure on $(M,g)$. Given an absolutely continuous probability measure $\mu=\varrho\, \vol_{g}$ with density $\varrho\in C_{c}(M)$, its  Boltzmann-Shannon entropy (relative to $\vol_{g}$) is defined as
\begin{equation}\label{eq:defEntintro}
 \Ent(\mu|\vol_{g}):=\int_{M} \varrho \log \varrho \, d\vol_{g}.
\end{equation}
We will be proving that the second order derivative of this entropy along a $c_{p}$-optimal  dynamical plan,
is equivalent to the Einstein Equation in Theorem \ref{thm:RiccieqT}.   See Figure~\ref{fig1}.
Throughout the paper we will assume the cosmological constant $\Lambda$ and the energy momentum tensor $T$ to be given, say from physics and/or mathematical general relativity.  Given $g,\Lambda$ and $T$ it is convenient to set
\begin{equation}\label{eq:deftildeTIntro}
\tilde{T}:=\frac{2\Lambda}{n-2} g + 8 \pi T - \frac{8\pi} {n-2} \Tr_{g}(T) \, g.
\end{equation}
so that the Einstein Equation can be written as $\Ric=\tilde{T}$, see Lemma~\ref{Lem-T}.

\begin{theorem}[Theorem \ref{thm:RiccieqT}]\label{thm:RiccieqTIntro}
Let $(M,g,\cC)$ be a space-time of dimension $n \geq 3$. Then the following assertions are equivalent:
\begin{enumerate}
\item[(1)]  $(M,g,\cC)$ satisfies the Einstein equations \eqref{eq:EFEintro}, which can be rewritten in terms of $\tilde{T}$
of (\ref{eq:deftildeTIntro})  as $\Ric=\tilde{T}$.
\item[(2)]  For every $p\in (0,1)$ and for every relatively  compact  open subset $E\subset\subset  \Int(\cC)$ there exist $R=R(E)\in (0,1)$ and a function $$\epsilon=\epsilon_{E}:(0,\infty)\to (0,\infty)\textrm{ with }\lim_{r\downarrow 0}\epsilon(r)=0 \textrm{ such that }
$$ 
$$
\forall x\in p_{TM\to M}(E) \textrm{ and }v\in T_{x}M\cap E \textrm{ with }g(v,v)=-R^{2}$$  the next assertion holds.
For every $r\in (0,R)$, setting $y=\exp^{g}_{x}(rv)$, there exists a regular $c_{p}$-optimal  dynamical plan $\Pi=\Pi(x,v,r)$ with associated curve of  probability measures 
$$\quad (\mu_{t}:=(\ee_{t})_{\sharp}\Pi)_{t\in [0,1]}\subset \cP(M)
$$ such that 
\begin{itemize}
\item $\mu_{1/2}=\vol_{g}(B^{g,E}_{r^{4}}(x))^{-1}\,  \vol_{g}\llcorner B^{g,E}_{r^{4}}(x)$,   
\item $\supp(\mu_{1})\subset \{\exp_{y}^{g}(r^{2}w): w\in T_{y}M\cap \cC, \, g(w,w)=-1 \}$
\end{itemize}
\noindent and which has convex/concave entropy in the following sense:
\begin{equation}\label{eq:RiceqTIntro}
\qquad \left| \tfrac{4}{r^{2}}\left[\Ent(\mu_{1}|\vol_{g})-2 \Ent(\mu_{1/2}|\vol_{g})+\Ent(\mu_{0}|\vol_{g}) \right] - \tilde{T}(v,v)\right|\le \epsilon(r)
\end{equation}

\item [(3)] There exists  $p\in (0,1)$ such that the  assertion as in  {\rm (2)}  holds true.
\end{enumerate}
\end{theorem}

\begin{figure}[ht]
	\centering
  \includegraphics[scale=0.5]{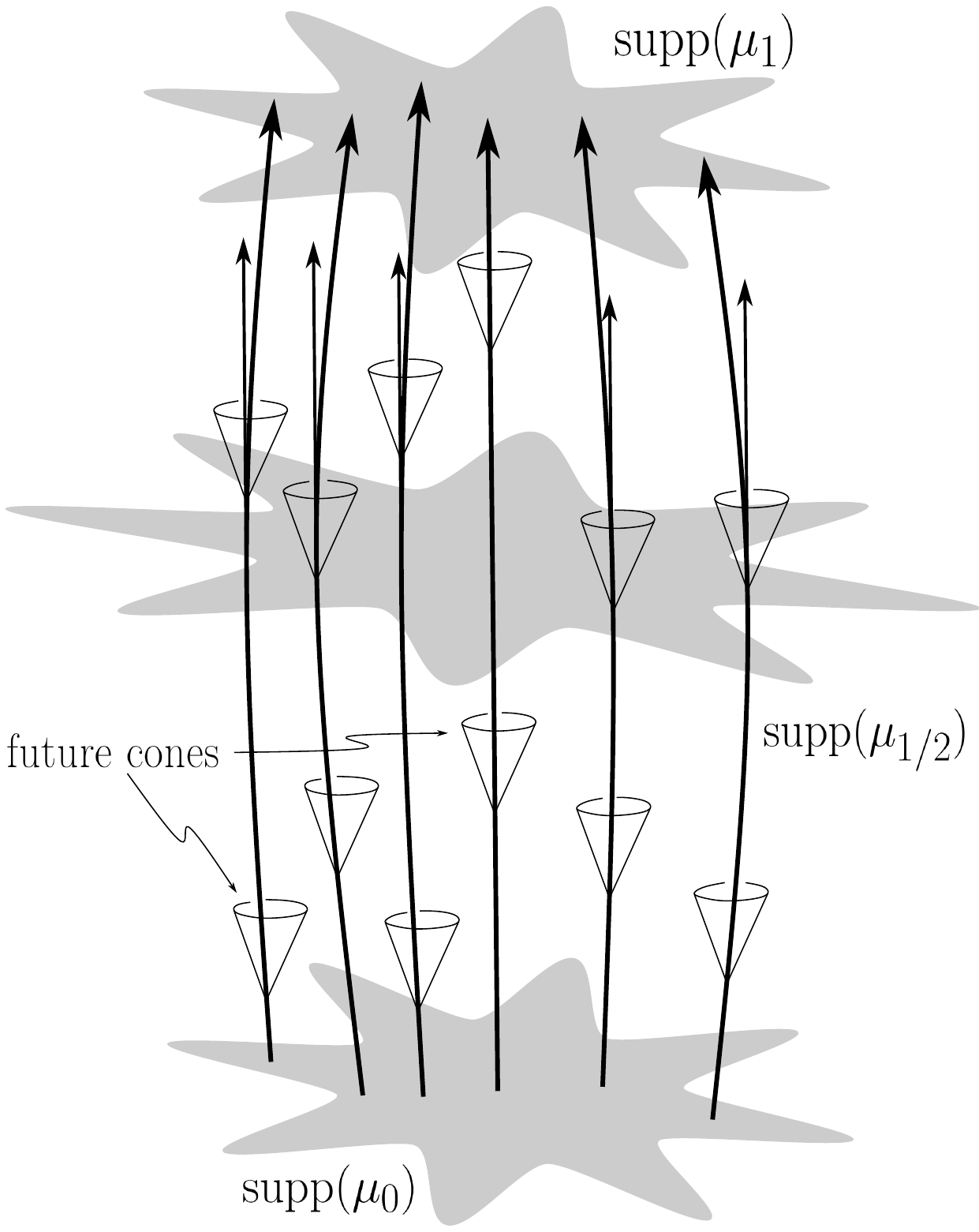}
	\caption{The transport in Theorem \ref{thm:RiccieqTIntro}}
	\label{fig1}
\end{figure}

\begin{remark}[On the regularity of the space-time]
For simplicity, in the paper we work with a smooth space-time. However, all the statements and proofs would be valid assuming that $M$ is a differentiable manifold endowed with a $C^{3}$-atlas and that  $g$ is a $C^2$-Lorentzian metric on $M$.
\end{remark}

\begin{remark}[A heuristic thermodynamic interpretation of Theorem \ref{thm:RiccieqTIntro}]
A curve $(\mu_{t})_{t\in [0,1]}\subset \cP(M)$ associated to a  $c_{p}$-optimal  dynamical plan  can be interpreted as  the evolution \footnote{strictly speaking $t$ is not the proper time, but only a variable parametrizing the evolution} of a distribution of gas passing through a given gas distribution  $\mu_{1/2}$ (that in Theorem \ref{thm:RiccieqTIntro} is assumed to be concentrated in the space-time near $x$).
 Theorem  \ref{thm:RiccieqTIntro} says that the Einstein equations can be equivalently formulated in terms of the convexity properties of the Bolzmann-Shannon entropy along such evolutions $(\mu_{t})_{t\in [0,1]}\subset \cP(M)$. 
Extrapolating a bit more,  we can say that the second law of thermodynamics (i.e. in a natural thermodynamic process, the sum of the entropies of the interacting thermodynamic systems decreases, due to our sign convention) concerns the \emph{first} derivative of the Bolzmann-Shannon entropy; gravitation (under the form of Ricci curvature) is instead related to the \emph{second} order derivative of the Bolzmann-Shannon entropy along a natural thermodynamic process.
\end{remark}

\begin{remark}[Disclaimer]
In Theorem \ref{thm:RiccieqTIntro} we are \emph{not claiming to solve} the general Einstein Equations via optimal transport; we are  instead \emph{proposing a novel formulation/characterization  of the solutions of the Einstein Equations based on optimal transport}, assuming the cosmological constant $\Lambda$ and the energy-momentum tensor $T$ being already given (this can be a bit controversial for a general $T$; however, the characterization is already new and interesting in the vacuum  case $T\equiv 0$ where there is no controversy). The aim is indeed to  bridge optimal transport and general relativity, with the goal of  stimulating fruitful connections between these two fascinating fields. In particular, optimal transport tools have been very successful to study Ricci curvature bounds in a  (low regularity) \emph{Riemannian and metric-measure} framework (see later in the introduction for the related literature) and it is thus natural to expect that optimal transport can be useful also in a low-regularity Lorentzian framework, where singularities play an important part in the theory; for example it is expected that, at least generically, singularities occur in black-hole interiors.
\end{remark}

For equivalent formulations of Theorem \ref{thm:RiccieqTIntro},  see Remark \ref{rem:Thmmu12} and Remark  \ref{rem:ThmRiemMeth}.

In the vacuum case $T\equiv 0$ with zero cosmological constant $\Lambda=0$, the Einstein equations read as
\begin{equation}\label{eq:EELambda0Intro}
\Ric \equiv 0,
\end{equation}
 for an $n$-dimensional space-time $(M,g,\cC)$. Specializing Theorem \ref{thm:RiccieqTIntro} with the choice $\tilde{T}=0$ (plus a small extra observation to sharpen the lower bound in \eqref{eq:Ric0OTIntro} from $-\epsilon(r)$ to $0$; moreover the same proof extends to $n=2$) gives the following optimal transport formulation of Einstein vacuum equations with zero cosmological constant.
 
\begin{corollary}\label{cor:Ricciflatintro}
Let $(M,g,\cC)$ be a space-time of dimension $n \geq 2$. Then the following assertions are equivalent:
\begin{enumerate}
\item[(1)]  $(M,g,\cC)$ satisfies the Einstein vacuum equations with zero cosmological constant, i.e.  $\Ric\equiv 0$.  
\item[(2)]  For every $p\in (0,1)$ and for every relatively  compact  open subset $E\subset\subset  \Int(\cC)$ there exist $R=R(E)\in (0,1)$ and a function $$\epsilon=\epsilon_{E}:(0,\infty)\to (0,\infty)\textrm{ with }\lim_{r\downarrow 0}\epsilon(r)/r^2=0 \textrm{ such that }
$$ 
$$
\forall x\in p_{TM\to M}(E) \textrm{ and }v\in T_{x}M\cap E \textrm{ with }g(v,v)=-R^{2}$$  the next assertion holds.
For every $r\in (0,R)$, setting $y=\exp^{g}_{x}(rv)$, there exists a regular $c_{p}$-optimal  dynamical plan $\Pi=\Pi(x,v,r)$ with associated curve of  probability measures 
$$\quad (\mu_{t}:=(\ee_{t})_{\sharp}\Pi)_{t\in [0,1]}\subset \cP(M)
$$ such that 
\begin{itemize}
\item $\mu_{1/2}=\vol_{g}(B^{g,E}_{r^{4}}(x))^{-1}\,  \vol_{g}\llcorner B^{g,E}_{r^{4}}(x)$,   
\item $\supp(\mu_{1})\subset \{\exp_{y}^{g}(r^{2}w): w\in T_{y}M\cap \cC, \, g(w,w)=-1 \}$
\end{itemize}
\noindent and which has almost affine entropy in the sense that
\begin{equation}\label{eq:Ric0OTIntro}
 0\leq \Ent(\mu_{1}|\vol_{g})-2 \Ent(\mu_{1/2}|\vol_{g})+\Ent(\mu_{0}|\vol_{g}) \leq  \epsilon(r).
\end{equation}
\item [(3)] There exists  $p\in (0,1)$ such that the  assertion as in  {\rm (2)}  holds true.
\end{enumerate}
\end{corollary}

\subsection{Outline of the argument}
As already mentioned, the Einstein Equations can be written as $\Ric=\tilde{T}$ where $\tilde{T}$ was defined in \eqref{eq:deftildeTIntro}, see Lemma~\ref{Lem-T}.
The optimal transport formulation of the Einstein equations will consist separately  of an optimal transport characterization of the two inequalities  
\begin{equation}\label{eq:RicgeqTintro}
\Ric\geq \tilde{T} 
\end{equation}
and 
\begin{equation}\label{eq:RicleqTintro}
\Ric\leq \tilde{T}, 
\end{equation}
respectively. 
The optimal transport characterization of the lower bound \eqref{eq:RicgeqTintro} will be achieved in Theorem  \ref{thm:RiccigeqT} and consists in showing that \eqref{eq:RicgeqTintro}  is equivalent to a \emph{convexity} property of the Bolzmann-Shannon entropy along  \emph{every} regular $c_{p}$-optimal dynamical plan.  The optimal transport characterization of the upper bound \eqref{eq:RicleqTintro}  will be achieved in Theorem  \ref{thm:RiccileqT} and consists in  showing that \eqref{eq:RicleqTintro}  is equivalent to the \emph{existence} of a large family of regular $c_{p}$-optimal dynamical plans (roughly the ones given by exponentiating the gradient of a smooth Kantorovich potential with Hessian vanishing at a given point)
along which the Bolzmann-Shannon entropy satisfies the corresponding \emph{concavity} condition.

Important ingredients in the proofs will be the following. In Theorem  \ref{thm:RiccigeqT}, for proving that Ricci lower bounds imply convexity properties of the entropy, we will perform Jacobi fields computations relating the Ricci curvature with the Jacobian of the change of coordinates of the optimal transport map (see Proposition \ref{prop:VolDist}); in order to establish the converse implication we will argue by contradiction via constructing  $c_{p}$-optimal dynamical plans very localized in the space-time  (Lemma  \ref{lem:SmoothKP}).
\\In Theorem  \ref{thm:RiccigeqT} we will consider the special class of regular $c_{p}$-optimal dynamical plans constructed in Lemma  \ref{lem:SmoothKP}, roughly the ones given by exponentiating the gradient of a smooth Kantorovich potential with Hessian vanishing at a given point $x\in M$.  For proving that Ricci upper bounds imply concavity properties of the entropy, we will need to establish the Hamilton-Jacobi equation satisfied by the evolved Kantorovich potentials (Proposition \ref{thm:HopfLax}) and a non-linear Bochner formula involving the $p$-Box operator (Proposition \ref{prop:qBochner}), the Lorentzian counterpart of the $p$-Laplacian. In order to show the converse implication we will argue by contradiction using Theorem  \ref{thm:RiccigeqT}.

\subsection{An Example. FLRW Spacetimes}\label{Ex-FLRW}
We illustrate Theorem \ref{thm:RiccieqTIntro} for the class of Friedmann-Lema\^itre-Robertson-Walker spacetimes 
(short FLRW spacetimes), a group of cosmological models well known in general relativity. See \cite[Chapter 12]{oneill} for a discussion of the geometry 
in the case $n=4$. 

FLRW spacetimes are of the form 
$$(M,g)=(I\times \Sigma, -ds^2+a^2(s)\sigma),$$
where $I\subset \R$ is an interval, $a\colon I\to (0,\infty)$ is smooth, and $(\Sigma,\sigma)$ is a Riemannian manifold with constant sectional curvature 
$k\in\{-1,0,1\}$. The Ricci and scalar curvature are given by 
$$\Ric=-(n-1)\frac{\ddot a}{a}ds^2+\left[\frac{\ddot a}{a}+(n-2)\left(\frac{\dot{a}^2+k}{a^2}\right)\right]a^2\sigma$$
and 
$${\rm Scal}=2(n-1)\frac{\ddot a}{a}+(n-1)(n-2)\frac{\dot{a}^2+k}{a^2},$$
respectively, where $\dot{a}:=\frac{da}{ds}$. The stress-energy tensor is thus (assuming $\Lambda =0$)
\begin{align*}
8\pi T&=\Ric -\frac{1}{2}{\rm Scal}\, g\\
&=\frac{(n-1)(n-2)}{2}\frac{\dot{a}^2+k}{a^2}ds^2-\left[(n-2)\frac{\ddot a}{a}+\frac{(n-2)(n-3)}{2}\frac{\dot{a}^2+k}{a^2}\right]a^2\sigma.
\end{align*}

The foliation 
$$\mathfrak{O}:=\{s\mapsto (s,\underline{x})\}_{\underline{x}\in\Sigma}$$ 
is a geodesic foliation by $c_p$-minimal geodesics. The orthogonal complement $\partial_s^\perp=T\Sigma$ with respect to $g$  is integrable. 
Consider the projection 
$$S\colon M=I\times \Sigma \to I$$ 
and for $r>0$ the function 
$$\phi\colon M\to \R,\quad x\mapsto r^{p-1}S(x).$$
It is easy to see that 
$$\nabla^q_g\phi(x):=-|\nabla_{g}\phi|_{g}^{q-2} \nabla_{g}\phi(x)=r\partial_s.$$
 For the $c_p$-transform we have (see Section \ref{SS:BasicsOT})
\begin{align*}
\phi^{c_{p}}(s,\underline{y})&=\inf_{x\in M} c_{p}(x,(s,\underline{y}))-\phi(x)\\
&=\inf_{s'<s} c_{p}((s',\underline{y}),(s,\underline{y}))-\phi(s',\underline{y})\\
&=\inf_{s'<s} -\frac{1}{p}(s-s')^p-r^{p-1}s' =\frac{p-1}{p}r^p-r^{p-1}s,
\end{align*}
where the second equality follows from the fact that the geodesics in $\mathfrak{O}$ minimize $c_p$ to the level sets of $\phi$. It follows that 
$$\phi^{c_p}(s,\underline{y})+\phi(s-r,\underline{y})=-\frac{1}{p}r^p=c_p((s-r,\underline{y}),(s,\underline{y})),$$
i.e.  $\partial^{c_p} \phi(x)=\{\exp_x(\nabla^q_g\phi(x))\}$ for all $x\in M$ whenever the right hand side is well defined (see Section \ref{SS:BasicsOT} for the 
definition).

It now follows by standard transportation theory (see for instance \cite[Theorem 1.13]{AGUser}) that for a Borel probability measure $\mu_{1/2}$ on 
$I\times \Sigma$ the family $\left(\mu_t:=(\Psi^t_{1/2})_\sharp \mu_{1/2} \right)_{t\in [0,1]}$, where 
$$\Psi^t_{1/2}\colon I\times \Sigma\to I\times \Sigma,\quad (s,\underline{x})\mapsto \left(s+r\left(t-\frac{1}{2}\right),\underline{x}\right)$$
defines a $c_p$-optimal dynamical plan as long as it is defined in accordance with the notation in \eqref{def_kant_pot} and $\phi$ is a smooth 
Kantorovich potential for $(\mu_t)_{t\in [0,1]}$.

If $\mu_{1/2}\ll \vol_{g}$ with density $\rho_{1/2}\in C_{c}(M)$, we get (compare with the proof of Theorem \ref{thm:RiccigeqT})
\begin{align*}
\Ent(\mu_t|\vol_{g})&=\int_{M} \log\rho_{t}\left(y\right)\,d\mu_{t}\left(y\right)= \int_{M} \log\rho_{t}(\Psi_{1/2}^t(x))\,d\mu_{1/2}(x)\\
&= \int_{M} \log[\rho_{1/2}(x)(\Det_g (D\Psi_{1/2}^t)(x))^{-1}]\,d\mu_{1/2}(x)\\
&=\Ent(\mu_{1/2}|\vol_{g})-\int_M \log[\Det_g (D\Psi_{1/2}^t)(x))]\,d\mu_{1/2}(x).
\end{align*}
We have $\Det_g (D\Psi_{1/2}^t)((s,\underline{x})))=\frac{a^{n-1}(s+r(t-1/2))}{a^{n-1}(s)}$ and thus obtain 
\begin{equation}\label{ERic}
-\frac{d^2}{dt^2}\log[\Det_g (D\Psi_{1/2}^t)(x))]=-(n-1)\frac{\ddot{a}a-\dot a^2}{a^2}r^2= \Ric(r\partial_s,r\partial_s)+(n-1) \frac{\dot a^2}{a^2}r^2.
\end{equation}
Neglecting the term $\frac{\dot a^2}{a^2}\geq 0$ we conclude (compare with the proof of Proposition \ref{prop:VolDist})
\begin{equation}\label{FLRW:Ricbelow}
\begin{split}
\frac{d^2}{dt^2}\Ent(\mu_t|\vol_{g})&=-\frac{d^2}{dt^2}\int_M \log[\Det_g (D\Psi_{1/2}^t)(x))]\,d\mu_{1/2}(x)\\
&\ge \int_M \Ric(r\partial_s,r\partial_s)_{\Psi_{1/2}^t(x)}\,d\mu_{1/2}(x),
\end{split}
\end{equation}
which implies
$$\tfrac{4}{r^{2}}\left[\Ent(\mu_{1}|\vol_{g})-2 \Ent(\mu_{1/2}|\vol_{g})+\Ent(\mu_{0}|\vol_{g}) \right] \ge \int_M\Ric(\partial_s,\partial_s)d\mu_{1/2}$$
 for $r\to 0$, i.e. one side of \eqref{eq:RiceqTIntro}.

Note that the gap in \eqref{FLRW:Ricbelow} is
$$(n-1)r^2 \int_M \frac{\dot a^2}{a^2} \,d\mu_{1/2}.$$
From this we see that the bound from above 
\begin{equation}\label{FLRW:Ricabove}
\begin{split}
\tfrac{4}{r^{2}}[\Ent(\mu_{1}|\vol_{g})-2 \Ent(\mu_{1/2}|\vol_{g})+&\Ent(\mu_{0}|\vol_{g})] \\
&\le \int_M\Ric(\partial_s,\partial_s)d\mu_{1/2}+ \epsilon(r),
\end{split}
\end{equation}
in \eqref{eq:RiceqTIntro} for the aforementioned transports 
holds for $r\to0$ if  $\mu_{1/2}$ is concentrated on $\{(s,\underline{x})|\;\dot{a}(s)=0\}$ or, more generally, if $\mu_{1/2}=\mu_{1/2}^{r}$ satisfies  $\lim_{r\to 0} \int_M \frac{\dot a^2}{a^2} \,d\mu^{r}_{1/2} = 0$.

The Levi-Civita connection $\nabla$ of $g$ satisfies 
$$\nabla_X \partial_s=\nabla_{\partial_s} X=\frac{\dot a}{a}X$$
for all vector fields $X$ tangent to $\Sigma$. Thus the Hessian of $\phi$ is given by
\begin{equation}\label{eq:Hessphia2}
\Hess_\phi=r^{p-1}\Hess_S=r^{p-1}\nabla_. \nabla_g S=-r^{p-1}\nabla_. \partial_s=-r^{p-1}\frac{\dot{a}}{a}(\id- \partial_s\otimes ds).
\end{equation}
It vanishes at $(s,\underline{x})$ if and only if $\dot{a}(s)=0$. Thus the inequality \eqref{FLRW:Ricabove} follows for these transports if the Hessian 
of $\phi$ vanishes  on $\supp (\mu_{1/2})$ or, more generally, if   $ \lim_{r\to 0} \int_{M}  \|  r^{1-p} \Hess_{\phi} \|^{2} d\mu^{r}_{1/2}=0$.

Finally we illustrate the theory laid out in Appendix \ref{AppB} via warping functions $a\colon I\to (0,\infty)$ of regularity below $C^2$. Compare with
\cite{Graf20} for related results on metrics of low regularity. 

Consider a probability measure $\nu_{1/2}\ll \vol^\sigma$ on $\Sigma$ and set
\[
\mu_{1/2}:=\frac{1}{s_1-s_0}\mathcal{L}^1|_{[s_0,s_1]}\otimes \nu_{1/2}
\]
for $s_0,s_1\in I$ with $s_0<s_1$ and $\mathcal{L}^1$ the $1$-dimensional Lebesgue measure. For $\mu_t:=(\Psi_{1/2}^t)_\sharp \mu_{1/2}$ we have
\begin{align*}
\frac{4}{r^2}[\Ent(\mu_{1}|\vol_{g})-2 \Ent(\mu_{1/2}|\vol_{g})+&\Ent(\mu_{0}|\vol_{g})]\\
=-\frac{4}{r^2(s_1-s_0)}&\int_{s_0}^{s_1}\log\left(\frac{a^{n-1}(s+\frac{r}{2})a^{n-1}(s-\frac{r}{2})}{ a^{2n-2}(s)}\right)ds
\end{align*}
Assuming $s_1-s_0\ll r$ and $r>0$ sufficiently small (i.e. the setting of Theorem \ref{thm:RiccieqTIntro}),  since $a$ is continuous,  we  obtain:
\begin{align*}
\frac{4}{r^2}[\Ent(\mu_{1}|\vol_{g})-2 \Ent(\mu_{1/2}|\vol_{g})+&\Ent(\mu_{0}|\vol_{g})]\\
=-\frac{4}{r^2}&\log\left(\frac{a^{n-1}(s+\frac{r}{2})a^{n-1}(s-\frac{r}{2})}{a^{2n-2}(s)}\right)+ \frac{\ve(r)}{r^2},
\end{align*}
 for some function $\ve(r)\to 0$ as $r\to 0$.
As an example, we discuss the case of the following $C^{1,1}$-warping function
\[
a\colon I\to \R,\quad a(s):=\begin{cases} \lambda_+ s^2+1,&\; s\ge 0\\ \lambda_- s^2+1,&\; s<0
\end{cases}
\]
for $\lambda_-, \lambda_+\in \R$.

Ignoring terms of higher order near $s=0$ we get
\begin{align*}
\frac{4}{r^2}[\Ent(\mu_{1}|\vol_{g})&-2 \Ent(\mu_{1/2}|\vol_{g})+\Ent(\mu_{0}|\vol_{g})]\\
=&\begin{cases}
-(n-1)[\lambda_+(\frac{2s}{r}+1)^2+\lambda_-(\frac{2s}{r}-1)^2-2\lambda_+(\frac{2s}{r})^2],& \; s\ge 0\\
-(n-1)[\lambda_+(\frac{2s}{r}+1)^2+\lambda_-(\frac{2s}{r}-1)^2-2\lambda_-(\frac{2s}{r})^2],& \; s\le 0
\end{cases}
\end{align*}
for $r\to 0$. It is now easy to see that the  possible accumulation points of the right hand side for $r\to 0$ and $s\in [-\frac{r}{2},\frac{r}{2}]$ lie between $-2(n-1)\lambda_+=\lim_{s\downarrow0}
\Ric(\partial_s,\partial_s)$ and $-2(n-1)\lambda_-=\lim_{s\uparrow0}\Ric(\partial_s,\partial_s)$. Furthermore, every value in that interval is  an accumulation point for the right hand side, as $r\to 0$.
In case the warping function lies in $C^2$, we get from \eqref{ERic} the asymptotic formula:
\begin{equation}\label{SynthRic}
\frac{4}{r^2}[\Ent(\mu_{1}|\vol_{g})-2 \Ent(\mu_{1/2}|\vol_{g})+\Ent(\mu_{0}|\vol_{g})]- (n-1) \int_M \frac{\dot a^2}{a^2} \,d\mu_{1/2}
\end{equation}
for $\Ric(\partial_s,\partial_s)$ when $r\to 0$. The second integral is well defined for $a\in W^{1,2}(I)$.  Therefore we can use \eqref{SynthRic} as a definition for $\Ric(\partial_s,\partial_s)$ in this case.
 This is indeed the spirit of the synthetic
bounds on the timelike Ricci curvature given in Definition \ref{def:RUB} and Definition \ref{def:TCD(KN)}.

\subsection{Related literature} 

\subsubsection{Ricci curvature via optimal transport in  Riemannian setting}
In the Riemannian framework, a  line of
research pioneered by McCann \cite{68MC}, Cordero-Erausquin-McCann-Schmuckenschl\"ager \cite{CMS,CMS2}, Otto-Villani \cite{OV} and von Renesse-Sturm, has culminated in a characterization
of Ricci-curvature lower bounds (by a constant $K\in \R$) involving only the displacement convexity of
certain information-theoretic entropies \cite{SVR}. This in turn led Sturm \cite{sturm:I, sturm:II} and
independently Lott-Villani \cite{lottvillani:metric} to develop a theory for
lower Ricci curvature bounds in a non-smooth metric-measure space setting. The theory of such spaces has seen a very fast development in the last years, see e.g.  \cite{AGMR12,  AGS, AGS11a, AMS, BS10, BS18, CaMi, CaMo, EKS, gigli:laplacian,GMS2013,MN}. An approach to the complementary upper
bounds on the Ricci tensor (again by a constant $K'\in \R$) has been recently proposed by Naber \cite{69NabUBR} (see also Haslhofer-Naber \cite{48HN}) in  terms of functional inequalities on  path spaces and martingales, and by Sturm \cite{StUB} (see also Erbar-Sturm \cite{ErSt18}) in terms of contraction/expansion rate estimates of the heat flow and in terms of displacement \emph{concavity} of the Shannon-Bolzmann entropy.  The Lorentzian time-like Ricci upper bounds of this paper have been inspired in particular by the work of Sturm  \cite{StUB}.

\subsubsection{Optimal transport in  Lorentzian setting}
The optimal transport problem in Lorentzian geometry was first proposed by Brenier \cite{Br03} and further investigated  in \cite{BP13, suhr, kellsuhr}.  An intriguing physical motivation for studying the optimal transport problem in Lorentzian setting called the ``early universe reconstruction problem'' \cite{BFHLMMS03,FMMS02}.  The Lorentzian  cost $C_{p}$, for $p\in (0,1)$, was  proposed by Eckstein-Miller \cite{EM17} and thoroughly studied by Mc Cann \cite{McCann18} very recently. In the same paper  \cite{McCann18}, Mc Cann gave an optimal transport formulation of the strong energy condition $\Ric\geq 0$ of Penrose-Hawking \cite{Pen,Haw66, HawPen70} in terms of displacement convexity of the Shannon-Bolzmann entropy under the assumption that the space time is globally hyperbolic. 

We learned of  the work of Mc Cann \cite{McCann18} when we were already in the final stages of writing the present paper. Though both papers (inspired by the aforementioned Riemannian setting) are based on the idea of analyzing convexity properties of entropy  functionals on the space of probability measures endowed with the cost $C_{p}$, $p\in (0,1)$, the two approaches are largely independent: while Mc Cann develops a general theory of optimal transportation in globally hyperbolic space times focusing on the strong energy condition $\Ric\geq 0$,  in this paper we decided to take the quickest path in order to reach our goal of giving an optimal transport formulation of the full Einstein's equations. Compared to \cite{McCann18},  in the present paper we remove the assumption of global hyperbolicity on the space-time, we extend the optimal transport formulation to any lower bound of the type $\Ric\geq \tilde{T}$ for any symmetric bilinear form $\tilde{T}$, and we also characterize general upper bounds $\Ric\leq \tilde{T}$.

\subsubsection{Physics literature}
The existence of strong connections between thermodynamics and general relativity is not new in the physics literature; it has its origins at least in the work Bekenstein \cite{Bek} and Hawking with collaborators \cite{Haw} in the mid-1970s  about the black hole thermodynamics.  These works inspired a new research field in theoretical physics, called entropic gravity (also known as emergent gravity), asserting  that gravity is an entropic force rather than a fundamental interaction. 
Let us give a  brief account.  In 1995  Jacobson \cite{Jac} derived the Einstein equations from the proportionality of entropy and horizon area of a black hole, exploiting the fundamental relation $\delta Q=T \, \delta S$  linking heat $Q$, temperature $T$ and entropy $S$.  Subsequently, other physicists, most notably Padmanabhan (see for instance the recent survey \cite{Pad}), have been exploring links between gravity and entropy.

More recently, in 2011 Verlinde \cite{Ver} proposed a heuristic argument suggesting that (Newtonian) gravity can be identified with an entropic force caused by changes in the information associated with the positions of material bodies. A
relativistic generalization of those arguments  leads to the Einstein equations.

The optimal transport formulation of Einstein equations obtained in the present paper involving the Shannon-Bolzmann entropy can be seen as an additional strong connection between  general relativity and thermodynamics/information theory. It would be interesting to explore this relationship further.

\subsection*{Acknowledgement}
The authors wish to thank Christina Sormani and the anonymous referee for several comments that  improved the exposition of the paper.

\section{Preliminaries}
\subsection{Some basics of Lorentzian geometry}
Let $M$ be a smooth manifold of dimension $n\geq 2$. It is convenient to fix a complete Riemannian metric $h$ on $M$. The norm $|\cdot|$ on $T_{x}M$ and the distance $\dist(\cdot,\cdot):M\times M\to \R^{+}$ are understood to be induced by $h$, unless otherwise specified. Recall that $h$ induces a Riemannian metric on $TM$. Distances on $TM$ are understood to the induced by such a metric. The metric ball around $x\in M$ with radius $r$, with respect to $h$, is denoted by $B_{r}^{h}(x)$ or simply by $B_{r}(x)$.

A \emph{Lorentzian metric} $g$ on $M$ is a smooth  $(0,2)$-tensor field such that
$$g|_{x}:T_{x}M\times T_{x}M \to \R$$
is symmetric and non-degenerate with signature $(-,+,\dots,+)$ for all $x\in M$.
It is well known that, if $M$ is compact,  the vanishing of the Euler characteristic of $M$ is equivalent to the existence of a Lorentzian metric;
on the other hand, any \emph{non-compact} manifold admits a Lorentzian metric.
 A non-zero tangent vector $v\in T_{x}M$ is called 
 \begin{itemize}
\item \emph{Time-like}: if $g(v,v)<0$,
\item \emph{Light-like} (or \emph{null}): if $g(v,v)=0$ as well as $v\neq 0$,
\item \emph{Spacelike}: if $g(v,v)>0$ or $v=0$. 
\end{itemize}
A non-zero tangent vector $v\in T_{x}M$ which is either time-like or light-like, i.e.  $g(v,v)\leq 0$ and $v\neq 0$, is called \emph{causal} (or \emph{non-spacelike}).
A Lorentzian manifold $(M,g)$ is said to be \emph{time-oriented} if $M$ admits a continuous no-where vanishing time-like vector field $X$. The vector field $X$ induces a partition on the set of causal vectors, into two equivalence classes: 
\begin{itemize}
\item The \emph{future pointing} tangent vectors  $v$ for which $g(X,v)< 0$,
\item The  \emph{past pointing}  tangent vectors  $v$ for which $g(X,v)> 0$. 
\end{itemize}
The closure of the set of future pointing time-like vectors is denoted 
$$
\cC={\rm Cl}(\{ v \in TM:\,\, g(v,v)<0 \textrm{ and } g(X,v)<0\})\subset TM.
$$
 Note that the fiber  $\cC_{x}:=\cC \cap  T_{x}M$ is a closed convex cone and the open interior $\Int(\cC)$ is a connected component of $\{v:g(v,v)<0\}$.
A time-oriented Lorentzian manifold $(M,g,\cC)$ is called a \emph{space-time}.

 An absolutely continuous curve $\gamma:I\to M$ is called ($\cC$)-\emph{causal} if  $\dot{\gamma}_{t}\in \cC$ for every differentiability point $t\in I$. A causal curve $\gamma:I\to M$ is called  \emph{time-like} if  for every $s\in I$ there exist $\ve, \delta>0$ such that $\dist(\dot\gamma_{t},\partial \cC)\geq \ve |\dot\gamma_{t}|$ for every $t\in I$ for which $\dot\gamma_{t}$ exists and $|s-t|<\delta$.  In \cite[Section 2.2]{BS1} time-like curves are defined in terms of the Clarke differential of a 
 Lipschitz curve. Whereas the definition via the Clarke differential is probably more satisfying from a conceptual point of view, the definition given here is easier to state. 
 All relevant sets and curves used below are independent of the definition, see \cite[Lemma 2.11]{BS1} and Proposition \ref{Propminimizer}, though.

We denote by $J^{+}(x)$ (resp. $J^{-}(x)$) the set of points $y\in M$ such that there exists a causal curve with initial point $x$ (resp. $y$) and final point  $y$ (resp. $x$), i.e. the causal future (resp. past) of $x$. 
The sets $I^{\pm}(x)$ are defined  analogously by replacing causal curves by time-like ones. The sets $I^{\pm}(p)$ are always open in any space-time, on the other hand the sets $J^{\pm}(p)$ are in general neither closed nor open.

For a subset $A\subset M$, define $J^{\pm}(A):=\cup_{x\in A} J^{\pm}(x)$, moreover set
\begin{equation}
J^{+}:=\{(x,y)\in M\times M\,:\, y\in J^{+}(x)\}.
\end{equation}

\subsection{The Lagrangian $\cL_{p}$, the action   $\cA_{p}$ and the cost $c_{p}$}
On a space-time  $(M,g,\cC)$ consider, for any $p\in (0,1)$, the following Lagrangian on $TM$:
\begin{equation}\label{def:pLangrangianLp}
\cL_{p}(v):=
\begin{cases}
-\frac{1}{p}(-g(v,v))^{\frac{p}{2}} & \text{ if $v\in \cC$}\\
+\infty &  \text{otherwise.}
\end{cases}
\end{equation}

The following fact appears in \cite[Lemma 3.1]{McCann18}. We provide a proof for the readers convenience. 
\begin{lemma}\label{lem:Lpconvex}
The function $\cL_{p}$ is fiberwise convex, finite (and non-positive)
on its domain and positive homogenous of degree $p$. Moreover $\cL_{p}$ is smooth and fiberwise strictly convex on $\Int(\cC)$.
\end{lemma}

\begin{proof}
It is clear from its very definition that the restriction of $\cL_{p}$ to $\Int(\cC)$ is smooth. A direct computation gives
\begin{align}
\frac{\partial \cL_{p}}{\partial v^{i}} &= (-g(v,v))^{\frac{p-2}{2}} g_{ik} v^{k}, \quad i=1,\dots,n \label{eq:dldvi}\\
\frac{\partial^{2} \cL_{p}}{\partial v^{i} \partial v^{j}} &= (-g(v,v))^{\frac{p-4}{2}} \left(-g(v,v) g_{ij} + (2-p) g_{ik} v^{k} g_{jl} v^{l} \right), \quad i,j=1,\dots,n  \label{eq:d2ldvidvj}.
\end{align}
Fix $v\in \Int(\cC)$. Decompose $w\in T_xM$ into $w^\parallel$ the part parallel to $v$ and $w^\perp$ the part orthogonal to $v$, all with respect to $g$. 
Then we have
\begin{align}
D^{2}_{vv} \cL_{p}(w,w)&=  (-g(v,v))^{\frac{p-4}{2}}  \Big( -g(w^\perp,w^\perp) g(v,v) -g(w^\parallel,w^\parallel) g(v,v) \\
& \qquad \qquad \qquad \qquad \quad  + (2-p) g(v,w^\parallel)^{2} \Big)\\
&=  (-g(v,v))^{\frac{p-4}{2}}  \left( -g(w^\perp,w^\perp) g(v,v) + (1-p) g(v,w^\parallel)^{2} \right).\label{eq:d2cLpww}
\end{align}
Since $g(w^\perp,w^\perp)\ge 0$ and $p<1$ we have 
\begin{align*}
D^{2}_{vv} \cL_{p}(w,w)>0
\end{align*} 
for $w\neq 0$.
\end{proof}

We define the Lagrangian action $\cA_{p}$ associated to $\cL_{p}$ as follows:
\begin{equation}\label{def:pActionAp}
\cA_{p}(\gamma):=\int_{0}^{1} \cL_{p}(\dot{\gamma}_{t}) dt \in (-\infty, 0]\cup\{+\infty\}.
\end{equation}
Note that if $\cA_p(\gamma)\in \R$, then $\gamma$ is causal.
A causal curve $\gamma:[0,1]\to M$ is an \emph{$\cA_{p}$-minimizer} between its endpoints $x,y\in M$ if
$$
\cA_{p}(\gamma)=\inf\{\cA_{p}(\eta)\,:\, \eta\in {\rm AC}([0,1],M), \eta_{0}=x, \eta_{1}=y\}.
$$ 

\begin{lemma}\label{lemma_minimizer}
Any $\cA_p$-minimizer with finite action is either a future pointing time-like geodesic of $(M,g)$ or a future pointing light-like pregeodesic of $(M,g)$, i.e. an orientation preserving 
reparameterization is a future pointing light-like geodesic of $(M,g)$.
\end{lemma}

\begin{proof}
Let $\gamma\colon [0,1]\to M$ be a $\cA_p$-minimizer with finite action. Then $\dot\gamma(t)\in \cC$ for a.e. $t$. By Jensen's inequality we have
$$\int_0^1 -\frac{1}{p}(-g(\dot\eta,\dot\eta))^{\frac{p}{2}}dt\ge -\frac{1}{p}\left(\int_0^1 \sqrt{-g(\dot\eta,\dot\eta)} dt\right)^p,$$
for any causal curve $\eta\colon [0,1]\to M$ with equality if and only if $\eta$ is parametrized proportionally to arclength.

Recall that the restriction of a minimizer to any subinterval of $[0,1]$ is a minimizer of the restricted action. 
Since any point in a spacetime admits a globally hyperbolic neighborhood, see \cite[Theorem 2.14]{minsan}, the Avez-Seifert Theorem \cite[Proposition 14.19]{oneill} implies that 
every minimizer of $\cA_1$ with finite action is a causal pregeodesic. 

Combining both points we see that if the action of $\gamma$ is negative, the curve is a time-like pregeodesic parameterized with respect to constant arclength, i.e. a time-like geodesic. If the action of $\gamma$ vanishes, the curve is a light-like pregeodesic.
\end{proof}

Consider the cost function
relative to the $p$-action $\cA_{p}$:
\begin{align*}
c_{p}:& \;M\times M\to \R\cup\{+\infty\} \\
          & \;(x,y)\mapsto \inf \{\cA_{p}(\eta)\,:\, \eta\in {\rm AC}([0,1], M), \eta_{0}=x, \eta_{1}=y\}.
\end{align*}

\begin{remark}
We will always assume that:
\begin{itemize}
\item[(i)] The cost function is bounded from below on bounded subsets of $M\times M$. By transitivity of the causal relation this follows from the assumption that
$c_p(x,y)>-\infty$ for all $x,y\in M$. 
\item[(ii)] 
The cost function is localizable, i.e. every point $x\in M$ has a neighborhood $U\subset M$ such that the cost function of the space-time $(U,g|_U,\cC|_U)$ 
coincides with the global cost function.
\end{itemize}
Note since the main results of this paper are local in nature,  the assumptions can always be satisfied by restricting the space-time to a suitable open subset. 
\end{remark}

\begin{proposition}\label{Propminimizer} Fix $p\in (0,1)$ and let $(M,g,\cC)$ be a space-time. Then every point has a neighborhood $U$ such that the following holds for the space-time $(U,g|_U,\cC|_U)$. 
For every pair of points $x,y\in U$ with $(x,y)\in J^{+}_U$, the causal relation of $(U,g|_U,\cC|_U)$,  there exists a curve $\gamma:[0,1]\to U$ with $\gamma_{0}=x$, 
$\gamma_{1}=y$, and minimizing $\cA_p$ among all curves $\eta\in {\rm AC}([0,1], M)$ with $\eta_{0}=x$ and $\eta_{1}=y$. Moreover $\gamma$ is a constant speed 
geodesic for the metric $g$, $\dot{\gamma}\in \cC$ whenever the tangent vector exists, and $\cA_p(\gamma)\in\R$. 
\end{proposition}

\begin{proof}
It is well known that in a space-time every point has a globally hyperbolic neighborhood. Let $U$ be such a neighborhood. 
If $(x,y)\in J^+_U$ there exists a curve with finite action $\mathcal{A}_p$ between $x$ and $y$. At the same time the action is bounded from below, e.g. by a 
steep Lyapunov function, see \cite{BS1}.
Therefore any minimizer $\gamma\colon [0,1]\to U$ has finite action, i.e. $\dot\gamma(t)\in \cC$ for almost all $t$. 
By Jensen's inequality we have
$$\int_0^1 -\frac{1}{p}(-g(\dot\eta,\dot\eta))^{\frac{p}{2}}dt\ge -\frac{1}{p}\left(\int_0^1 \sqrt{-g(\dot\eta,\dot\eta)} dt\right)^p,$$
for any causal curve $\eta\colon [0,1]\to U$ with equality if and only if $\eta$ is parametrized proportionally to arclength. By the Avez-Seifert Theorem \cite[Proposition 14.19]{oneill} every minimizer of the right hand side is a causal 
pregeodesic. Combining both it follows that every $\mathcal{A}_p$-minimizer is a causal geodesic. 
\end{proof}

\subsection{Ricci curvature and Jacobi equation}

We now fix the notation regarding  curvature for a Lorentzian manifold $(M,g)$ of dimension $n\geq 2$.
Called  $\nabla$ the Levi-Civita connection of  $(M,g)$, the Riemann curvature tensor is defined by
\begin{equation}\label{def:R}
R(X,Y)Z=\nabla_{X} \nabla_{Y} Z- \nabla_{Y} \nabla_{X} X-\nabla_{[X,Y]}Z,
\end{equation}
where $X,Y,Z$ are smooth vector fields on $M$ and $[X,Y]$ is the Lie bracket of $X$ and $Y$.
\\For each $x\in M$, the Ricci curvature is a symmetric bilinear form $\Ric_{x}: T_{x}M\times T_{x}M\to \R$ defined by
\begin{equation}\label{def:Ricci}
\Ric_{x}(v,w):= \sum_{i=1}^{n} g(e_{i}, e_{i}) g(R(e_{i}, w) v, e_{i}),
\end{equation}
where $\{e_{i}\}_{i=1,\dots,n}$ is an orthonormal basis of $T_{x}M$, i.e. $|g(e_{i}, e_{j})|=\delta_{ij}$ for all  $i,j=1\ldots, n$.

Given a endomorphism $\cU:T_{x}M\to T_{x}M$ and a $g$-orthonormal basis $\{e_{i}\}_{i=1,\dots,n}$ of $T_{x}M$, we associate to $\cU$ the matrix 
\begin{equation}\label{eq:defmatrixU}
(\cU_{ij})_{i,j=1,\dots,n}, \quad \cU_{ij}:=g(e_{i}, e_{j}) \, g(\cU e_{i}, e_{j}).
\end{equation}
The trace  $\Tr_{g} (\cU)$ and the determinant $\Det_{g} (\cU)$ 
of the endomorphism $\cU$ with respect to the Lorentzian metric $g$ are  
by definition the trace $\tr(\cU_{ij})$ and the determinant $\det(\cU_{ij}))$
of the matrix $(\cU_{ij})_{i,j=1,\dots,n}$, respectively. It is standard to check that such a definition is independent of the chosen orthonormal basis of $T_{x}M$.
Note that $\Ric_{x}(v,w)$ is the trace of curvature endomorphism $R(\cdot, w) v:T_{x}M\to T_{x}M$.

A smooth curve $\gamma:I\to M$ is called a \emph{geodesic} if $\nabla_{\dot{\gamma}} \dot{\gamma}=0$.  A vector field $J$ along a geodesic $\gamma$ is said to be a \emph{Jacobi field} if it satisfies the \emph{Jacobi equation}:
\begin{equation}\label{eq:JacobiEq}
\nabla_{\dot{\gamma}}( \nabla_{\dot{\gamma}} J)+R(J, \dot{\gamma})\dot{\gamma}=0.
\end{equation}

\subsection{The $q$-gradient of a function}

Finally let us recall the definition of gradient and hessian. Given a smooth function $f:M\to\R$, the \emph{gradient of $f$} denoted by $\nabla_{g} f$ is defined by the identity
$$
g( \nabla_{g} f, Y)=df(Y), \quad \forall Y\in TM,
$$
where $df$ is the differential of $f$. The \emph{Hessian of $f$}, denoted by $\Hess_{f}$ is defined to be the covariant derivative of $df$:
$$
\Hess_f:=\nabla(df).
$$  
It is related to the gradient through the formula
$$
\Hess_f(X,Y)=g(\nabla_{X} \nabla_{g}f, Y),   \quad \forall X,Y\in TM,
$$
and satisfies the symmetry 
\begin{equation}\label{eq:HessSym}
\Hess_f(X,Y)=\Hess_f(Y,X) \quad \forall X,Y\in TM.
\end{equation}

Next we recall some notions for the causal character of functions. 
\begin{itemize}
\item A function $f\colon M\to \R\cup \{\pm\infty\}$ is a \emph{causal function} if $f(x)\le f(y)$ for all $(x,y)\in J^+$; 
\item it is a \emph{time function} if $f(x)<f(y)$ for all
$(x,y)\in J^+\setminus \Delta$, where $\Delta$ denotes the diagonal in $M\times M$. 
\item Following \cite{BS1} we call a differentiable ($C^k$-) function $f\colon M\to \R$ $(k\in \N\cup\{\infty\})$ a \emph{($C^k$-) Lyapunov} or 
\emph{($C^k$-) temporal} function if $df_x|_{\cC_x\setminus \{0\}}>0$ for all $x\in M$.
\end{itemize}
Let $q$ be the conjugate exponent  to $p$, i.e. 
$$\frac{1}{p}+\frac{1}{q}=1, \; \text{ or equivalently } (p-1)(q-1)=1.$$ 
Notice that, since  $p$ ranges in $(0,1)$ then $q$ ranges in $(-\infty, 0)$. In order to describe the optimal transport maps later in the paper, it is  useful to introduce the $q$-gradient (cf. \cite{Kell}) 
\begin{equation}\label{def:nablaqphi}
\nabla^{q}_{g} \phi:= -|g(\nabla_{g} \phi, \nabla_{g} \phi)|^{\frac{q-2}{2}} \nabla_{g} \phi
\end{equation}
for differentiable Lyapunov functions $\phi\colon M\to\R$; in particular,  $\nabla^{q}_{g} \phi(x)\in \cC_{x}\setminus\{0\}$. 
Notice that, 
$$\text{For } v\in \cC_{x}\setminus\{0\}, \; \nabla_{g} \phi(x)= -|g(v,v)|^{\frac{p-2}{2}} v \text{ if and only if } \nabla^{q}_{g} \phi(x)=v.$$ 
Moreover 
$$x\mapsto \nabla^{q}_{g}\phi(x)\text{ is continuous (resp.  $C^{k}$, $k\geq 1$) on }  U\subset \{|g(\nabla^{q}_{g} \phi, \nabla^{q}_{g} \phi) |>0\}$$
 if and only if
 $$x\mapsto \nabla _{g}\phi(x)  \text{  is continuous (resp.  $C^{k}$, $k\geq 1$) on  } U\subset \{|g(\nabla_{g} \phi, \nabla_{g} \phi)| >0\}. $$

The motivation for the use of the $q$-gradient comes from the Hamiltonian formulation of the dynamics; let us briefly mention  a few key facts that will play a role later in the paper.
For $\alpha\in T_{x}^{*}M$, let 
\begin{equation}\label{eq:defHp}
\cH_{p}(\alpha)=\sup_{v\in T_{x}M}\left[ \alpha(v) - \cL_{p}(v) \right]
\end{equation}
be the Legendre transform of $\cL_p$. Denote with $g^*$ the dual Lorentzian metric on $T^*M$ and $\cC^*\subset T^*M$ the dual cone field to $\cC$. Then $\cH_p$ satisfies 
\begin{equation}\label{eq:ReprcHp}
\cH_{p}(\alpha):=
\begin{cases}
-\frac{1}{q}(-g^*(\alpha,\alpha))^{\frac{q}{2}} & \text{ if $\alpha\in \cC^*\setminus T^{*,0}M$}\\
+\infty &  \text{otherwise}
\end{cases},
\end{equation}
for $(p-1)(q-1)=1$. By analogous computations as performed in the proof of Lemma \ref{lem:Lpconvex}, one can check that
\begin{equation}\label{eq:nablagqphiDHpdphi}
\nabla_{g}^{q}\phi(x)=D\cH_{p}(-d\phi(x)).
\end{equation}
By well known properties of the Legendre transform (see for instance \cite[Theorem A.2.5]{CS}) it follows that $D\cH_{p}$ is invertible  on $\Int (\cC^{*})$  with inverse given by $D\cL_{p}$.
Thus \eqref{eq:nablagqphiDHpdphi} is equivalent to
\begin{equation}\label{eq:DLpnablaphidphi}
D\cL_{p}(\nabla_{g}^{q}\phi(x))=-d\phi(x).
\end{equation}

\subsection{$c_{p}$-concave functions and regular  $c_{p}$-optimal  dynamical plans}\label{SS:BasicsOT}
We denote by ${\cP}(M)$ the set of Borel probability measures on $M$. For any $\mu_1,\mu_{2}\in \cP(M)$, we say that  a Borel probability measure 
$$\pi\in \cP(M\times M) \textrm{ is a coupling of }\mu_{1} \textrm{ and }\mu_{2}
$$
 if $(p_{i})_{\sharp}\pi=\mu_{i}, i=1,2$, where $p_{1}, p_{2}:M\times M\to M$ are the projections onto the first and second coordinate. Recall that the push-forward $(p_{1})_{\sharp} \pi$ is defined by  
$$(p_{1})_{\sharp} \pi (A):=\pi \big(p_{1}^{-1}(A)\big)$$ for any Borel subset $A\subset M$. The set of couplings of $\mu_{1},\mu_{2}$ is denoted by $\Cpl(\mu_{1},\mu_{2})$. The \emph{$c_{p}$-cost of a coupling} $\pi$ is  given by 
$$
\int_{M\times M} c_{p}(x,y) d\pi(x,y) \in  [-\infty, 0]\cup \{+\infty\}.
$$
Denote by $C_{p}(\mu_{1}, \mu_{2})$ the \emph{minimal cost relative to $c_{p}$} among all couplings from $\mu_{1}$ to $\mu_{2}$, i.e.
$$
C_{p}(\mu_{1}, \mu_{2}):=\inf \left\{\int c_{p} d\pi \,:\, \pi \in \Cpl(\mu_{1},\mu_{2}) \right\}\in  [-\infty, 0]\cup \{+\infty\}.
$$
If $C_{p}(\mu_{1}, \mu_{2})\in\R$, a coupling achieving the infimum  is said to be $c_{p}$-\emph{optimal}. 

We next define the notion of \emph{$c_{p}$-optimal dynamical plan}. To this aim, it is convenient to consider the set of $\cA_{p}$-minimizing curves,  denoted by  $\Gamma_{p}$. The set $\Gamma_{p}$ is endowed with the $\sup$ metric induced by the auxiliary Riemannian metric $h$. It will be useful to consider the maps for $t\in [0,1]$:
\begin{align*}
\ee_{t}: \Gamma_{p}\to M,& \quad \ee_{t}(\gamma):=\gamma_{t} \\
\partial \ee_{t}: \Gamma_{p}\to TM,& \quad \partial \ee_{t}(\gamma):=\dot\gamma_{t}\in T_{\gamma_{t}}M.
\end{align*}

A \emph{ $c_{p}$-optimal dynamical plan} is a probability measure $\Pi$ on $\Gamma_{p}$ such that $(\ee_{0}, \ee_{1})_{\sharp} \Pi$ is a  $c_{p}$-optimal   coupling from  $\mu_{0}:=(\ee_{0})_{\sharp}\Pi$ to  $\mu_{1}:=(\ee_{1})_{\sharp}\Pi$.
\\ We will mostly be interested in $c_{p}$-optimal dynamical plans obtained by ``exponentiating the $q$-gradient of a $c_{p}$-concave function'', what we will call \emph{regular $c_{p}$-optimal dynamical plans}. In order to define them precisely, let 
 us first recall some basics of Kantorovich duality (we adopt the convention of \cite{AGUser}). 
 
 Fix two subsets $X,Y\subset M$. A function $\phi:X\to \R\cup\{-\infty\}$ is said \emph{$c_{p}$-concave} (with respect to $(X,Y)$) if it is not identically $-\infty$ and there exists $u:Y\to \R\cup \{- \infty\}$ such that
$$
\phi(x)=\inf_{y\in Y}c_{p}(x,y)-u(y), \quad \text{ for every } x\in X.
$$
Then, its \emph{$c_{p}$-transform} is the function $\phi^{c_{p}}:Y\to \R\cup\{-\infty\}$ defined by
\begin{equation}\label{def:phicp}
\phi^{c_{p}}(y):=\inf_{x\in X} c_{p}(x,y)-\phi(x),
\end{equation}
and its \emph{$c_{p}$-superdifferential} $\partial^{c_{p}} \phi (x)$ at a point $x\in X$ is defined by 
\begin{equation}\label{def:cpSuperDiff}
\partial^{c_{p}} \phi(x):=\{y\in Y\,:\, \phi(x)+\phi^{c_{p}}(y)=c_{p}(x,y)\}.
\end{equation}
Note that
\begin{equation}\label{eq:casesphidcp}
\begin{cases}
\phi(x)=c_{p}(x,y)- \phi^{c_{p}}(y), \quad \text{for all } x\in X,  y\in \partial^{c_{p}}\phi(x)\\
\phi(x)\leq  c_{p}(x,y)-\phi^{c_{p}}(y), \quad \text{for all } x\in X, y\in Y.
\end{cases}
\end{equation}
From the definition it follows readily that if $\phi$ is $c_p$-concave, then for $(x,z)\in J^+\cap (X\times X)$ we have
$$\phi(z)=\inf_{y\in Y}c_{p}(z,y)-u(y) \geq \inf_{y\in Y}c_{p}(x,y)-u(y)=\phi(x),$$
i.e. $\phi$ is a causal function. The same argument gives that $-\phi^{c_p}$ is a causal function as well.

\begin{definition}[Regular $c_{p}$-optimal  dynamical  plan]\label{def:RgcpOptDynPlan}
A $c_{p}$-optimal  dynamical  plan $\Pi\in \cP(\Gamma_{p})$ is  \emph{regular} if the following holds.
\\There exists  $U,V\subset M$ relatively compact open subsets and  a smooth $c_{p}$-concave (with respect to $(U,V)$)  function $\phi_{1/2}:U\to \R$ such that
\begin{enumerate}
\item $\nabla^{q}_{g}\phi_{1/2}(x)\in \cK \subset \subset \Int(\cC)$ for every  $x\in U$ and $$\big(-g(\nabla^{q}_{g}\phi_{1/2}, \nabla^{q}_{g}\phi_{1/2}) \big)^{1/2}< {\rm{Inj}}_{g}(U), $$ where ${\rm{Inj}}_{g}(U)$ is  the injectivity radius of $g$ on $U$;
\item Setting $$\Psi_{1/2}^{t}(x)={\rm exp}^{g}_{x}((t-1/2)\nabla^{q}_{g}\phi_{1/2}(x))$$ and $\mu_{t}:=(\ee_{t})_{\sharp}\Pi$ for every $t\in [0,1]$, it holds that 
$$\supp(\mu_{1/2})\subset U, \; \mu_{t}=(\Psi_{1/2}^{t})_{\sharp}\mu_{1/2} \text{ and }\mu_{t}\ll \vol_{g} \, \forall t\in [0,1].$$
\end{enumerate}
\end{definition}

Roughly, the above notion of regularity asks that the $\cA_{p}$-minimizing curves performing the optimal transport from  $\mu_{0}:=(\ee_{0})_{\sharp}\Pi$ to  $\mu_{1}:=(\ee_{1})_{\sharp}\Pi$ have velocities contained in $\cK$, i.e. they are all ``uniformly''  time-like future pointing. Moreover it also implies that $\cup_{t\in [0,1]}\supp(\mu_{t})\subset M$ is compact; in addition the optimal transport is assumed to be driven by a smooth potential $\phi_{1/2}$. Even if these conditions may appear a bit strong, we will prove in Lemma \ref{lem:SmoothKP} that there are a lot of such regular plans; moreover in the paper we will show that it is enough to consider such particular optimal transports in order to characterize upper and lower bounds on the (causal-)Ricci curvature and thus characterize the solutions of Einstein equations.

\section{Existence, regularity and evolution of Kantorovich potentials}

In order to characterize Lorentzian Ricci curvature upper bounds, it will be useful the next proposition concerning the  evolution of Kantorovich potentials along a regular $\cA_{p}$-minimizing curve of probability measures $(\mu_{t})_{t\in [0,1]}$ given by exponentiating the $q$-gradient of a smooth $c_{p}$-concave function with  time-like gradient. To this aim it is convenient to consider,  for $0\le s<t\le 1$,  the restricted minimal action
$$c^{s,t}_p(x,y):=\inf\left\{\left.\int_s^t \cL_p(\dot\gamma(\tau))d\tau\right|\;   \gamma\in {\rm AC}([s,t], M),\,\gamma(s)=x,\,\gamma(t)=y\right\}.$$

\begin{proposition}\label{thm:HopfLax}
Let $(M,g,\cC)$ be a space-time, fix $p\in (0,1)$ and let $q\in (-\infty, 0)$ be the H\"older conjugate exponent,  i.e. $\frac{1}{p}+\frac{1}{q}=1$ or equivalently $(p-1)(q-1)=1$. Let $U,V\subset M$ be relatively compact open subsets and  $\phi_{1/2}$ be a smooth $c_{p}$-concave function relative to $(U,V)$ such that
\begin{itemize}
\item $\phi_{1/2}$ is a smooth Lyapunov function on $U$,
\item $\big(-g(\nabla^{q}_{g}\phi_{1/2}, \nabla^{q}_{g}\phi_{1/2}) \big)^{1/2}< {\rm{Inj}}_{g}(U)$.
\end{itemize}
For $t\in [0,1]$,  let 
$$\Psi_{1/2}^{t}:U\to M,\;  \Psi_{1/2}^{t}(x):=\exp_{x}((t-1/2) \nabla^{q}_{g}\phi_{1/2}(x)).$$
 For every $x\in U$, define 
\begin{equation}\label{def:phit}
\phi_{t}(\Psi_{1/2}^{t}(x))=\phi(t,\Psi_{1/2}^{t}(x)):=
\begin{cases}
\phi_{1/2}(x)-c_{p}^{1/2,t}(x, \Psi_{1/2}^{t}(x)) & \text{ for $t\in [1/2, 1]$}\\
\phi_{1/2}(x)+c_{p}^{t, 1/2}(\Psi_{1/2}^{t}(x), x) & \text{ for $t\in [0, 1/2)$}
\end{cases}.
\end{equation}
Then the map $(t,y)\mapsto \phi(t,y)$ defined on $\bigcup_{t\in [0,1]}{\{t\}}\times \Psi_{1/2}^{t}(U)$ is $C^{\infty}$ and satisfies  the Hamilton-Jacobi equation 
\begin{equation}\label{eq:qHopfLax}
\partial_{t} \phi_{t} (t,y)+\frac{1}{q} (-g(\nabla_{g} \phi_{t}(y), \nabla_{g} \phi_{t} (y)))^{q/2}=0, \; \forall (t,y)\in \bigcup_{t\in [0,1]}{\{t\}}\times \Psi_{1/2}^{t}(U)
\end{equation}
with 
\begin{equation}\label{eq:velocity}
\frac{d}{dt}{\Psi}_{1/2}^{t}(x)=\nabla^{q}_{g} \phi_{t} (\Psi_{1/2}^{t}(x)) , \quad \forall (t,x) \in [0,1]\times U.
\end{equation}
\end{proposition}

\begin{proof}

\textbf{Step 1}: smoothness of $\phi$.\\
The fact that $t\mapsto \Psi_{1/2}^{t}$ is a smooth 1-parameter family of maps performing   $c_{p}$-optimal transport gives that $\phi$ defined in \eqref{def:phit} satisfies  (cf. \cite[Theorem 6.4.6]{CS}) 
\begin{equation}\label{eq:phitphis}
\phi_{t}(\Psi_{1/2}^{t}(x))=\phi_{s}(\Psi_{1/2}^{s}(x))-c_{p}^{s,t}(\Psi_{1/2}^{s}(x), \Psi_{1/2}^{t}(x)), \quad  \forall x\in U, \;0\leq s<t\leq 1.
\end{equation}
In particular it holds  
\begin{align}
\phi_{t}(\Psi_{1/2}^{t}(x))&=\phi_{0}(\Psi_{1/2}^{0}(x))-c_{p}^{0,t}(\Psi_{1/2}^{0}(x), \Psi_{1/2}^{t}(x))\quad \forall x\in U, t\in [0,1], \label{eq:phitphi0} \\
\phi_{t}(\Psi_{1/2}^{t}(x))&=\phi_{1}(\Psi_{1/2}^{1}(x))+c_{p}^{t,1}(\Psi_{1/2}^{t}(x), \Psi_{1/2}^{1}(x))\quad \forall x\in U, t\in [0,1].  \label{eq:phitphi1}
\end{align}
Since by construction everything is defined inside the injectivity radius and all  the transport rays are non-constant, from \eqref{eq:phitphi0} (respectively \eqref{eq:phitphi1}) it is manifest that the map $(t,y)\mapsto \phi(t,y)$ is $C^{\infty}$ on  $\bigcup_{t\in (0,1]} \{t\}\times \Psi_{1/2}^{t}(U)$ (resp. $ \bigcup_{t\in [0,1)} \{t\}\times \Psi_{1/2}^{t}(U)$). The smoothness of $\phi$ on  $\bigcup_{t\in [0,1]} \{t\}\times \Psi_{1/2}^{t}(U)$ follows.

\textbf{Step 2}: validity of the Hamilton-Jacobi equation \eqref{eq:qHopfLax}.\\
We  consider  $t\in (1/2,1]$, the case $t\in [0,1/2]$ being analogous.
Fix $y=\Psi_{1/2}^{t}(x)$ for some arbitrary $x\in U$ and $t\in (1/2,1]$, and let $\gamma\colon [0,s]\to M$ be a smooth curve with 
$\dot \gamma(0)=v\in T_{y}M$. From \eqref{def:phit} we have 
$$\phi(t+s,\gamma_s)\geq -\int_{0}^{s}  \cL_{p} (\dot{\gamma}_{\tau}) \, d\tau+\phi(t,\gamma_{0}),$$
with equality for $\gamma(\tau)=\Psi_{t+\tau}(x)$ for all $\tau\in [0,s]$. Dividing by $s$ and taking the limit for $s\to 0$, we obtain 
$$
\lim_{s\to 0} \frac{\phi(t+s, \gamma_{s})- \phi(t, \gamma_{0}) }{s} \geq -\cL_{p}(v),
$$
which in turn implies 
$$
\partial_{t} \phi_{t}(y)\geq -d\phi_t(v)- \cL_{p}(v), \quad \text{for every } v\in T_{y}M.
$$
Note that equality holds for $v=\nabla^{q}_{g}\phi_{1/2}(x)$. For $\alpha\in T_{y}^{*}M$, let 
$$
\cH_{p}(\alpha)=\sup_{v\in T_{y}M}\left[ \alpha(v) - \cL_{p}(v) \right]
$$
denote the Legendre transform of $\cL_p$. Thus we get 
\begin{equation}\label{HJ_1}
\partial_{t} \phi_{t}(y)=\cH_p(-d(\phi_t)_y).
\end{equation}
Recalling that $\cH_p$ has the representation \eqref{eq:ReprcHp}, we have
$$\cH_p(-d(\phi_t)_{y})=-\frac{1}{q} \left(-g(\nabla_{g} \phi_{t}(y), \nabla_{g} \phi_{t} (y))\right)^{q/2},$$
which, together with \eqref{HJ_1}, implies \eqref{eq:qHopfLax}.

\textbf{Step 3}: validity of \eqref{eq:velocity}. \\
Since $\Psi_{1/2}^{t}$ is a smooth 1-parameter family of maps performing   $c_{p}$-optimal transport and the function $\phi$ defined in \eqref{def:phit} is smooth, it  coincides with the viscosity solution (resp. backward solution) 
\begin{equation}\label{def:phits}
\phi_{t}(y)=
\begin{cases}
\sup_{z\in \Psi_{1/2}^{s}(U)} \phi_{s}(z)-c_{p}^{s,t}(z, y) & \text{ for $t\in [s, 1]$}\\
\inf_{z\in \Psi_{1/2}^{s}(U)} \phi_{s}(z)+c_{p}^{t, s}(y, z) & \text{ for $t\in [0, s)$}
\end{cases},
\end{equation}
for every $s\in (0,1), \, y\in \Psi_{1/2}^{t}(U)$.

Let us discuss the case $t\in (s,1]$, the other is analogous. From  \eqref{eq:phitphis} it follows that $\Psi_{1/2}^{s}(x)$ is a maximum point in  the right hand side of \eqref{def:phits} corresponding to $y=\Psi_{1/2}^{t}(x)$. Thus
$$
d\phi_{s}(\Psi_{1/2}^{s}(x))=d\left[c_{p}^{s,t}(\cdot, \Psi_{1/2}^{t}(x))\right] (\Psi_{1/2}^{s}(x))= -D\cL_{p} \left(\frac{d}{ds} \Psi_{1/2}^{s}(x) \right).
$$
By construction $\frac{d}{ds} \Psi_{1/2}^{s}(x)\in \Int (\cC)$ and, as already observed,  $D\cL_{p}$ is invertible on $\Int (\cC)$ with inverse given by $D\cH_{p}$. We conclude that
$$
\frac{d}{ds} \Psi_{1/2}^{s}(x)= D \cH_{p} \left(- d\phi_{s}(\Psi_{1/2}^{s}(x))\right)=\nabla^{q}_{g}  \phi_{s}(\Psi_{1/2}^{s}(x)).
$$
\end{proof}

We next show that for every point $\bar{x}\in M$ and every $v\in \cC_{\bar x}$ ``small enough'' we can find a smooth  $c_{p}$-concave function $\phi$ defined on a neighbourhood of $\bar{x}$, such that $\nabla_{g}^{q}\phi=v$ and the hessian of $\phi$ vanishes at $\bar x$. 
This is well known in the Riemannian setting (e.g. \cite[Theorem 13.5]{Vil}) and should be compared with the recent paper by Mc Cann \cite{McCann18} in the Lorentzian framework.
The second part of the next lemma shows that the class of regular  $c_{p}$-optimal dynamical plans is non-empty, and actually rather rich.

\begin{lemma}\label{lem:SmoothKP}
Let $(M,g,\cC)$ be a space-time,  fix $\bar{x}\in M$ and $v\in \cC_{\bar{x}}$ with $g(v,v)<0$. Then there  exists $\varepsilon=\varepsilon(\bar{x},v)>0$ with the following property: 
\begin{enumerate}
\item For every $s\in (0,\varepsilon)$, for every $C^{2}$ function $\phi:M\to \R$ satisfying
\begin{equation}\label{eq:ConstrKPv}
\nabla^{q}_{g}\phi(\bar{x})=sv, \quad \Hess_\phi(\bar{x})=0,
\end{equation}
 there exists a  neighbourhood $U_{\bar{x}}$ of $\bar{x}$ and a neighbourhood $U_{\bar{y}}$ of $\bar{y}:=\exp^{g}_{\bar{x}} (sv)$ such that $\phi$ is $c_{p}$-concave relatively to $(U_{\bar{x}}, U_{\bar{y}})$.
\item Let  
\begin{align*}
\Psi_{1/2}^{t}(x)&:=\exp_{x}((t-1/2) \nabla^{q}_{g}\phi(x)), \quad \forall x\in U_{\bar{x}} \\
\tilde{\Psi}&:U_{\bar{x}}\to {\rm AC}([0,1],M), \quad x\mapsto \Psi_{1/2}^{(\cdot)}(x).
\end{align*}
 Then, for every $\mu_{1/2}\in \cP(M)$ with $\supp(\mu_{1/2})\subset U_{\bar{x}}$, the measure $\Pi:=(\tilde{\Psi})_{\sharp} \mu_{1/2}$ is a  $c_{p}$-optimal dynamical plan. 
\end{enumerate}
\end{lemma}

\begin{proof}
\textbf{(1)} Calling $\bar{y}=\bar{y}(sv):=\exp^{g}_{\bar{x}} (sv)$,  notice that $\nabla^{q}_{g}\phi(\bar{x})=sv$ is equivalent to
\begin{equation}
 d\phi(\bar{x})=D_{x} c_{p}(\bar{x}, \bar{y}),
\end{equation}
where $D_{x} c_{p}(\bar{x}, \bar{y})$ denotes the differential at $\bar{x}$ of the function $x\mapsto  c_{p}(x, \bar{y})$.  Indeed, a computation shows that $D_{x}c_{p}(\bar{x}, \bar{y})=-D\cL_{p}(sv)$ and thus the claim follows from \eqref{eq:DLpnablaphidphi}.

Let ${\phi}:M\to \R$ be any smooth function satisfying
\begin{equation}\label{eq:Constrphitilde}
\nabla^{q}_{g} {\phi}(\bar{x})=sv, \quad \Hess_\phi(\bar{x})=0.
\end{equation}
In what follows we denote with $\Hess_{x, c_{p}}(\bar{x}, \bar{y})$(resp. $\Hess_{v, \cL_{p}}(sv)$ the Hessian of the function $x\mapsto c_{p}(x, \bar{y})$ evaluated at $x=\bar{x}$ (resp. the Hessian of the function $T_{\bar{x}}M\ni w\mapsto \cL_{p}(sv+w)$).
By taking normal coordinates centred at $\bar{x}$ one can check that the operator norm 
$$\|{\Hess_{x, c_{p}}}(\bar{x}, \bar{y})-{\Hess_{v,\cL_{p}}}(sv) \|\to 0 \quad \text{ as } t\to 0.$$
Recalling that from \eqref{eq:d2cLpww} there exists $C_{p,v}>0$ such that  ${\Hess_{v,\cL_{p}}}(sv)\geq C_{p,v} s^{-2+p}$ as quadratic forms, we infer
\begin{equation}\label{eq:Hesscptphi>0}
{\Hess_{x,c_{p}}}(\bar{x}, \bar{y})-{\Hess_\phi} (\bar{x}) >0\quad  \text{as quadratic forms, for every $s\in (0,\varepsilon)$,}
\end{equation}
for some $\varepsilon=\varepsilon(\bar{x},v)>0$ small enough.
Since by construction we have $D_{x} c_{p}(\bar{x}, \bar{y})-d{\phi}(\bar{x})=0$, by the Implicit Function Theorem there exists a neighbourhood $U_{\bar x}\times U_{\bar y}$ of $(\bar x, \bar y)\in M\times M$ and a smooth function $F: U_{\bar{y}}\to U_{\bar x}$ such that  $F(\bar{y})=\bar{x}$ and
$$
D_{x} c_{p}(F(y),y)-d{\phi}(F(y))=0, \quad  \text{for every }y\in U_{\bar y}. 
$$
 Differentiating the last equation in $y$ at $\bar{y}$ and using that ${\Hess_\phi}(\bar{x})=0$, we obtain
 \begin{equation}\label{eq:D2xycpDpsi}
 D^{2}_{yx} c_{p}(\bar{x}, \bar{y}) + {\Hess_{x,c_{p}}}(\bar{x}, \bar{y}) D F (\bar{y})=0.
 \end{equation}
 Using  normal coordinates centred at $\bar{x}$ and  \eqref{eq:d2ldvidvj}  one can check that the operator norm 
 $$
 \left\|D^{2}_{yx} c_{p}(\bar{x}, \bar{y})-  (-g(sv,sv))^{\frac{p-4}{2}} \left(-g(sv,sv) g + (2-p) (sv)^{*}\otimes (sv)^{*}  \right) \right \|\to 0
 $$
 as $s\to 0$,  where $(sv)^{*}=g(sv,\cdot)$ is the covector associated to $sv$. 

 Since by assumption  $g(v,v)<0$ and $p\in (0,1)$, it follows from the reverse Cauchy-Schwartz inequality that $\det[D^{2}_{yx} c_{p}(\bar{x}, \bar{y})]> 0$ for $s\in (0,\varepsilon)$.
Recalling that $\det[{\Hess_{x,c_{p}}}(\bar{x}, \bar{y})]\neq 0 $,  from \eqref{eq:D2xycpDpsi} we infer that $\det(DF(\bar{y}))\neq 0$. By the Inverse Function Theorem, up to reducing the neighbourhoods, we get that  $F: U_{\bar{y}}\to U_{\bar x}$  is a smooth diffeomorphism. 
Define now 
$$u:U_{\bar{y}}\to\R, \quad u(y):=c_{p}(F(y),y)- {\phi} (F(y)).$$
For every fixed $y\in U_{\bar{y}}$, the function $U_{\bar{x}}\ni x\mapsto c_{p}(x, y)- {\phi} (x)-u(y)$ vanishes at $x=F(y)$; moreover, from \eqref{eq:Hesscptphi>0}, it follows that $x=F(y)$ is the strict global minimum of such a function  on  $U_{\bar{x}}$, up to further reducing $U_{\bar{x}}$ and $U_{\bar{y}}$, possibly. In other words, the function $U_{\bar{x}}\times U_{\bar{y}}\ni (x,y)\mapsto  c_{p}(x, y)- {\phi} (x)-u(y)$ is always non-negative and vanishes exactly on the graph of $F$. It follows that 
\begin{equation}\label{eq:tildephiinf}
 {\phi} (x) = \inf_{y\in U_{\bar{y}}}   c_{p}(x,y)-u(y), \quad \text{for every } x\in U_{\bar{x}},
\end{equation}
i.e. $\phi:U_{\bar{x}}\to \R$ is a smooth $c_{p}$-concave function relative to $(U_{\bar{x}}, U_{\bar{y}})$  satisfying  \eqref{eq:ConstrKPv}.
\\

\textbf{Proof of (2).}
\\ \textbf{Step 1}. We first show  that $(\ee_{1/2}, \ee_{t})_{\sharp}\Pi$  is a  $c_{p}$-optimal coupling for $(\mu_{1/2}, \mu_{t})$ and every $t\in [1/2,1]$. To keep notation short, it is convenient to define  
\begin{align*}
\Psi'_{t}(x):=\exp_{x}(t \nabla^{q}_{g}\phi(x)), &\quad \forall x\in U_{\bar{x}} \\
\tilde{\Psi}':U_{\bar{x}}\to {\rm AC}([0,1],M), & \quad  x\mapsto \Psi'_{(\cdot)}(x). 
\end{align*}
Setting $\Pi':=(\tilde{\Psi}')_{\sharp} \mu_{1/2}$, we have $(\ee_{t})_{\sharp} \Pi'=(\ee_{t+1/2})_{\sharp} \Pi$ for every $t\in [0,1/2]$. 
If we show that 
\begin{equation}\label{eq:e0e1Pi'}
(\ee_{0},\ee_{1})_{\sharp}\Pi' \text{ is a  $c_{p}$-optimal coupling for  } \big((\ee_{0})_{\sharp}\Pi', (\ee_{1})_{\sharp}\Pi'\big), 
\end{equation}
then, by { the} triangle inequality, it will follow that $(\ee_{0},\ee_{t})_{\sharp}\Pi'$ is a  $c_{p}$-optimal coupling for $\big((\ee_{0})_{\sharp}\Pi', (\ee_{t})_{\sharp}\Pi' \big) $ for every $t\in [0,1]$; in particular our claim that $(\ee_{1/2}, \ee_{t})_{\sharp}\Pi$  is a  $c_{p}$-optimal coupling for $(\mu_{1/2}, \mu_{t})$, $t\in [1/2,1]$, will be proved. Thus, the rest of step 1 will be devoted to establish \eqref{eq:e0e1Pi'}.

 Since by construction $c_{p}:U_{\bar{x}}\times U_{\bar{y}}\to \R$ is smooth, by classical optimal transport theory it is well know that  the $c_{p}$-superdifferential $\partial^{c_{p}} \phi \subset U_{\bar{y}}$  is $c_{p}$-cyclically monotone (see for instance \cite[Theorem 1.13]{AGUser}). Therefore, in order to have \eqref{eq:e0e1Pi'}, it is enough to prove that
 \begin{equation}\label{eq:cpsubdifphi}
\partial^{c_{p}} \phi(x)=\{\exp_{x}(\nabla^{q}_{g}\phi(x))\}\, \text{ for every  }  x\in U_{\bar{x}}.
\end{equation}
Let us first show that $\partial^{c_{p}} \phi(x)\neq \emptyset$, for every $x\in U_{\bar{x}}$.
From the proof of part (1), there exists a smooth diffeomorphism $F:U_{\bar{y}}\to U_{\bar{x}}$ such that
\begin{equation}\label{eq:phiFcpu}
\begin{cases}
{\phi} (F(y)) =   c_{p}(F(y),y)-u(y), \quad \text{for all } y\in U_{\bar{y}}\\
{\phi} (x) \leq    c_{p}(x,y)-u(y), \quad \text{for all } x\in U_{\bar{x}},  y\in U_{\bar{y}}.
\end{cases}
\end{equation}
From the definition of ${\phi}^{c_{p}}$ in \eqref{def:phicp}, it is readily seen that  ${\phi}^{c_{p}}=u$ on $U_{\bar{y}}$. Thus \eqref{eq:phiFcpu} combined with \eqref{eq:casesphidcp} gives that $y\in \partial^{c_{p}} \phi(F(y))$ for every $y \in U_{\bar{y}}$ or, equivalently,  $F^{-1}(x)\in \partial^{c_{p}} \phi(x)$ for every $x\in U_{\bar{x}}$. In particular, $\partial^{c_{p}} \phi(x)\neq \emptyset$, for every $x\in U_{\bar{x}}$.

Now fix $x\in U_{\bar{x}}$ and   pick $y\in \partial^{c_{p}} \phi(x)\subset U_{\bar{y}}$.  
Since $z\mapsto c_{p}(z, y)$ is differentiable on $U_{\bar{x}}$, we get 
\begin{equation}\label{eq:cpzyxy}
c_{p} (z,y)= c_{p}(x,y)+(D_{x}c_{p}(x,y))[(\exp_{x}^{h})^{-1}(z)]+o({\rm d}_{h}(z,x)), \quad \text{for every $z\in U_{\bar{x}}$}.
\end{equation}
From $y\in \partial^{c_{p}} \phi(x)$, we have
\begin{align*}
\phi(z)-\phi(x) &\overset{\eqref{eq:casesphidcp}}{\leq}   c_{p}(z,y)-c_{p}(x,y) \\
& \overset{\eqref{eq:cpzyxy}}{=} (D_{x}c_{p}(x,y))[(\exp_{x}^{h})^{-1}(z)]+o({\rm d}_{h}(z,x)), \; \text{for every $z\in U_{\bar{x}}$}.
\end{align*}
Since $\phi$ is differentiable at  $x\in U_{\bar{x}}$, it follows that 
$$d\phi(x)=D_{x}c_{p}(x,y)= -D \cL_{p}(w),$$
where $w\in \Int(\cC_{x})$ is such that $y=\exp^{g}_{x}(w)$, { which by \eqref{eq:nablagqphiDHpdphi} is equivalent to 
}
\begin{equation*}
w= D \cH_{p} \left( -d\phi(x)\right)=\nabla^{q}_{g}  \phi(x),
\end{equation*}
which yields $y=\exp^{g}_{x}(w)= \exp^{g}_{x}(\nabla^{q}_{g}  \phi(x))$,  concluding the proof of \eqref{eq:cpsubdifphi}.
\\

\textbf{Step 2}.
Up to further reducing the open set $U_{\bar x}$ and the scale parameter $s>0$ in the definition of $\phi$, we can assume that $\phi$ satisfies the assumptions of Proposition \ref{thm:HopfLax} and that $\bar{\phi}:\Psi'_{1/2}(U_{\bar{x}})\to \R$ defined by
\begin{equation}\label{def:phit}
\bar{\phi}(\Psi'_{1/2}(x)):= \phi(x)-c_{p}^{0,1/2}(x, \Psi'_{1/2}(x)), \quad \forall x\in \Psi'_{1/2}(U_{\bar{x}}),
\end{equation}
is still a Lyapunov function satisfying  
\begin{equation}\label{eq:barphiInj}
\big(-g(\nabla^{q}_{g}\bar{\phi}, \nabla^{q}_{g} \bar{\phi}) \big)^{1/2}< {\rm{Inj}}_{g}(\Psi'_{1/2}(U_{\bar{x}})).
\end{equation}
\\It is easily seen that 
$$x=\exp_{z}\left( - \frac{1}{2} \nabla^{q}_{g}\bar{\phi}(z)\right), \quad \forall x\in U_{\bar x}, \, z:= \Psi'_{1/2}(x).$$
Moreover, thanks to \eqref{eq:barphiInj}, the curve 
$$
[0,1]\ni t\mapsto \exp_{z}\left( (t-1)  \nabla^{q}_{g}\bar{\phi}(z)\right), \quad \forall x\in U_{\bar x}, \, z:= \Psi'_{1/2}(x),
$$
is still a $g$-geodesic, solving the method of characteristics associated to the optimal transport problem (see for instance \cite[Ch. 5.1]{CS} and \cite[Ch. 7]{Vil}; this is actually a variation of step 1 and of the proof of Proposition \ref{thm:HopfLax}).
It follows that the map 
$$\Xi:\Psi'_{1/2}(U_{\bar{x}})\to M, \quad \Xi(z):=  \exp_{z}\left( - \nabla^{q}_{g}\bar{\phi}(z)\right),  \quad \forall x\in U_{\bar x}, \, z:= \Psi'_{1/2}(x)$$
is a $c_{p}$-optimal transport map. In other words, for every  $\bar{\mu}\in {\mathcal P}(M)$ with $\supp \, \bar{\mu}\subset \Psi'_{1/2}(U_{\bar{x}})$, 
$$
\pi:=(\Xi, \id)_{\sharp} \, \bar \mu \in  {\mathcal P}(M \times M)
$$
is a $c_{p}$-optimal coupling for its marginals. We conclude that 
$$
(\ee_{0}, \ee_{1})_{\sharp}\Pi= (\Xi, \id)_{\sharp} \big( (\ee_{1})_{\sharp}\Pi \big)
$$
is a $c_{p}$-optimal coupling for its marginals and thus $\Pi$ is a $c_{p}$-optimal dynamical plan.
\end{proof}

We next establish some basic properties of $c_{p}$-optimal  dynamical  plans which will turn out to be useful for the OT-characterization of Lorentzian Ricci curvature upper and lower bounds.
\begin{lemma}\label{prop:proRegPOpt}
Let $(M,g,\cC)$ be a space-time and let $\Pi$ be a regular  $c_{p}$-optimal dynamical plan with 
\begin{align*}
(\ee_{t})_{\sharp}\Pi&=\mu_{t}=(\Psi_{1/2}^{t})_{\sharp}\mu_{1/2}\ll \vol_{g},\\
 \Psi_{1/2}^{t}(x)&={\rm exp}^{g}_{x}((t-1/2)\nabla^{q}_{g}\phi(x)).
 \end{align*}
  Then 
\begin{enumerate}
\item $\nabla \nabla^{q}_{g}\phi(x):T_{x}M\to T_{x}M$ is a symmetric endomorphism, i.e. 
$$g(\nabla_{X} \nabla^{q}_{g}\phi, Y)= g(\nabla_{Y} \nabla^{q}_{g}\phi, X), \quad   \forall X,Y\in T_{x}M,\; \forall x\in \supp(\mu_{s}).$$
\item  $\Psi_{1/2}^{t}:\supp \, \mu_{1/2}\to M$ is a diffeomorphism on its image for all $t\in [0,1]$.
\item  Calling $\mu_{t}=\rho_{t} \vol_{g}$, the  following Monge-Amp\`ere  equation holds true:
\begin{equation}\label{eq:MongeAmpere}
\rho_{1/2}(x)= \Det_{g}[ D \Psi_{1/2}^{t} (x)] \rho_{t}(\Psi_{1/2}^{t}(x)), \quad \mu_{1/2}\text{-a.e. } x,  \;\forall t\in[0,1].
\end{equation}
\end{enumerate}
\end{lemma}

\begin{proof}
(1) By construction, $\phi$ is smooth on $U$ and $g(\nabla_{g} \phi, \nabla_{g} \phi)<0$. Thus also  $\nabla^{q}_{g}\phi:M\to TM$ is a smooth section of the tangent bundle and the symmetry of the endomorphism $\nabla \nabla^{q}_{g}\phi(x):T_{x}M\to T_{x}M$ follows by Schwartz's Lemma.

(2) is a straightforward consequence of assumption (1) in Definition \ref{def:RgcpOptDynPlan}.

(3) is a straightforward consequence of the change of variable formula.

\end{proof}

\noindent
It will be convenient to consider the matrix of Jacobi fields
\begin{equation}\label{def:Bt}
\cB_{t}(x):=D\Psi_{1/2}^{t}(x):T_{x}M\to T_{\Psi_{1/2}^{t}(x)} M, \quad \text{for all }x\in \supp \, \mu_{1/2},
\end{equation}
along the geodesic $t\mapsto \gamma_{t}:=\Psi_{1/2}^{t}(x)$; recalling  \eqref{eq:JacobiEq}, $\cB_{t}(x)$  satisfies the Jacobi equation
 \begin{equation}\label{eq:JacobieqB}
\nabla_{t} \nabla_{t} \cB_{t}(x)+R(\cB_{t}(x), \dot{\gamma}_{t})\dot{\gamma_{t}}=0,
\end{equation}
where we denoted $\nabla_{t}:=\nabla_{\dot{\gamma}_{t}}$ for short.

Since by  Lemma \ref{prop:proRegPOpt}  we know that $\cB_{t}$ is non-singular for all $x \in \supp \, \mu_{1/2}$, we can define
\begin{equation}\label{def:Ut}
\cU_{t}(x):=\nabla_{t} \cB_{t} \circ \cB_{t}^{-1}: T_{\gamma_{t}} M \to T_{\gamma_{t}} M, \quad \text{for all }x\in \supp \, \mu_{1/2}.
\end{equation}
 The next proposition will be key in the proof of the lower bounds on causal Ricci curvature. It is well known in Riemannian and Lorentzian geometry, 
see for instance \cite[Lemma 3.1]{CMS2} and \cite{EhJuKi}; in any case we report a proof for the reader's convenience.

\begin{proposition}\label{prop:VolDist}
Let $\cU_{t}$ be defined in \eqref{def:Ut}.   Then $\cU_{t}$ is a symmetric endomorphism of $T_{\gamma_{t}}M$ (i.e. the matrix $(\cU_{t})_{ij}$ with respect to an orthonormal basis is symmetric) and it holds
\begin{align*}
\nabla_t \cU_{t} + \cU_{t}^{2} + R(\cdot, \dot{\gamma}_t)\dot{\gamma}_t=0.
\end{align*}
Taking the trace with respect to $g$ yields
\begin{equation}\label{eq:truT'}
\Tr_{g}(\nabla_t \cU_{t}) + \Tr_{g}(\cU_{t}^{2}) +\Ric(\dot{\gamma}_{t}, \dot{\gamma}_{t})=0.
\end{equation}
Setting $y(t):=\log\Det_{g}\cB_{t}$, it holds
\begin{equation}\label{eq:eqdiffyx}
y''(t)+ \frac{1}{n}(y'(t))^2  +  \Ric(\dot\gamma_{t}, \dot\gamma_{t}) \leq 0.
\end{equation}
\end{proposition}

\begin{proof}
Using \eqref{eq:JacobieqB} we get
\begin{align*}
\nabla_t \cU_{t}&= \left(\nabla_t\nabla_t \cB_{t}\right) \cB_{t}^{-1}  +  \nabla_t \cB_{t}  \nabla_{t}(\cB_{t}^{-1}) = - R(\cdot, \dot{\gamma}_{t}) \dot{\gamma}_{t}   - (\nabla_t \cB_{t}) \cB_{t}^{-1} (\nabla_t \cB_{t}) \cB_{t}^{-1} \\
&= - R(\cdot, \dot{\gamma}_{t}) \dot{\gamma}_{t}    - \cU_{t}^{2}. 
\end{align*}
Taking the trace with respect to $g$ yields the second identity. 
\\The rest of the proof is devoted to show \eqref{eq:eqdiffyx}. Let $(e_{i}(t))_{i=1,\dots,n}$  be an orthonormal basis of $T_{\gamma_{t}}M$ parallel along $\gamma$.  Setting $y(t)=\log\det \cB_{t}$, we have that
\begin{align}
y'(t_{0}) &= \left. \frac{d}{dt} \right|_{ t=t_{0}} \log\Det_{g} \big( \cB_{t}\cB_{t_{0}}^{-1} \big)  \nonumber \\
&= \left. \frac{d}{dt} \right|_{ t=t_{0}} \log\det \left[ \big( g(e_{i}(t), e_{j}(t))  g (\cB_{t}\cB_{t_{0}}^{-1} e_{i}(t), e_{j}(t)) \big)_{i,j} \right]  \nonumber \\
&= \Tr_{g} \left[(\nabla_{t} \cB_{t}) \cB_{t_{0}}^{-1} \right]|_{t=t_{0}} =   \Tr_{g}(\cU_{t_{0}}). \quad  \label{eq:y'}
\end{align}
We next show that $\cU_{t}$ is a symmetric endomorphism of $T_{\gamma_{t}}M$, i.e. the matrix $(\cU_{t})_{ij}$ is symmetric. 
To this aim, calling $\cU_{t}^{*}$ the adjoint, we observe that 
\begin{equation}\label{eq:U-U*}
\cU_{t}^{*}-\cU_{t} =(\cB_{t}^{*})^{-1} \left[ (\nabla_{t}\cB_{t}^{*}) \cB_{t}   - \cB_{t}^{*} (\nabla_{t}\cB_{t}) \right]   \cB_{t}^{-1},
\end{equation}
and that 
\begin{equation}\label{eq:DtUU*}
\nabla_{t} \left[ (\nabla_{t}\cB_{t}^{*}) \cB_{t}   - \cB_{t}^{*} (\nabla_{t}\cB_{t}) \right] = (\nabla_{t}\nabla_{t}\cB_{t}^{*}) \cB_{t}   - \cB_{t}^{*} (\nabla_{t}\nabla_{t}\cB_{t}).
\end{equation}
Now the Jacobi equation \eqref{eq:JacobieqB} reads
\begin{equation}\label{eq:D2B}
\nabla_{t} \nabla_{t} \cB_{t}= -  R(\cB_{t}, \dot{\gamma}_{t})\dot{\gamma}_{t}=- {\mathcal R}(t) \cB_{t},  
\end{equation}
where 
\begin{equation*}
{\mathcal R}(t):T_{\gamma_{t}} M \to T_{\gamma_{t}} M, \quad  {\mathcal R}(t)[v]:=R(v,\dot{\gamma}_{t}) \dot{\gamma}_{t} 
\end{equation*}
is symmetric; indeed, in the orthonormal basis $(e_{i}(t))_{i=1,\ldots,n}$, it  is represented by the symmetric matrix 
$$\Big(g(e_{i}(t), e_{j}(t))\, g(R(e_{i}(t),\dot{\gamma}_{t}) \dot{\gamma}_{t}, e_{j}(t)) \Big)_{i,j=1,\dots, n}.$$ 
Plugging \eqref{eq:D2B} into \eqref{eq:DtUU*}, we obtain that 
$$(\nabla_{t}\cB_{t}^{*}) \cB_{t}   - \cB_{t}^{*} (\nabla_{t}\cB_{t})$$
 is constant in $t$.
But $\cB_{1/2}={\rm Id}_{T_{\gamma_{1/2}}M}$ and $\nabla_{t}\cB_{t}|_{t=1/2}= -\nabla \nabla_{g}^{q}\phi$ is symmetric by assertion (1) in   Lemma \ref{prop:proRegPOpt}.  
Taking into account \eqref{eq:U-U*}, we conclude that $\cU_{t}$ is symmetric for every $t\in [0,1]$.
\\Using that $\cU_{t}$ is symmetric, by Cauchy-Schwartz inequality,  we have that  
\begin{equation}\label{eq:CSUT}
\Tr_{g}\big[\cU_{t}^{2} \big]= \tr\big[ (\cU^{2})_{ij} \big]\geq \frac{1}{n} \left(\tr\big[ \cU_{ij} \big] \right)^{2}= \frac{1}{n} \left(\Tr_{g}\big[ \cU_{ij} \big] \right)^{2}.
\end{equation}
 The desired estimate \eqref{eq:eqdiffyx} then follows  from the combination of \eqref{eq:truT'}, \eqref{eq:y'} and \eqref{eq:CSUT}.
\end{proof}

\section{Optimal transport formulation of the Einstein equations}

The Einstein equations of General Relativity for an $n$-dimensional space-time $(M^{n},g,\cC)$,  $n\geq 3$, read as
\begin{equation}\label{eq:EFE}
\Ric-\frac{1}{2} {\rm Scal} \, g +\Lambda g=8\pi T,
\end{equation}
where ${\rm Scal}$ is the scalar curvature,  $\Lambda \in \R$ is the cosmological constant, and $T$ is the energy-momentum tensor.

\begin{lemma}\label{Lem-T}
The space-time  $(M^{n},g,\cC)$,    $n\geq 3$, satisfies the Einstein Equation \eqref{eq:EFE} if and only if 
\begin{equation}\label{eq:EFERicT}
\Ric=\frac{2\Lambda}{n-2} g + 8 \pi T - \frac{8\pi} {n-2} \Tr_{g}(T) \, g.    
\end{equation}
\end{lemma}

\begin{proof}
Taking the trace of \eqref{eq:EFE},  one can express the scalar curvature as 
\begin{equation}\label{eq:Scal}
 {\rm Scal} = \frac{2n \Lambda}{n-2} - \frac{16 \pi}{n-2} \Tr_{g}(T).
\end{equation}
Plugging  \eqref{eq:Scal} into  \ref{eq:EFE} gives the equivalent formulation \eqref{eq:EFERicT} of Einstein equations just in  terms of the metric, the  Ricci and the energy-momentum tensors.
\end{proof}

The optimal transport formulation of the Einstein equations will consist separately  of an optimal transport characterization of the two inequalities  
$$\Ric\geq \tilde{T} \quad\text{and}\quad \Ric\leq \tilde{T},$$
where
\begin{equation}\label{eq:tildeTgT}
\tilde{T}:= \frac{2\Lambda}{n-2} g + 8 \pi T - \frac{8\pi} {n-2} \Tr_{g}(T) \, g.
\end{equation}

Subsection  \ref{SS:Ricgeq} will be devoted to the lower bound and Subsection \ref{SS:Ricleq} to the upper bound on the Ricci tensor.

A key role in such an optimal transport formulation will be played by the (relative) Boltzmann-Shannon entropy defined below. Denote by $\vol_{g}$ the standard volume measure on $(M,g)$. Given an absolutely continuous probability measure $\mu=\varrho\, \vol_{g}$ with density $\varrho\in C_{c}(M)$, define its  Boltzmann-Shannon entropy (relative to $\vol_{g}$) as
\begin{equation}\label{eq:defEnt}
 \Ent(\mu|\vol_{g}):=\int_{M} \varrho \log \varrho \, d\vol_{g}.
\end{equation}

\subsection{OT-characterization of $\Ric\geq \tilde{T}$}\label{SS:Ricgeq}

In establishing the  Ricci curvature lower bounds, the next elementary lemma will be key (for the proof see for instance \cite[Chapter 16]{Vil}  or \cite[Lemma 9.1]{AMS})

\begin{lemma}\label{lem:comparisonODE}
Define the function $\rG:[0,1]\times [0,1]\to [0,1]$ by
\begin{equation}
  \label{eq:defgreen}
  \rG(s,t):=
  \begin{cases}
    (1-t)s&\text{if }s\in[0,t],\\
    t(1-s)&\text{if }s\in [t,1],
  \end{cases}
\end{equation}
so that for all $t\in (0,1)$ one has
\begin{equation}
  \label{eq:greenindentity}
  -\frac{\partial^2}{\partial s^2}  \rG(s,t) =\delta_{t}\quad\text{in $\mathscr D'(0,1)$},\qquad 
   \rG (0,t)= \rG (1,t)=0.
\end{equation}
If  $u\in C([0,1],\R)$ satisfies $u''\geq f$ in $\mathscr D'(0,1)$ for some $f\in L^{1}(0,1)$ then 
\begin{align}\label{equ:SLv}
u(t)\leq (1-t) u(0) + t u(1) - \int_{0}^{1} \rG (s,t)\,   f(s) \, d s, \quad \forall t\in [0,1].
\end{align}
In particular, if $f\equiv c\in \R$ then 
\begin{align}\label{equ:SLvconst}
u(t)\leq (1-t) u(0) + t u(1) -c \frac{t(1-t)}{2}, \quad \forall t\in [0,1].
\end{align}
\end{lemma}

The characterization of Ricci curvature lower bounds (i.e. $\Ric\geq Kg$ for some constant $K\in \R$) via displacement convexity of the  entropy is by now classical in the Riemannian setting, let us briefly recall the key contributions. Otto \& Villani \cite{OV} gave a nice heuristic argument for the implication ``$\Ric\geq K g \Rightarrow$ $K$-convexity of the  entropy''; this implication was  proved  for $K=0$ by Cordero-Erausquin, McCann \& Schmuckenschl\"ager \cite{CMS}; the equivalence for every $K\in \R$ was then established by Sturm \& von Renesse \cite{SVR}.
Our optimal transport characterization of $\Ric\geq \frac{2\Lambda}{n-2} g + 8 \pi T - \frac{8\pi} {n-2} \Tr_{g}(T) \, g$ is inspired by such fundamental papers (compare also with \cite{KM} for the implication (3)$\Rightarrow$(1)). Let us also mention that the characterization of $\Ric\geq Kg$ for $K\geq 0$  via displacement convexity in the globally hyperbolic Lorentzian setting has recently been obtained independently by Mc Cann \cite{McCann18}. Note that Corollary \ref{thm:RicciLowerBound} extends such a result to any lower bounds $K\in \R$ and to the case of general (possibly non   globally hyperbolic) space  times.

The next general result  will be applied with $n\geq 3$ and $\tilde{T}$ be as in \eqref{eq:tildeTgT}.

\begin{theorem}[OT-characterization of $\Ric\geq \tilde{T}$]\label{thm:RiccigeqT}
Let $(M,g,\cC)$ be a space-time of dimension $n\geq 2$ and let $\tilde{T}$ be a quadratic form on $M$.  Then the following are equivalent:
\begin{enumerate}
\item  $\Ric(v,v)\geq \tilde{T}(v,v)$ for every causal vector $v\in \cC$.
\item  For every $p\in (0,1)$,  for every regular dynamical $c_{p}$-optimal  plan  $\Pi$ it holds
\begin{equation}\label{eq:EntConvtT}
\Ent(\mu_t|\vol_{g})\leq (1-t) \Ent(\mu_0|\vol_{g})+t \Ent(\mu_1|\vol_{g}) - \int  \int_{0}^{1} \rG(s,t) \tilde{T}(\dot{\gamma}_{s}, \dot{\gamma}_{s}) ds\, d\Pi(\gamma),
\end{equation}
where we denoted $\mu_{t}:=(\ee_{t})_{\sharp}\Pi$, $t\in [0,1]$, the curve of probability measures associated to $\Pi$.
\item  There exists $p\in (0,1)$ such that for every regular dynamical $c_{p}$-optimal  plan  $\Pi$ the convexity property  \eqref{eq:EntConvtT} holds.
\end{enumerate}
\end{theorem}

\begin{proof}

\textbf{(1)$\Rightarrow$ (2)}\\
Fix $p\in (0,1)$.
Let $\Pi$ be a regular dynamical $c_{p}$-optimal  coupling and let $(\mu_{t}:=(\ee_{t})_{\sharp}\Pi)_{t\in [0,1]}$ be the corresponding curve of probability measures with $\mu_{t}=\rho_{t} \, \vol_{g} \ll \vol_{g}$ compactly supported. By definition of  regular dynamical $c_{p}$-optimal  coupling there exists a smooth function  $\phi_{1/2}$  such that, calling 
$$\Psi_{1/2}^{t}(x)={\rm exp}^{g}_{x}((t-1/2)\nabla^{q}_{g}\phi_{1/2}(x)),$$ it holds  $\mu_{t}=(\Psi_{1/2}^{t})_{\sharp} \mu_{1/2}$ for every $t\in [0,1]$. Moreover the Jacobian $D\Psi_{1/2}^{t}$ is non-singular for every $t\in [0,1]$ on $\supp(\mu_{1/2})$.   Recall the definition of  $\cB_{t}$ and $\cU_{t}$ along the geodesic $t\mapsto \gamma_{t}:=\Psi_{1/2}^{t}(x)$.
\begin{align*}
\cB_{t}(x)&:=D\Psi_{1/2}^{t}(x):T_{x}M\to T_{\Psi_{1/2}^{t}(x)} M, &  \text{for } \mu_{1/2}\text{-a.e. }x, \\
\cU_{t}(x)&:=\nabla_{t} \cB_{t} \circ \cB_{t}^{-1}: T_{\gamma_{t}} M \to T_{\gamma_{t}} M,  &\text{for } \mu_{1/2}\text{-a.e. }x.
\end{align*}
Calling  $y_{x}(t):=\log\Det_{g}\cB_{t}(x)$ and $\gamma^{x}_{t}:=\Psi_{1/2}^{t}(x)$, from Proposition \ref{prop:VolDist} we get
\begin{equation}\label{eq:eqdiffyxPftT} 
y_{x}''(t)  + \tilde{T}(\dot\gamma^{x}_{t}, \dot\gamma^{x}_{t})\leq   y_{x}''(t)  +  \Ric(\dot\gamma^{x}_{t}, \dot\gamma^{x}_{t}) \leq 0, \quad \mu_{1/2}\text{-a.e. }x.
\end{equation}
Now, for  $t \in [0,1]$ we have
\begin{align}
\Ent(\mu_t|\vol_{g}) &=\int \log\rho_{t}(y) \, d\mu_{t}(y) = \int_M  \log\rho_{t}(\Psi_{1/2}^{t}(x)) \, d\mu_{1/2}(x)   \nonumber \\
&= \int  \log[\rho_{1/2}(x) (\Det_{g}(D\Psi_{1/2}^{t})(x))^{-1}]\, d\mu_{1/2}(x) \nonumber \\
&= \Ent(\mu_{1/2}|\vol_{g})-   \int  y_{x}(t)\, d\mu_{1/2}(x),  \label{eq:EntHpytT}
\end{align}
where the second to last equality follows from Lemma \ref{prop:proRegPOpt}(3). 
Using  \eqref{eq:eqdiffyxPftT} we obtain
\begin{equation}
\frac{d^{2}}{dt^{2}}\Ent(\mu_t|\vol_{g})\geq  \int  \tilde{T}(\dot\gamma^{x}_{t}, \dot\gamma^{x}_{t}) d\mu_{1/2}(x)=   \int \tilde{T} (\dot{\gamma}_{t}, \dot{\gamma}_{t}) d\Pi(\gamma), \quad \forall t\in (0,1).
\end{equation}
Using Lemma \ref{lem:comparisonODE},    we get  \eqref{eq:EntConvtT}.

\textbf{(2)$\Rightarrow$ (3)}: trivial.

\textbf{(3)$\Rightarrow$ (1)}
\\ 
We argue by contradiction. Assume there exist $x_{0}\in M$ and $v\in T_{x_{0}}M\cap \cC$ with $g(v,v)<0$  such that 
the Ricci curvature at $x_{0}$ in the direction of $v\in \cC_x$ satisfies 
\begin{equation}\label{eq:contradEnt}
\Ric(v,v)\leq (\tilde{T}+3\epsilon g)(v,v), 
\end{equation} 
 for some $\epsilon>0$. Thanks to Lemma \ref{lem:SmoothKP}, for $\eta\in (0, \bar{\eta}(x_{0}, v)]$ small enough, there exists $\bar{\delta}>0$ and a  $c_{p}$-convex function $\phi_{1/2}$, smooth on $B_{\bar{\delta}}(x_{0})$   and satisfying
 \begin{equation}\label{eq:AssPhiNonDeg}
\nabla_{g}^{q} \phi_{1/2}(x_{0}) = \eta v \neq 0 \quad \text{ and }  \Hess_{\phi_{1/2}} (x_{0})=0. 
\end{equation}
From now  on we fix  $\eta\in \big(0, \min( \bar{\eta}(x_{0}, v), {\rm inj}_{g}(B_{\bar{\delta}}(x_{0})) ) \big]$, where ${\rm inj}_{g}(B_{\bar{\delta}}(x_{0}))$ is the  injectivity radius of $B_{\bar{\delta}}(x_{0})$ with respect to the metric $g$.  
It is easily checked that, for $\delta\in (0,\bar{\delta})$, small enough the map
$$x\mapsto \Psi_{1/2}^{t}(x)=\exp^{g}_x((t-1/2)\nabla^{q}_{g} \phi_{1/2}(x))$$ is 
a diffeomorphism from $B_{\delta}(x_{0})$ onto its image for any $t\in [0,1]$. Moreover, since  $\nabla_{g}^{q} \phi_{1/2}(x_{0})\in \Int(\cC)$ and arguing by continuity and by parallel transport along the geodesics $t\mapsto \Psi_{1/2}^{t}(x), x\in B_{\delta}(x_{0})$,  for $\delta>0$ small enough we have that 
\begin{equation}\label{eq:PiRegular}
\bigcup_{t\in [0,1]}\bigcup_{x\in B_{\delta}(x_{0})} \frac{d}{dt} \Psi_{1/2}^{t}(x) \subset \subset \Int( \cC).
\end{equation} 
Define $\mu_{\frac 1 2 }:= \vol_{g}(B_{\delta}(x_{0}))^{-1}  \,\vol_{g} \llcorner B_{\delta}(x_{0})$.
Let $\Pi$ be the $c_{p}$-optimal  dynamical plan representing  the curve of probability measures  $\left(\mu_{t}:=(\Psi_{1/2}^{t})_{\sharp}\mu_{\frac  1 2} \right)_{t\in [0,1]}$. Note that \eqref{eq:PiRegular} together with \eqref{eq:AssPhiNonDeg} ensures that $\Pi$ is regular, for $\delta>0$ small enough.
\\Calling $\gamma^{x}_{t}:=\Psi_{1/2}^{t}(x)=\exp^{g}_{x}\big( (t-1/2) \nabla^{q}_{g} \phi_{1/2}(x)\big)$ for $x \in B_{\delta}(x_{0})$ the geodesic performing the transport,  note  that by continuity  there exists  $\delta>0$ small enough such that  
\begin{equation}\label{eq:AssAbsRicp}
\Ric(\dot{\gamma}^{x}_{t}, \dot{\gamma}^{x}_{t})<   (\tilde{T}+2\epsilon g)(\dot{\gamma}^{x}_{t}, \dot{\gamma}^{x}_{t}), \quad \forall x \in  B_{\delta}(x_{0}), \quad \forall  t \in \left[0,1\right]. 
\end{equation}
The identity \eqref{eq:truT'} proved in Proposition \ref{prop:VolDist} reads as
\begin{equation}\label{ricattiContr}
[\Tr_{g}(\cU_{t}^{x})]' + \Tr_{g}[(\cU_{t}^{x})^{2} ]+\Ric(\dot{\gamma}^{x}_{t}, \dot{\gamma}^{x}_{t})=0,   \quad  \forall x \in  B_{\delta}(x_{0}), \quad \forall t \in [0,1].
\end{equation}
Since by construction $\cU_{1/2}^{x_{0}}:=\nabla_{t} \cB^{x_{0}}_{1/2}( \cB^{x_{0}}_{1/2})^{-1}=  \nabla  \nabla_{g}^{q}\phi_{1/2} (x_{0})=0$ and $g(\nabla_{g}^{q}\phi_{1/2} (x_{0}),\nabla_{g}^{q}\phi_{1/2} (x_{0}))<0$,  again by continuity we can choose $\delta>0$ even smaller so that
\begin{align}\label{eq:claim}
\Tr_{g}[(\cU_{t}^{x})^{2} ]< -\epsilon  g(\dot{\gamma}^{x}_{t}, \dot{\gamma}^{x}_{t}),   \quad  \forall x \in  B_{\delta}(x_{0}), \quad \forall t \in [0,1].
\end{align}
The combination of  \eqref{eq:AssAbsRicp}, \eqref{ricattiContr} and   \eqref{eq:claim} yields
\begin{align*}
[\Tr_{g}(\cU_{t}^{x})]'+  (\tilde{T}+\epsilon g)   (\dot{\gamma}^{x}_{t}, \dot{\gamma}^{x}_{t})>0 ,  \quad \forall x \in  B_{\delta}(x_{0}), \quad \forall  t \in \left[0,1 \right].
\end{align*}
Recalling \eqref{eq:y'}, the last inequality can be rewritten as
\begin{align*}
y_{x}(t)''+ (\tilde{T}+\epsilon g)   (\dot{\gamma}^{x}_{t}, \dot{\gamma}^{x}_{t})  > 0 ,  \quad \forall x \in  B_{\delta}(x_{0}), \quad \forall  t \in \left[0,1 \right].
\end{align*}
The combination of the last inequality with \eqref{eq:EntHpytT} gives
\begin{equation}\label{eq:Entmugs''}
\frac{d^{2}}{dt^{2}}\Ent(\mu_t|\vol_{g})< \int (  \tilde{T}+\epsilon g )   (\dot{\gamma}^{x}_{t}, \dot{\gamma}^{x}_{t}) \, d\mu_{1/2}(x), \quad \forall t\in (0,1).
\end{equation}
By applying Lemma \ref{lem:comparisonODE} we get that
\begin{align*}
&\Ent(\mu_t|\vol_{g})\geq  (1-t)\, \Ent(\mu_0|\vol_{g})+t \, \Ent(\mu_1|\vol_{g}) \\
& \qquad \qquad \qquad \qquad  - \int  \int_{0}^{1} \rG(s,t) (\tilde{T}+\epsilon g)(\dot{\gamma}_{s}, \dot{\gamma}_{s}) ds\, d\Pi(\gamma), \\
&\quad =  (1-t)\, \Ent(\mu_0|\vol_{g})+t \, \Ent(\mu_1|\vol_{g})  - \int  \int_{0}^{1} \rG(s,t) \tilde{T}(\dot{\gamma}_{s}, \dot{\gamma}_{s}) ds\, d\Pi(\gamma) \\
&\qquad \qquad - \epsilon \frac{t(1-t)}{2} \int g(\dot{\gamma}, \dot{\gamma}) d\Pi(\gamma),
\end{align*}
where, in the equality we used that for every fixed $x\in  B_{\delta}(x_{0})$ the function $t\mapsto g(\dot{\gamma}^{x}_{t}, \dot{\gamma}^{x}_{t})$ is constant (as $t\mapsto {\gamma}^{x}_{t}$ is by construction a $g$-geodesic).

This  clearly contradicts \eqref{eq:EntConvtT}, as $\int g(\dot{\gamma}, \dot{\gamma}) d\Pi(\gamma)<0$.
\end{proof}

In the vacuum case,  i.e.  $T\equiv 0$, the inequality $\Ric\geq \tilde{T}$   with $\tilde{T}$ as in \eqref{eq:tildeTgT}  reads as $\Ric\geq Kg$ with $K=\frac{2\Lambda}{n-2} \in \R$.  Note that for $v\in \cC$ it holds $g(v,v)\leq 0$ so, when comparing the next result with its Riemannian counterparts \cite{OV, CMS,SVR}, the sign of the lower bound $K$ is reversed. 

\begin{corollary}[The vacuum case $T\equiv 0$]\label{thm:RicciLowerBound}
Let $(M,g,\cC)$ be a space-time of dimension $n\geq 2$ and let $K\in \R$.  Then the following are equivalent:
\begin{enumerate}
\item  $\Ric(v,v)\geq K g(v,v)$ for every causal vector $v\in \cC$.
\item  For every $p\in (0,1)$,  for every regular dynamical $c_{p}$-optimal  plan  $\Pi$ it holds
\begin{align}
\Ent(\mu_t|\vol_{g})&\leq (1-t) \Ent(\mu_0|\vol_{g})+t \Ent(\mu_1|\vol_{g}) -K \frac{t(1-t)}{2} \int g(\dot{\gamma}, \dot{\gamma}) d\Pi(\gamma) \nonumber \\
                            &=  (1-t) \Ent(\mu_0|\vol_{g})+t \Ent(\mu_1|\vol_{g}) -K t (1-t) \int \cA_{2}(\gamma)  d\Pi(\gamma) \label{eq:EntConv},
\end{align}
where we denoted $\mu_{i}:=(\ee_{i})_{\sharp}\Pi$, $i=0,1$, the endpoints of the curve of probability measures associated to $\Pi$.

\item  There exists $p\in (0,1)$ such that for every regular dynamical $c_{p}$-optimal  plan  $\Pi$ the convexity property  \eqref{eq:EntConv} holds.
\end{enumerate}

\end{corollary}

\begin{remark}[The strong energy condition]
The strong energy condition asserts that, called $T$ the energy-momentum tensor, it holds $T(v,v)\geq \frac{1}{2}\Tr_{g}(T)$ for every time-like vector $v\in TM$  satisfying $g(v,v)=-1$.  Assuming that the space-time $(M,g,\cC)$ satisfies the Einstein equations \eqref{eq:EFE} with zero cosmological constant $\Lambda=0$, the strong energy condition is equivalent to $\Ric(v,v)\geq 0$ for every time-like vector $v\in TM$. This corresponds to the case $K=0$ in Corollary \ref{thm:RicciLowerBound}.

The strong energy condition, proposed  by Hawking and Penrose \cite{Pen,Haw66, HawPen70},  plays a key role in general relativity. For instance, in the presence of trapped surfaces, it  implies that the space-time has singularities (e.g. black holes) \cite{EH, W}.
\end{remark}

\subsection{OT-characterization of $\Ric\leq  \tilde{T}$}\label{SS:Ricleq}

The goal of the present section is to provide an optimal transport formulation of upper bounds on \emph{time-like} Ricci curvature in the \emph{Lorentzian setting}. More precisely, given a quadratic form $\tilde{T}$ (which will later be chosen to be equal to the right hand side of Einstein equations, i.e. as in \eqref{eq:tildeTgT}), we aim to find an optimal transport formulation of the condition  
$$``\Ric(v,v)\leq \tilde{T}(v,v)\text{ for every time-like vector }v\in \cC\text{''}.$$
The Riemannian counterpart, in the special case of $\Ric\leq Kg$ for some constant $K\in \R$,  has been recently established by Sturm \cite{StUB}.

In order to state the result, let us fix some notation. Given a  relatively  compact  open subset $E\subset\subset  \Int(\cC)$ let $p_{TM\to M}:TM\to M$ be the canonical projection map and ${\rm inj}_{g}(E)>0$ be  the injectivity radius of the exponential map of $g$ restricted to $E$. For  $x\in p_{TM\to M}(E)$ and $r\in (0, {\rm inj}_{g}(E))$ we  denote 
$$B^{g,E}_{r}(x):= \{\exp_{x}^{g}(t w): w\in T_{x}M\cap E, \, g(w,w)=-1, t\in [0,r] \}.$$

\begin{definition}[$r$-concentrated regular $c_{p}$-optimal dynamical  plan]\label{def:r-concRODP}
Fix a  relatively  compact  open subset $E\subset\subset  \Int(\cC)$, $x\in p_{TM\to M}(E)$, $v\in T_{x}M \cap E$, and $r>0$. A regular $c_{p}$-optimal dynamical  plan $\Pi$ is called \emph{$r$-concentrated} in the direction of $v$ (with respect to $E$) if $\mu_{t}:=(\ee_{t})_{\sharp}\Pi$ for $t\in [0,1]$ satisfies:
\begin{enumerate}
\item $\mu_{1/2}=\vol_{g}(B^{g,E}_{r^{4}}(x))^{-1}\,  \vol_{g}\llcorner B^{g,E}_{r^{4}}(x)$;   
\item $\supp(\mu_{1})\subset \{\exp_{y}^{g}(r^{2}w): w\in T_{y}M\cap \cC, \, g(w,w)=-1 \}$, where $y:=\exp^{g}_{x}(rv)$.
\end{enumerate}
\end{definition}
 
The next general result  will be applied with  $n\geq 3$ and $\tilde{T}$ as in  \eqref{eq:tildeTgT}.

\begin{theorem}[OT-characterization of $\Ric\leq \tilde{T}$]\label{thm:RiccileqT}
Let $(M,g,\cC)$ be a space-time of dimension $n\geq 2$ and let $\tilde{T}$ be a quadratic form on $M$.  Then the following assertions are equivalent:
\begin{enumerate}
\item[(1)]  $\Ric(v,v)\leq \tilde{T}(v,v)$ for every causal vector $v\in \cC$.
\item[(2)] For every $p\in (0,1)$ and for every relatively  compact  open subset $E\subset\subset  \Int(\cC)$ there exist $R=R(E)\in (0,1)$ and a function 
\begin{align*}
&\epsilon=\epsilon_{E}:(0,\infty)\to (0,\infty)\textrm{ with }\lim_{r\downarrow 0}\epsilon(r)=0 \textrm{ such that }\\
&\forall x\in p_{TM\to M}(E) \textrm{ and }v\in T_{x}M\cap E \textrm{ with }g(v,v)=-R^{2}
\end{align*}
 the next assertion holds.
For every $r\in (0,R)$,  there exists an $r$-concentrated regular $c_{p}$-optimal  dynamical plan $\Pi=\Pi(x,v,r)$ in the direction of $v$ (with respect to $E$)
%
%
%
which has $ \tilde{T}(v,v)$-concave entropy in the sense that
\begin{equation}\label{eq:ChUBtT}
\frac{4}{r^{2}}\left[\Ent(\mu_{1}|\vol_{g})-2 \Ent(\mu_{1/2}|\vol_{g})+\Ent(\mu_{0}|\vol_{g}) \right] \leq  \tilde{T}(v,v)+\epsilon(r).
\end{equation}
\item [(3)] There exists  $p\in (0,1)$ such that the analogous  assertion as in  {\rm (2)} holds true.
\end{enumerate}
\end{theorem}

\begin{remark}\label{rem:RiccileqT}
Given an auxiliary Riemannian metric $h$ on $M$,  with analogous arguments as in the proof below, the condition $(2)$ in Theorem \ref{thm:RiccileqT} can be replaced by 
\begin{enumerate}
\item[(2)']   For every $p\in (0,1)$ and for  every $x\in M$ there exist $R=R(x)>0$  and a function  $\epsilon=\epsilon_{x}:(0,\infty)\to (0,\infty)$  with $\lim_{r\downarrow 0}\epsilon(r)=0$  such that
for all  $v\in \Int(\cC_{x})$  with  $h(v,v)\leq R^{2}$ the next assertion holds.
For every $r\in (0,R)$,   there exists an $r$-concentrated regular $c_{p}$-optimal  dynamical plan $\Pi=\Pi(x,v,r)$ in the direction of $v$ (with respect to $E$) satisfying \eqref{eq:ChUBtT}.
\end{enumerate}
Moreover, both in (2) and (2') one can replace $B^{g,E}_{r^{4}}(x)$  (resp.  $\{\exp_{y}^{g}(r^{2}w)\,:\, w\in T_{y}M\cap \cC, \, g(w,w)=-1 \}$)  by $B^{h}_{r^{4}}(x)$ (resp.  $B^{h}_{r^{2}}(y)$).
\end{remark} 

\medskip

\begin{proof}

\textbf{(1)$\Rightarrow$ (2)}
\\Let $(M,g,\cC)$ be a space time and let $h$ be an auxiliary Riemannian metric on $M$ such that 
$$\frac{1}{4} h(w,w) \leq |g(w,w)|\leq 4 h(w,w), \quad \forall w\in E.$$  
 We denote with ${\rm d}^{TM}_{h}$ the distance on $TM$ induced by the auxiliary Riemannian metric $h$.
Once the compact subset $E \subset\subset  \Int(\cC)$ is fixed, thanks to Lemma \ref{lem:SmoothKP} there exist a constant
$$R=R(E)\in (0,\min(1, {\rm inj}_{g}(E)))$$ and a function
$$\epsilon=\epsilon_{E}:(0,\infty)\to (0,\infty) \text{ with } \lim_{r\downarrow 0}\epsilon(r)=0$$
 such that
 $$ \forall r\in (0,R/10), \,x\in p_{TM\to M}(E), \, v\in T_{x}M\cap E \text{ with } g(v,v)=-R^{2}$$  
 we can find a  $c_{p}$-convex function $\phi:M\to \R$ with the following properties:
\begin{enumerate}
\item $\phi$ is smooth on  $B^{h}_{100r}(x)$, $\nabla^{q}_{g}\phi(x)=v$, $\nabla^{q}_{g}\phi\in E$ on $B^{h}_{10r}(x)$, \\${\rm d}^{TM}_{h}(\nabla^{q}_{g}\phi,v)\leq   \epsilon(r)$ on $B^{h}_{10r}(x)$;
\item $|\Hess_\phi|_{h} \leq  \epsilon(r)$ on $B^{h}_{10r}(x)$. 
\end{enumerate}
For $t\in [0,1]$, consider the map $$\Psi_{1/2}^{t}:z\mapsto \exp_{z}(r(t-1/2) \nabla^{q}_{g}\phi(z)).$$
Notice that 
$$\Psi_{1/2}^{t} ( B^{g,E}_{r^{4}}(x))\subset B^{h}_{10r}(x), \quad \forall t\in [0,1].$$
Let 
$$\mu_{1/2}=\vol_{g}(B^{g,E}_{r^{4}}(x))^{-1}\,  \vol_{g}\llcorner B^{g,E}_{r^{4}}(x)$$  and define 
$$\mu_{t}:=(\Psi_{1/2}^{t})_{\sharp} (\mu_{1/2}) \quad \forall t\in [0,1].
$$ 
By the properties of $\phi$,  the plan $\Pi$ representing the curve of probability measures $(\mu_{t})_{t\in [0,1]}$ is a regular  $c_{p}$-optimal  dynamical plan and  
$$\supp(\mu_{1})\subset \{\exp_{y}^{g}(r^{2}w): w\in T_{y}M\cap \cC, \, g(w,w)=-1 \}.
$$
By  Proposition \ref{thm:HopfLax} we can find  a smooth family of functions $(\phi_{t})_{t\in [0,1]}$ defined on 
$\bigcup_{t\in [0,1]} \{t\}\times \supp(\mu_{t})$ with $\phi_{1/2}=\phi$  satisfying  
\begin{align}
(\partial_{t} \phi_{t})(\gamma_{t})+\frac{r}{q} \left(-g(\nabla_{g} \phi_{t} (\gamma_{t}), \nabla_{g} \phi_{t}(\gamma_{t})) \right)^{q/2}=0 & \quad  \text{for } \Pi\text{-a.e. }\gamma,  \text{ for all } t\in [0,1],   \label{eq:qHopfLaxProof1} \\
r\nabla_{g}^{q} \phi_{t} (\gamma_{t})-\dot{\gamma}_{t}=0 &\quad \text{for } \Pi\text{-a.e. }\gamma,  \text{ for all } t\in [0,1]. \label{eq:velocityProof1}
\end{align}
Moreover, using  the properties of $\phi_{1/2}=\phi$ and  the smoothness of the family $(\phi_{t})_{t\in [0,1]}$, we have  
\begin{equation}\label{eq:Estphit}
{\rm d}^{TM}_{h} (\nabla^{q}_{g}\phi_{t} (\gamma_{t}), v)\leq \epsilon(r), \; |\Hess_{\phi_{t}}(\gamma_{t})|_{h} \leq \epsilon(r), \; \Pi\text{-a.e. }\gamma,  \text{ for all } t\in [0,1],
\end{equation}
up to renaming $\epsilon(r)$ with a suitable function  
$$\epsilon=\epsilon_{E}:(0,\infty)\to (0,\infty)$$ with 
$$\lim_{r\downarrow 0}\epsilon(r)=0.$$
The curve $[0,1]\ni t\mapsto \Ent(\mu_{t}|\vol_{g}) \in \R$ is smooth and, in virtue of  \eqref{eq:velocityProof1}, it  satisfies
\begin{equation}\label{eq:dtEntmut1}
\frac{d}{dt} \Ent(\mu_{t}|\vol_{g})= r\int_{M} g(\nabla_{g}^{q} \phi_{t}, \nabla_{g} \rho_{t}) d\vol_{g}= r \int_{M} \Box_{g}^{q} \phi_{t} \,d\mu_{t}, \quad  \text{for all }t\in [0,1], 
\end{equation}
where $$\rho_{t}:=\frac{d\mu_{t}}{d\vol_{g}}$$ is the density of $\mu_{t}$,  
$$ \Box_{g}^{q} \phi_{t}:={\rm div}(-\nabla_{g}^{q} \phi_{t}) $$ is the $q$-Box of  $\phi_{t}$ (the Lorentzian analog of the $q$-Laplacian),
and where we used the continuity equation
$$
\frac{d}{dt} \rho_t + r\, {\rm div}(\rho_{t} \nabla_{g}^{q} \phi_{t})=0.
$$

For what follows it is useful to consider the linearization of the $q$-Box at a smooth function  $f$, denoted by $L^{q}_{f}$ and defined by the following relation:
\begin{equation}\label{eq:linDeltaq}
\left.\frac{d}{dt}\right|_{t=0} \Box^{q}_{g}(f+tu)= L^{q}_{f}u, \quad \forall u\in C^{\infty}_{c}(M).
\end{equation} 
\\The map $[0,1]\ni t\mapsto  \int_{M} \Box_{g}^{q} \phi_{t} \,d\mu_{t} \in \R$ is smooth and, in virtue of   \eqref{eq:qHopfLaxProof1} and  \eqref{eq:velocityProof1}, it   satisfies
\begin{equation*}\label{eq:dtDeltaqphitmut1}
\frac{d}{dt} r   \int_{M} \Box_{g}^{q} \phi_{t} \,d\mu_{t}= - r^{2} \int_{M} L^{q}_{\phi}  \left( \frac{1}{q} (-g(\nabla_{g} \phi_{t}, \nabla_{g} \phi_{t}))^{q/2} \right)+ g( \nabla_{g} \Box_{g}^{q} \phi_{t}, -\nabla_{g}^{q} \phi_{t}) \,d\mu_{t},
\end{equation*}
 for every $t\in [0,1]$. 
Using the $q$-Bochner identity \eqref{eq:q-Bochner} together with  the assumption $\Ric (w,w)\leq \tilde{T}(w, w)$ for  any $w\in \cC$  and the estimates  \eqref{eq:Estphit} on $\phi_{t}$, we can rewrite the last formula as  
\begin{align*}
\frac{d}{dt}& \frac{1}{r}  \int_{M} \Box_{g}^{q} \phi_{t} \,d\mu_{t}= \nonumber\\
&\int_{M}   |g(\nabla_{g} \phi_{t}, \nabla_{g} \phi_{t})|^{q-2} \
\Big[ \Ric (\nabla_{g} \phi_{t}, \nabla_{g} \phi_{t})  +  g(\Hess_{\phi_{t}}, \Hess_{\phi_{t}})    \nonumber\\
&\qquad \qquad \qquad \qquad \qquad \quad  + \left( (q-2) \frac{\Hess_{\phi_{t}}(\nabla_{g}\phi_{t}, \nabla_{g}\phi_{t})}{|g(\nabla_{g} \phi_{t}, \nabla_{g} \phi_{t})|} \right)^{2} \nonumber\\
&\qquad \qquad \qquad \qquad \qquad \quad \left.-  2 (q-2) \frac{ \Hess_{\phi_{t}}\big((\nabla_{g}\phi_{t}, \Hess_{\phi_{t}}(\nabla_{g}\phi_{t})\big)}{|g(\nabla_{g} \phi_{t}, \nabla_{g} \phi_{t})|} \right] d\mu_{t} \nonumber \\
&\leq  \int_{M} \big( \tilde{T}(v,v) + \epsilon(r) \big) d\mu_{t}=  \tilde{T}(v,v) + \epsilon(r)  \quad  \text{for all }t\in [0,1],   
\end{align*}
up to renaming $\epsilon(r)$ with a suitable function  $$\epsilon=\epsilon_{E}:(0,\infty)\to (0,\infty)
\textrm{ with }\lim_{r\downarrow 0}\epsilon(r)=0. $$
Thus
\begin{align}
&\Ent(\mu_{1}|\vol_{g})-2 \Ent(\mu_{1/2}|\vol_{g})+\Ent(\mu_{0}|\vol_{g}) \nonumber\\
&\qquad \qquad = \int_{0}^{1/2}\left( \frac{d}{dt} \Ent(\mu_{t}|\vol_{g})|_{t=s+1/2} -  \frac{d}{dt} \Ent(\mu_{t}|\vol_{g})|_{t=s} \right)  ds \nonumber\\
&\qquad \qquad =  \int_{0}^{1/2} \int_{s}^{s+1/2}  \frac{d^{2}}{dt^{2}} \Ent(\mu_{t}|\vol_{g}) dt \,  ds  \nonumber\\
&\qquad \qquad \leq   \frac{\tilde{T}(v,v)+\epsilon(r)}{4} \, r^{2} \nonumber.
\end{align}

\textbf{(2)$ \Rightarrow$ (3)}: trivial.
\\

\textbf{(3)$\Rightarrow$ (1)}
\\ Fix $p\in (0,1)$ given by $(3)$ and assume by contradiction that there exists $x\in M$, $\epsilon>0$ and $v\in T_{x}M\cap \cC$ with $-g(v,v)=1$ such that 
$$\Ric(v,v)\geq (\tilde{T}-2\epsilon g) (v,v).$$
Then, by continuity, we can find a relatively compact neighbourhood $E\subset\subset  \Int(\cC)$ of $v$ in $TM$ such that
\begin{equation}\label{eq:RicwwE}
\Ric(w,w)\geq (\tilde{T}-\epsilon g)(w,w), \quad \forall w\in E.
\end{equation}
By Lemma \ref{lem:SmoothKP}  we can construct a  $c_{p}$-convex function  $\phi:M\to \R$ such that $\phi$ is smooth on a neighbourhood of $x$ and  $$\nabla_{g}^{q}\phi(x)=Rv.$$ For $t\in [0,1]$, define 
$$\Psi_{1/2}^{t}(z):=\exp_{z}^{g}(2r(t-1/2) \nabla_{g}^{q}\phi(z)).$$
By continuity, for $r\in (0, R)$ small enough, we have that  
\begin{equation}\label{dtPsi12tr2}
\frac{1}{r}\frac{d}{dt} \Psi_{1/2}^{t}(z)\in E, \quad {\rm d}^{TM}_{h} \left(\frac{d}{dt} \Psi_{1/2}^{t}(z), r v\right)\leq \epsilon r,   \quad  \forall z\in B^{g,E}_{r^{4}}(x), \quad \forall t\in [0,1].
\end{equation}
  Moreover $\Psi_{1/2}^{1} (B^{g,E}_{r^{4}}(x))\subset B^{g,E}_{r^{2}}(y)$.
\\Set $\mu_{1/2}:= \vol_{g}(B^{g,E}_{r^{4}}(x))^{-1} \, \vol_{g}\llcorner B^{g,E}_{r^{4}}(x)$ and consider $\mu_{t}:= (\Psi_{1/2}^{t})_{\sharp} \mu_{1/2}$. Notice that   $$\supp(\mu_{1})\subset  B^{g,E}_{r^{2}}(y)\subset  \{\exp_{y}^{g}(r^{2}w): w\in T_{y}M\cap \cC, \, g(w,w)=-1 \}.$$
By the above construction, we get that $(\mu_{t})_{t\in [0,1]}$ can be represented by a regular $c_{p}$-optimal  dynamical plan $\Pi$ such that $$\supp((\partial \ee)_{\sharp}\Pi )\subset E.$$  Therefore \eqref{eq:RicwwE} together with Theorem \ref{thm:RiccigeqT} yields
\begin{align}
\Ent(\mu_{1/2}|\vol_{g})&\leq \frac{1}{2} \Ent(\mu_0|\vol_{g})+ \frac{1}{2} \Ent(\mu_1|\vol_{g}) \nonumber\\
&\qquad  - \int  \int_{0}^{1} \rG(s,1/2) (\tilde{T}-\epsilon g)(\dot{\gamma}_{s}, \dot{\gamma}_{s}) ds\, d\Pi(\gamma), \nonumber\\
& \leq  \frac{1}{2} \Ent(\mu_0|\vol_{g})+\frac{1}{2} \Ent(\mu_1|\vol_{g}) - \frac{r^{2} (\tilde{T}(v,v)-\epsilon g(v,v) +C\epsilon r)}{8}  \label{eq:EntConvContrK-eps},
\end{align}
where in the second inequality we used \eqref{dtPsi12tr2} and that $C>0$ is a constant independent of $r$ and $\epsilon$. Note that $\epsilon>0$ in \eqref{eq:EntConvContrK-eps} is fixed  independently of $r>0$. Clearly \eqref{eq:EntConvContrK-eps} contradicts the existence of $\epsilon_{E}(r)\to 0$ as $r\to 0$ so that \eqref{eq:ChUBtT} holds.  
\end{proof}

In the vacuum case when $T\equiv 0$, the inequality $\Ric\leq \tilde{T}$  with $\tilde{T}$ as in \eqref{eq:tildeTgT} reads as 
$$\Ric\leq Kg \textrm{ with }K=\frac{2\Lambda}{n-2} \in \R.
$$  Note that for $v\in \cC$ it holds $g(v,v)\leq 0$ so, when comparing the next result with its Riemannian counterpart \cite{StUB}, the sign of the lower bound $K$ is reversed.

\begin{corollary}\label{thm:RicciUpperBound}
Let $(M,g,\cC)$ be a space-time of dimension $n \geq 2$ and let $K\in \R$. Then the following assertions are equivalent:
\begin{enumerate}
\item[(1)]  $\Ric(v,v)\leq K g(v,v)$ for every causal vector $v\in \cC$.
\item[(2)] For every $p\in (0,1)$ and for every relatively  compact  open subset $E\subset\subset  \Int(\cC)$ there exist $R=R(E)\in (0,1)$ and a function 
\begin{align*}
&\epsilon=\epsilon_{E}:(0,\infty)\to (0,\infty)\textrm{ with }\lim_{r\downarrow 0}\epsilon(r)=0 \textrm{ such that }\\
&\forall x\in p_{TM\to M}(E) \textrm{ and }v\in T_{x}M\cap E \textrm{ with }g(v,v)=-R^{2}
\end{align*}
 the next assertion holds.
For every $r\in (0,R)$,  there exists an $r$-concentrated regular $c_{p}$-optimal  dynamical plan $\Pi=\Pi(x,v,r)$ in the direction of $v$ (with respect to $E$) 
which has $-K$-concave entropy in the sense that
\begin{equation}\label{eq:ChUB}
\frac{4}{r^{2}}\left[\Ent(\mu_{1}|\vol_{g})-2 \Ent(\mu_{1/2}|\vol_{g})+\Ent(\mu_{0}|\vol_{g}) \right] \leq  -K+\epsilon(r).
\end{equation}
\item [(3)] There exists  $p\in (0,1)$ such that the analogous  assertion as in  {\rm (2)} holds true.
\end{enumerate}
\end{corollary}

\subsection{Optimal transport formulation of the Einstein  equations}

Recall that Einstein equations of general relativity,  with cosmological constant equal to $\Lambda \in \R$ and energy-momentum tensor $T$, read as
\begin{equation}\label{eq:EFERicT2}
\Ric=\frac{2\Lambda}{n-2} g + 8 \pi T - \frac{8\pi} {n-2} \Tr_{g}(T) \, g,
\end{equation}
 for an $n$-dimensional space-time $(M,g,\cC)$.  Combining  Theorem  \ref{thm:RiccigeqT}  with Theorem  \ref{thm:RiccileqT}, both with the choice 
 $$\tilde{T}=\frac{2\Lambda}{n-2} g + 8 \pi T - \frac{8\pi} {n-2} \Tr_{g}(T) \, g,$$ we obtain the following optimal transport formulation of  \eqref{eq:EFERicT2}.

\begin{theorem}\label{thm:RiccieqT}
Let $(M,g,\cC)$ be a space-time of dimension $n \geq 2$ and let $\tilde{T}$ be a quadratic form on $M$. Then the following assertions are equivalent:
\begin{enumerate}
\item[(1)]  $\Ric(v,v)=\tilde{T}(v,v)$ for every $v\in T_{x}M$
\item[(2)]  $\Ric(v,v)=\tilde{T}(v,v)$ for every causal vector $v\in \cC$.
\item[(3)]  For every $p\in (0,1)$ and for every relatively  compact  open subset $E\subset\subset  \Int(\cC)$ there exist $R=R(E)\in (0,1)$ and a function 
\begin{align*}
&\epsilon=\epsilon_{E}:(0,\infty)\to (0,\infty)\textrm{ with }\lim_{r\downarrow 0}\epsilon(r)=0 \textrm{ such that }\\
&\forall x\in p_{TM\to M}(E) \textrm{ and }v\in T_{x}M\cap E \textrm{ with }g(v,v)=-R^{2}
\end{align*}
 the next assertion holds.
For every $r\in (0,R)$,  there exists an $r$-concentrated regular $c_{p}$-optimal  dynamical plan $\Pi=\Pi(x,v,r)$ in the direction of $v$ (with respect to $E$)
 satisfying
\begin{equation}\label{eq:RiceqT}
\tilde{T}(v,v)-\epsilon(r)\leq \frac{4}{r^{2}}\left[\Ent(\mu_{1}|\vol_{g})-2 \Ent(\mu_{1/2}|\vol_{g})+\Ent(\mu_{0}|\vol_{g}) \right] \leq \tilde{T}(v,v)+\epsilon(r).
\end{equation}
\item [(4)] There exists  $p\in (0,1)$ such that the analogous  assertion as in  {\rm (3)}  holds true.
\end{enumerate}
\end{theorem}

\begin{remark}[$\mu_{1/2}$ can be chosen more general]\label{rem:Thmmu12}
From the proof of Theorem \ref{thm:RiccieqT} it follows that one can replace (2) (and analogously (3)) with the following (a priori stronger, but a fortiori equivalent) statement. For every $p\in (0,1)$ the following holds.  For every relatively  compact  open subset $E\subset\subset  \Int(\cC)$ there exist $R=R(E)\in (0,1)$ and a function 
$$\epsilon=\epsilon_{E}:(0,\infty)\to (0,\infty) \textrm{ with }\lim_{r\downarrow 0}\epsilon(r)=0$$ 
such that 
$$
\forall x\in p_{TM\to M}(E), \forall v\in T_{x}M\cap E \textrm{ with }g(v,v)=-R^{2}$$  the next assertion holds. For every $r\in (0,R)$ and every 
$$
\mu_{1/2}\in \cP(M) \textrm{ with }\mu_{1/2}\ll \vol_{g} \textrm{ and }\supp (\mu_{1/2})\subset B^{g,E}_{r^{4}}(x),
$$
 setting $y=\exp^{g}_{x}(rv)$, there exists a regular $c_{p}$-optimal  dynamical plan $\Pi=\Pi(\mu_{1/2},v,r)$ with associated curve of  probability measures 
$$(\mu_{t}:=(\ee_{t})_{\sharp}\Pi)_{t\in [0,1]}\subset \cP(M)$$ 
such that 
$$\supp(\mu_{1})\subset \{\exp_{y}^{g}(r^{2}w): w\in T_{y}M\cap \cC, \, g(w,w)=-1 \}$$ and satisfying \eqref{eq:RiceqT}.
\end{remark}

\begin{remark}[An equivalent statement via an auxiliary Riemannian metric $h$]\label{rem:ThmRiemMeth}
Given an auxiliary Riemannian metric $h$ on $M$,  with analogous arguments as in the proof below, the condition $(3)$ in Theorem \ref{thm:RiccieqT}  can be replaced by 
\begin{enumerate}
\item[(3)']  For every $p\in (0,1)$ the following holds.  For every $x\in M$ there exist $R=R(x)>0$ and  a function 
\begin{align*}
\epsilon=\epsilon_{x}:(0,\infty)\to (0,\infty) \textrm{ with }\lim_{r\downarrow 0}\epsilon(r)=0 \text{ such that}\\ 
\forall v\in \Int(\cC_{x}) \text{ with } h(v,v)\leq R^{2} \text{ the next assertion holds.}
\end{align*}
For every $r\in (0,R)$,   there exists an $r$-concentrated regular $c_{p}$-optimal  dynamical plan $\Pi=\Pi(x,v,r)$ in the direction of $v$ 
 satisfying \eqref{eq:RiceqT}.
\end{enumerate}
Moreover, both in (3) and (3') one can replace $B^{g,E}_{r^{4}}(x)$ by $B^{h}_{r^{4}}(x)$
and  $$\{\exp_{y}^{g}(r^{2}w)\,:\, w\in T_{y}M\cap \cC, \, g(w,w)=-1 \}$$ by  $B^{h}_{r^{2}}(y)$.
\end{remark} 

\begin{remark}[The tensor $\tilde{T}$]
As mentioned above, we will assume the cosmological constant $\Lambda$ and the energy momentum tensor $T$ to be given, say from physics and/or mathematical general relativity.  Given $g,\Lambda$ and $T$, for convenience of notation we set $\tilde{T}$ to be defined in \eqref{eq:deftildeTIntro}.
\\Let us stress that not any symmetric bilinear form $\tilde{T}$ would correspond to a physically meaningful situation; in order to be physically relevant,  it is crucial that $\tilde{T}$ is given by  \eqref{eq:deftildeTIntro}  where $T$ is a physical energy-momentum tensor (in particular $T$ has to satisfy $\nabla^{a} T_{ab}=0$, i.e.  be ``freely gravitating'', it has to satisfy some suitable energy condition like the ``dominant energy condition'', etc.). 
\end{remark}

\medskip

\begin{proof}[Proof of Theorem \ref{thm:RiccieqT}]
\textbf{(1)$ \Rightarrow$ (2)}: trivial.
\\

\textbf{(2)$ \Rightarrow$ (1)}: Follows by the identity theorem for polynomials and the fact that $\cC$ has non-empty open interior.

\medskip
\textbf{(2)$ \Rightarrow$ (3)}: From the implication (1)$ \Rightarrow$ (2) in Theorem   \ref{thm:RiccileqT}, we get a regular $c_{p}$-optimal  dynamical plan $\Pi=\Pi(x,v,r)$ as in (3) such that the upper bound in \eqref{eq:RiceqT} holds.  Moreover, from \eqref{eq:Estphit}
it holds that
\begin{equation}\label{eq:gammaPiv}
{\rm d}^{TM}_{h} (\dot{\gamma}_{t}, rv)\leq r \epsilon(r), \quad \text{for $\Pi$-a.e. } \gamma,   \text{ for all } t\in [0,1].
\end{equation}
Recalling that the implication  (1)$ \Rightarrow$ (2) in Theorem \ref{thm:RiccigeqT} gives the convexity property \eqref{eq:EntConvtT} of the entropy along  \emph{every} regular $c_{p}$-optimal  dynamical plan, and using \eqref{eq:gammaPiv}, we conclude that also the lower bound in \eqref{eq:RiceqT} holds.
\\

\textbf{(3)$ \Rightarrow$ (4)}: trivial.
\\

\textbf{(4)$ \Rightarrow$ (2)}. 
\\The fact that  $$\Ric(v,v)\leq \tilde{T}(v,v) \quad \forall v\in \cC$$ follows directly from Theorem  \ref{thm:RiccileqT}. The fact that $$\Ric(v,v)\geq \tilde{T}(v,v) \quad \forall v\in \cC$$  can be showed following arguments already used in the paper, let us briefly discuss it. 
 Fix $p\in (0,1)$ given by $(4)$ and assume by contradiction that 
 $$
 \exists x\in M, \, \delta>0 \textrm{ and }v\in T_{x}M\cap \cC \textrm{ with }-g(v,v)>0$$
  such that $$\Ric(v,v)\leq (\tilde{T}+3\delta g) (v,v) .$$ Thanks to Lemma \ref{lem:SmoothKP}, up to replacing $v$ with $sv$ for some $s\in (0,1)$ small enough, we know that we can  construct a  $c_{p}$-convex function $\phi:M\to \R$, smooth in a neighbourhood of $x$ and satisfying 
  $$\nabla^{q}_{g}\phi(x)=v \qquad \qquad \Hess_\phi (x)=0.$$
 Then, by continuity, we can find
\begin{itemize}
\item  a relatively compact open neighbourhood $E\subset\subset  \Int(\cC)$ of $v$  in $TM$ such that
\begin{equation}\label{eq:RicwwEleq}
\Ric(w,w)\leq  (\tilde{T}+2\delta g) (w,w), \quad \forall w\in E;
\end{equation}
\item $\phi$ is smooth on  $B^{h}_{100r}(x)$, 
$$\nabla^{q}_{g}\phi\in E \textrm{ on }B^{h}_{10r}(x),$$ 
$${\rm d}^{TM}_{h}(\nabla^{q}_{g}\phi,v) \leq   \epsilon(r) \textrm{ on }B^{h}_{10r}(x);$$
\item $|\Hess_\phi|_{h} \leq  \epsilon(r)$ on $B^{h}_{10r}(x)$; 
\end{itemize}
where 
$$\epsilon(r)\to 0 \textrm{ as } r\to 0.
$$
For $t\in [0,1]$, consider the map 
$$\Psi_{1/2}^{t}:z\mapsto \exp_{z}(r(t-1/2) \nabla^{q}_{g}\phi(z)).$$ 
Notice that 
$$\Psi_{1/2}^{t} ( B^{g,E}_{r^{4}}(x))\subset B^{h}_{10r}(x) \quad \forall t\in [0,1].
$$
Call $$\mu_{1/2}=\vol_{g}(B^{g,E}_{r^{4}}(x))^{-1}\,  \vol_{g}\llcorner B^{g,E}_{r^{4}}(x)$$  and define 
$$\mu_{t}:=(\Psi_{1/2}^{t})_{\sharp} (\mu_{1/2}) \textrm{ for } t\in [0,1].
$$
 By the properties of $\phi$,  the plan $\Pi$ representing the curve of probability measures $(\mu_{t})_{t\in [0,1]}$ is a regular  $c_{p}$-optimal  dynamical plan and  
 $$\supp(\mu_{1})\subset \{\exp_{y}^{g}(r^{2}w): w\in T_{y}M\cap \cC, \, g(w,w)=-1 \}.
 $$
  Moreover
\begin{equation}\label{eq:dtPsiContrEE}
\frac{1}{r}\frac{d}{dt} \Psi_{1/2}^{t}(z)\in E  \quad  \forall z\in B^{h}_{r}(x), \quad \forall t\in [0,1].
\end{equation}
We can now follow verbatim the arguments in (1)$\Rightarrow$(2) of Theorem \ref{thm:RiccileqT} by using   \eqref{eq:RicwwEleq} and \eqref{eq:dtPsiContrEE}, obtaining  a function $\epsilon(r)\to 0$ as $r\to 0$ such that
\begin{equation}
\Ent(\mu_{1}|\vol_{g})-2 \Ent(\mu_{1/2}|\vol_{g})+\Ent(\mu_{0}|\vol_{g}) \leq  \frac{ (\tilde{T}+\delta g) (v,v)  +\epsilon(r)}{4} \, r^{2}.  \nonumber
\end{equation}
The last inequality clearly contradicts the lower bound in \eqref{eq:RiceqT}.

\end{proof}

In the vacuum case $T\equiv 0$ with cosmological constant $\Lambda \in \R$, the Einstein equations read as
\begin{equation}\label{eq:EELambda}
\Ric \equiv \frac{\Lambda}{\frac{n}{2}-1} g,
\end{equation}
 for an $n$-dimensional space-time $(M,g,\cC)$. Specializing Theorem \ref{thm:RiccieqT}  with the choice 
 $$\tilde{T}=\frac{\Lambda}{\frac{n}{2}-1} g$$
  and using Corollary \ref{thm:RicciLowerBound} to sharpen the lower bound in \eqref{eq:RicLambdaOT} for the constant case, we obtain the following optimal transport formulation of Einstein vacuum equations.

\begin{theorem}\label{thm:RicciConstLambda}
Let $(M,g,\cC)$ be a space-time of dimension $n \geq 3$ and let $\Lambda\in \R$. Then the following assertions are equivalent:
\begin{enumerate}
\item[(1)] The space-time $(M,g,\cC)$ satisfies Einstein vacuum equations  of General Relativity \eqref{eq:EELambda} corresponding to cosmological constant equal to $\Lambda$.   
\item[(2)]  For every $p\in (0,1)$ and for every relatively  compact  open subset $E\subset\subset  \Int(\cC)$ there exist $R=R(E)\in (0,1)$ and a function 
\begin{align*}
&\epsilon=\epsilon_{E}:(0,\infty)\to (0,\infty)\textrm{ with }\lim_{r\downarrow 0}\epsilon(r)=0 \textrm{ such that }\\
&\forall x\in p_{TM\to M}(E) \textrm{ and }v\in T_{x}M\cap E \textrm{ with }g(v,v)=-R^{2}
\end{align*}
 the next assertion holds.
For every $r\in (0,R)$,  there exists an $r$-concentrated regular $c_{p}$-optimal  dynamical plan $\Pi=\Pi(x,v,r)$ in the direction of $v$ (with respect to $E$)
 satisfying
\begin{equation}\label{eq:RicLambdaOT}
-\frac{\Lambda}{\frac{n}{2}-1} \leq \frac{4}{r^{2}}\left[\Ent(\mu_{1}|\vol_{g})-2 \Ent(\mu_{1/2}|\vol_{g})+\Ent(\mu_{0}|\vol_{g}) \right] \leq - \frac{\Lambda}{\frac{n}{2}-1}+\epsilon(r).
\end{equation}
\item [(3)] There exists  $p\in (0,1)$ such that the analogous  assertion as in  {\rm (2)}  holds true.
\end{enumerate}
\end{theorem}

\medskip

It is worth to isolate the case of zero cosmological constant.

\begin{corollary}\label{cor:Ricciflat}
Let $(M,g,\cC)$ be a space-time of dimension $n \geq 2$. Then the following assertions are equivalent:
\begin{enumerate}
\item[(1)] The space-time $(M,g,\cC)$ satisfies Einstein vacuum equations of General Relativity with zero cosmological constant, i.e.  $\Ric\equiv 0$.  
\item[(2)]  For every $p\in (0,1)$ and for every relatively  compact  open subset $E\subset\subset  \Int(\cC)$ there exist $R=R(E)\in (0,1)$ and a function 
\begin{align*}
&\epsilon=\epsilon_{E}:(0,\infty)\to (0,\infty)\textrm{ with }\lim_{r\downarrow 0}\epsilon(r)/r^2=0 \textrm{ such that }\\
&\forall x\in p_{TM\to M}(E) \textrm{ and }v\in T_{x}M\cap E \textrm{ with }g(v,v)=-R^{2}
\end{align*}
 the next assertion holds.
For every $r\in (0,R)$,  there exists an $r$-concentrated regular $c_{p}$-optimal  dynamical plan $\Pi=\Pi(x,v,r)$ in the direction of $v$ (with respect to $E$)
 satisfying
\begin{equation}\label{eq:Ric0OT}
0\leq \Ent(\mu_{1}|\vol_{g})-2 \Ent(\mu_{1/2}|\vol_{g})+\Ent(\mu_{0}|\vol_{g}) \leq  \epsilon(r).
\end{equation}
\item [(3)] There exists  $p\in (0,1)$ such that the analogous  assertion as in  {\rm (2)}  holds true.
\end{enumerate}
\end{corollary}

\begin{appendix}
\section{A $q$-Bochner identity in Lorentzian setting}

In this section we prove a Bochner type identity in Lorentzian setting for the linearization of the $q$-Box operator, the Lorentzian analog of the $q$-Laplacian; let us mention that related results have been obtained in the Riemannian \cite{Matei00, Valt12}   and  Finsler settings \cite{Xia1,Xia2} but  at best our knowledge this section is original in the Lorentzian $\cL_{p}$ framework. 

Throughout the section,  $(M,g,\cC)$ is a space-time,  $U\subset M$ is an open subset and  $q\in (-\infty, 0)$ is fixed.
Let  $\phi\in C^{3}(U)$ satisfy 
$$-\nabla_{g} \phi\in \Int(\cC) \textrm{ on }U.$$ 
Denote by  
$$\nabla^{q}_{g} \phi:=- |g(\nabla_{g} \phi, \nabla_{g} \phi)|^{\frac{q-2}{2}} \nabla_{g} \phi$$ 
the $q$-gradient, by  
$$ \Box_{g}^{q} \phi:={\rm div}(- \nabla_{g}^{q} \phi) $$ 
the $q$-Box operator of  $\phi$ and by $L^{q}_{\phi}$ 
the linearization of the $q$-Box operator at $\phi$ defined by the following relation:
\begin{equation}\label{eq:linDeltaqDef}
\left.\frac{d}{dt}\right|_{t=0} \Box^{q}_{g}(\phi+tu)= L^{q}_{\phi}u, \quad \forall u\in C^{\infty}_{c}(M).
\end{equation}

The ultimate goal of the section is to prove the following result.

\begin{proposition}\label{prop:qBochner}
Under the above notation, the following  $q$-Bochner identity holds:
 
 \begin{eqnarray} \label{eq:q-Bochner}
 &&-L^{q}_{\phi}\left(\frac{(-g(\nabla_{g} \phi, \nabla_{g} \phi))^{ \frac{q} {2}}}{q} \right)= g( \nabla_{g} \Box_{g}^{q} \phi,- \nabla_{g}^{q} \phi)  \nonumber\\
 &&\qquad \qquad \qquad  + (q-2)^{2}  |g(\nabla_{g} \phi, \nabla_{g} \phi)|^{q-2} \left( \frac{\Hess_{\phi}(\nabla_{g}\phi, \nabla_{g}\phi)}{|g(\nabla_{g} \phi, \nabla_{g} \phi)|} \right)^{2} \nonumber\\
&& \qquad \qquad \qquad  +   |g(\nabla_{g} \phi, \nabla_{g} \phi)|^{q-2} \big( \Ric (\nabla_{g} \phi, \nabla_{g} \phi)  +  g(\Hess_{\phi}, \Hess_{\phi}) \big)  \nonumber\\
&&\qquad \qquad \qquad  - 2 (q-2)   \, |g(\nabla_{g} \phi, \nabla_{g} \phi)|^{q-3}  \Hess_{\phi}\left(\nabla_{g}\phi, \Hess_{\phi}(\nabla_{g}\phi)\right) .
\end{eqnarray}
\end{proposition}

The proof of Proposition \ref{prop:qBochner} requires some preliminary lemmas.  First of all we derive an explicit expression for the operator  $L^{q}_{\phi}$.

\begin{lemma}\label{lem:cLphiu}
Under the above notation, it holds

\begin{align}
L^{q}_{\phi}u=&   \big(-g(\nabla_{g}\phi, \nabla_{g} \phi) \big)^{\frac{q-2}{2}} \left( \Box_{g}u - (q-2)  \frac{ \Hess_{u}(\nabla_{g}\phi, \nabla_{g} \phi)}{ -g(\nabla_{g} \phi, \nabla_{g} \phi)} \right) \nonumber\\
& + (2-q)   \big(-g(\nabla_{g}\phi, \nabla_{g} \phi) \big)^{-1} g(\nabla_{g}\phi, \nabla_{g} u) \Box_{g}^{q} \phi  \nonumber\\
&+ 2 (2-q)   \,\big(-g(\nabla_{g} \phi, \nabla_{g} \phi) \big)^{\frac{q-4}{2}}  \Hess_{\phi}\left(\nabla_{g}\phi, \nabla_{g}u+ \frac{g(\nabla_{g}\phi, \nabla_{g} u)}{-g(\nabla_{g}\phi, \nabla_{g} \phi)} \nabla_{g}\phi  \right) \label{eq:Lqphi}.
\end{align}
\end{lemma}

\begin{proof}
By the very definitions of  $L^{q}_{\phi}u$ and $ \Box_{g}^{q} \phi$, we have
\begin{align}
L^{q}_{\phi}u&={\rm div}\left( \left.\frac{d}{dt}\right|_{t=0}  (-g(\nabla_{g}(\phi+tu), \nabla_{g} (\phi+tu)))^{\frac{q-2}{2}}   \nabla_{g} (\phi+tu) \right) \nonumber \\
&= {\rm div} \left( (2-q)   \big(-g(\nabla_{g}\phi, \nabla_{g} \phi) \big)^{\frac{q-4}{2}} g(\nabla_{g}\phi, \nabla_{g} u) \, \nabla_{g}\phi   \right. \nonumber \\
&\qquad \qquad \left.  + \big(-g(\nabla_{g}\phi, \nabla_{g} \phi) \big)^{\frac{q-2}{2}}  \, \nabla_{g}u  \right). \label{eq:Lqphi1}
\end{align}
In order to explicit the last formula, compute
\begin{align} 
\nabla_{g}  \big(-g(\nabla_{g} \phi, \nabla_{g} \phi)\big)^{\alpha} &= -2 \alpha \,\big(-g(\nabla_{g} \phi, \nabla_{g} \phi) \big)^{\alpha-1}    \Hess_{\phi}(\nabla_{g}\phi)  \label{eq:nablanablaphialpha} \\
\nabla_{g}  \big(g(\nabla_{g} \phi, \nabla_{g} u)\big) &=  \Hess_{u}(\nabla_{g}\phi)+\Hess_{\phi}(\nabla_{g}u)  \label{eq:nablanablaphinablau}.
\end{align}
Plugging  \eqref{eq:nablanablaphialpha} and  \eqref{eq:nablanablaphinablau} into  \eqref{eq:Lqphi1} gives  \eqref{eq:Lqphi}.
\end{proof}

We next show a $q$-Bochner identity for the operator $A^{q}_{\phi}$ defined as 
\begin{equation}\label{eq:defLqphi}
A^{q}_{\phi}(u):= \big(-g(\nabla_{g} \phi, \nabla_{g} \phi) \big)^{\frac{q-2}{2}} \left( \Box_{g} u-(q-2) \frac{\Hess_{u}(\nabla_{g}\phi, \nabla_{g}\phi)}{-g(\nabla_{g}\phi, \nabla_{g}\phi)} \right).
\end{equation}

\begin{lemma}\label{lem:qBochner2ord}
Under the above notation, the following identity holds:
\begin{align}
 &-L_{\phi}^{q}\left(\frac{(-g(\nabla_{g} \phi, \nabla_{g} \phi))^{\frac{q}{2}} }{q} \right)+ g( \nabla_{g} \Box_{g}^{q} \phi, \nabla_{g}^{q} \phi) = \nonumber \\ 
&\qquad  = (q-2) |g(\nabla_{g} \phi, \nabla_{g} \phi)|^{\frac{q-4}{2}}  \Box^{q}_{g}\phi\, \Hess_{\phi}(\nabla_{g}\phi, \nabla_{g}\phi)   \nonumber \\ 
& \qquad \qquad + |g(\nabla_{g} \phi, \nabla_{g} \phi)|^{q-2} \Big(q(q-2) \left( \frac{\Hess_{\phi}(\nabla_{g}\phi, \nabla_{g}\phi)}{|g(\nabla_{g} \phi, \nabla_{g} \phi)|} \right)^{2}      \nonumber \\ 
&  \qquad \qquad \qquad \qquad \qquad \qquad \quad  +  \Ric (\nabla_{g} \phi, \nabla_{g} \phi)+  g(\Hess_{\phi}, \Hess_{\phi})  \Big).
 \label{eq:q-Bochner2ord}
\end{align}
\end{lemma}

\begin{proof}
We perform the computation at an arbitrary point $x_{0}\in U$. In order to simplify the computations, we consider normal coordinates $(x^{i})$ in a neighbourhood of $x_{0}$ with $\frac{\partial}{\partial x^{1}}\in \cC$. \
It holds
\begin{align}
&-\frac{\Box_{g}(-g(\nabla_{g} \phi, \nabla_{g} \phi))^{ \frac{q} {2}}}{q}=g^{ij} \partial_{i}  \left( (-g(\nabla_{g} \phi, \nabla_{g} \phi))^{\frac{q-2}{2}} g^{kl} \partial_{jk}\phi \, \partial_{l} \phi  \right) \nonumber\\
&\qquad \quad =  (-g(\nabla_{g} \phi, \nabla_{g} \phi))^{\frac{q-2}{2}} \Big(- \frac{q-2}{ -g(\nabla_{g} \phi, \nabla_{g} \phi)} g^{ij} g^{mn} \partial_{im} \phi \partial_{n} \phi   g^{kl} \partial_{jk}\phi \, \partial_{l} \phi    \nonumber\\
& \qquad \qquad \qquad \qquad \qquad \qquad \qquad \qquad \qquad  + g^{ij} g^{kl}   \partial_{ijk}\phi \, \partial_{l} \phi + g^{ij} g^{kl} \partial_{jk}\phi \, \partial_{li} \phi  \Big). \label{Delta|nablaphi|1}
\end{align}
Now, from the symmetry of second order derivatives and the very definition of the the Riemann tensor \eqref{def:R}, we have
\begin{equation}\label{eq:3rdOrder}
\partial_{ijk}\phi= \partial_{ikj}\phi = g(R(\partial_{x^{i}}, \partial_{x^{k}}) \nabla_{g} \phi,  \partial_{x^{j}}) +  \partial_{kij}\phi.
\end{equation}
Thus 
$$
g^{ij} g^{kl}   \partial_{ijk}\phi \, \partial_{l} \phi = g(\nabla_{g} \Box_{g} \phi, \nabla_{g}\phi) + \Ric(\nabla_{g} \phi, \nabla_{g} \phi),
$$
and we can rewrite \eqref{Delta|nablaphi|1} as
\begin{align} \label{Delta|nablaphi}
-\frac{\Box_{g}(-g(\nabla_{g} \phi, \nabla_{g} \phi))^{ \frac{q} {2}}}{q}=&(-g(\nabla_{g} \phi, \nabla_{g} \phi))^{\frac{q-2}{2}} \Big((2-q)  \frac{ g(\Hess_{\phi} (\nabla_{g} \phi), \Hess_{\phi} (\nabla_{g} \phi))}{ -g(\nabla_{g} \phi, \nabla_{g} \phi)}   \nonumber\\
& \qquad  \qquad   \qquad  \qquad  \quad   +  g(\nabla_{g} \Box_{g} \phi, \nabla_{g}\phi) + \Ric(\nabla_{g} \phi, \nabla_{g} \phi)   \nonumber\\
&\qquad  \qquad   \qquad  \qquad  \quad + g(\Hess_{\phi}, \Hess_{\phi}) \Big).
\end{align}
We now compute the second part of $-L_{\phi}^{q}\left(\frac{(-g(\nabla_{g} \phi, \nabla_{g} \phi))^{\frac{q}{2}} }{q} \right)$. To this aim observe that
\begin{align}
\nabla_{g} \left( \frac{\big(-g(\nabla_{g} \phi, \nabla_{g} \phi) \big)^{\frac{q}{2} }}{q} \right)&= - \big(-g(\nabla_{g} \phi, \nabla_{g} \phi) \big)^{\frac{q-2}{2}} \, \Hess_{\phi}(\nabla_{g}\phi)  \label{eq:nablanablanabla}   \\
\Hess_{\frac{(-g(\nabla_{g} \phi, \nabla_{g} \phi))^{\frac{q}{2} }}{q}} (\nabla_{g}\phi, \nabla_{g}\phi)&= (q-2) \big(-g(\nabla_{g} \phi, \nabla_{g} \phi) \big)^{\frac{q-4}{2}} [\Hess_{\phi} (\nabla_{g} \phi, \nabla_{g} \phi)]^{2}  \nonumber\\
&- \big(-g(\nabla_{g} \phi, \nabla_{g} \phi) \big)^{\frac{q-2}{2}} g(\nabla_{ \nabla_{g} \phi}\nabla_{ \nabla_{g} \phi} \nabla_{g} \phi, \nabla_{g}\phi). \label{eq:Hessnablanabla}
\end{align}
It is useful to express $\Box^{q}_{g}$ in terms of  $\Box_{g}$:
\begin{align}
\Box^{q}_{g} \phi&:=  {\rm div}\Big( \big(-g(\nabla_{g} \phi, \nabla_{g} \phi) \big)^{\frac{q-2}{2}} \nabla_{g}\phi \Big) \nonumber\\
&= \big(-g(\nabla_{g} \phi, \nabla_{g} \phi) \big)^{\frac{q-2}{2}} \left( \Box_{g} \phi-(q-2) \frac{\Hess_{\phi}(\nabla_{g}\phi, \nabla_{g}\phi)}{-g(\nabla_{g}\phi, \nabla_{g}\phi)} \right). \label{eq:DeltaqDelta}
\end{align}
Using \eqref{eq:DeltaqDelta}, we can write
\begin{align} \label{eq:nablaDeltapnabla}
g(&\nabla_{g}  \Box^{q}_{g} \phi,  \nabla_{g}\phi)=  \nonumber\\
&= g\left(\nabla_{g} \Big(\big(-g(\nabla_{g} \phi, \nabla_{g} \phi) \big)^{\frac{q-2}{2}} \left( \Box_{g} \phi-(q-2) \frac{\Hess_{\phi}(\nabla_{g}\phi, \nabla_{g}\phi)}{-g(\nabla_{g}\phi, \nabla_{g}\phi)} \right)\Big) , \nabla_{g}\phi \right)  \nonumber\\
&=\big(-g(\nabla_{g} \phi, \nabla_{g} \phi) \big)^{\frac{q-2}{2}}  g(\nabla_{g} \Box_{g} \phi, \nabla_{g}\phi) -(q-2) \frac{\Hess_{\phi}(\nabla_{g} \phi, \nabla_{g} \phi )}{-g(\nabla_{g} \phi, \nabla_{g} \phi )}  \Box^{q}_{g} \phi    \nonumber\\
&\qquad - (q-2) \big(-g(\nabla_{g} \phi, \nabla_{g} \phi) \big)^{\frac{q-4}{2}}  \Big( g(\nabla_{\nabla_{g}\phi}\nabla_{\nabla_{g}\phi} \nabla_{g}\phi ,\nabla_{g}\phi)   \nonumber\\
&\qquad \qquad \qquad \qquad \qquad \qquad \qquad \qquad + g(\Hess_{\phi} (\nabla_{g}\phi), \Hess_{\phi} (\nabla_{g}\phi) \Big) \nonumber\\
&\qquad  -2(q-2)  \big(-g(\nabla_{g} \phi, \nabla_{g} \phi) \big)^{\frac{q-6}{2}}  [\Hess_{\phi} (\nabla_{g} \phi, \nabla_{g} \phi)]^{2}.
\end{align}
Plugging \eqref{Delta|nablaphi} and \eqref{eq:Hessnablanabla} into \eqref{eq:defLqphi}, and simplifying using  \eqref{eq:nablaDeltapnabla}, gives the desired  \eqref{eq:q-Bochner2ord}.
\end{proof}

\textbf{Proof of Proposition \ref{prop:qBochner}}
Combining the expression of $L^{q}_{\phi}u$ as in  \eqref{eq:Lqphi} with the definition of $A^{q}_{\phi}u$ as in \eqref{eq:defLqphi} we can write
\begin{align}
L^{q}_{\phi}u=&   A^{q}_{\phi}u + (2-q)   \big(-g(\nabla_{g}\phi, \nabla_{g} \phi) \big)^{-1} g(\nabla_{g}\phi, \nabla_{g} u) \Box_{g}^{q} \phi  \nonumber\\
&+ 2 (2-q)   \,\big(-g(\nabla_{g} \phi, \nabla_{g} \phi) \big)^{\frac{q-4}{2}}  \Hess_{\phi}\left(\nabla_{g}\phi, \nabla_{g}u+ \frac{g(\nabla_{g}\phi, \nabla_{g} u)}{-g(\nabla_{g}\phi, \nabla_{g} \phi)} \nabla_{g}\phi  \right) \label{eq:LqphiuLq}.
\end{align}
Specializing \eqref{eq:LqphiuLq} with $u=\frac{(-g(\nabla_{g} \phi, \nabla_{g} \phi))^{ \frac{q} {2}}}{q}$ and using  \eqref{eq:nablanablanabla} gives
\begin{align}\label{eq:Lqphinablaphi2}
&   -L^{q}_{\phi}\left(\frac{(-g(\nabla_{g} \phi, \nabla_{g} \phi))^{ \frac{q} {2}}}{q} \right)=  -A^{q}_{\phi}\left(\frac{(-g(\nabla_{g} \phi, \nabla_{g} \phi))^{ \frac{q} {2}}}{q} \right)  \nonumber\\
 &  \qquad \qquad  \qquad  \qquad - (q-2) \big(-g(\nabla_{g} \phi, \nabla_{g} \phi) \big)^{\frac{q-2}{2}}  \frac{\Hess_{\phi}(\nabla_{g}\phi, \nabla_{g}\phi)}{ -g(\nabla_{g}\phi, \nabla_{g} \phi)}  \Box_{g}^{q} \phi  \nonumber\\
& \qquad \qquad  \qquad  \qquad - 2 (q-2)   \,\big(-g(\nabla_{g} \phi, \nabla_{g} \phi) \big)^{q-3}  \Hess_{\phi}\left(\nabla_{g}\phi, \Hess_{\phi}(\nabla_{g}\phi)\right)  \nonumber\\
&  \qquad \qquad  \qquad  \qquad - 2 (q-2)   \,\big(-g(\nabla_{g} \phi, \nabla_{g} \phi) \big)^{q-2}  \left(\frac{ \Hess_{\phi}\left(\nabla_{g}\phi,  \nabla_{g}\phi  \right)}{-g(\nabla_{g}\phi, \nabla_{g} \phi)}\right)^{2}.
\end{align}
Now, the combination of  \eqref{eq:Lqphinablaphi2} and \eqref{eq:q-Bochner2ord} yields
\begin{align}\label{eq:Lqphinablaphi22}
 &-L^{q}_{\phi}\left(\frac{(-g(\nabla_{g} \phi, \nabla_{g} \phi))^{ \frac{q} {2}}}{q} \right)=- g( \nabla_{g} \Box_{g}^{q} \phi, \nabla_{g}^{q} \phi) \nonumber\\
 & \qquad \qquad + (q-2)^{2}  |g(\nabla_{g} \phi, \nabla_{g} \phi)|^{q-2} \left( \frac{\Hess_{\phi}(\nabla_{g}\phi, \nabla_{g}\phi)}{|g(\nabla_{g} \phi, \nabla_{g} \phi)|} \right)^{2} \nonumber\\
& \qquad \qquad+   |g(\nabla_{g} \phi, \nabla_{g} \phi)|^{q-2} \big( \Ric (\nabla_{g} \phi, \nabla_{g} \phi)  +  g(\Hess_{\phi}, \Hess_{\phi}) \big)  \nonumber\\
&\qquad \qquad - 2 (q-2)   \, |g(\nabla_{g} \phi, \nabla_{g} \phi)|^{q-3}  \Hess_{\phi}\left(\nabla_{g}\phi, \Hess_{\phi}(\nabla_{g}\phi)\right) .
\end{align}
\hfill$\Box$

\section{A synthetic formulation of Einstein's vacuum equations in a non-smooth setting}\label{AppB}

\subsection{The Lorentzian synthetic framework}
The goal of this appendix in to give a synthetic  formulation of Einstein's vacuum equations (i.e. zero stress-energy tensor $T\equiv 0$ but possibly non-zero cosmological constant $\Lambda$) and show its stability under a natural adaptation for Lorentzian synthetic spaces of the measured Gromov-Hausdorff convergence (classically designed as a notion of convergence for metric measure spaces).  
\\This appendix was written in early 2021, more than two years after the rest of the paper was posted in arXiv (in October 2018). One of the reasons is that we will build on top of the synthetic time-like Ricci curvature lower bounds in a non-smooth setting developed by the first author in collaboration with Cavalletti in \cite{CaMoLor} (posted on arXiv in 2020).
\\

The general framework is given by Lorentzian pre-length/geodesic spaces. Let us start by recalling some basics of theory, following the approach of Kunziger-S\"amann \cite{KS}.
\\A \emph{causal space}  $(X,\ll,\leq)$ is a set $X$ endowed with a preorder $\leq$ and a transitive relation $\ll$ contained in $\leq$.
\\We write $x<y$ if $x\leq y$ and  $x\neq y$.  When $x\ll y$  (resp. $x\leq y$),  we say that $x$ and $y$ are \emph{time-like} (resp. \emph{causally}) related.

We define the \emph{chronological} (resp. \emph{causal}) future of a subset $E\subset X$ as
\begin{align*}
I^{+}(E)&:=\{y\in X\,:\, \exists x\in E,\, x\ll y\} \\
J^{+}(E)&:=\{y\in X\,:\, \exists x\in E,\, x\leq y\}
\end{align*}
respectively. Analogously, we define $I^{-}(E)$ (resp. $J^{-}(E)$) the   \emph{chronological} (resp. \emph{causal}) past of $E$. In case $E=\{x\}$ is a singleton, with a slight abuse of notation, we will write $I^{\pm}(x)$ (resp. $J^{\pm}(x)$) instead of  $I^{\pm}(\{x\})$ (resp. $J^{\pm}(\{x\})$).

%
Recall that a metric space $(X,\sfd)$ is said to be \emph{proper} if closed and bounded subsets are compact.

\begin{definition}[Lorentzian pre-length space $(X,\sfd, \ll, \leq, \tau)$]
A \emph{Lorentzian pre-length space} $(X,\sfd, \ll, \leq, \tau)$ is a  casual space $(X,\ll,\leq)$ additionally equipped with a proper metric $\sfd$  and a lower semicontinuous function $\tau: X\times X\to [0,\infty]$,  called \emph{time-separation function}, satisfying
\begin{equation}\label{eq:deftau}
\begin{split}
\tau(x,y)+\tau(y,z)\leq \tau (x,z) &\quad\forall x\leq y\leq z \quad \text{reverse triangle inequality} \\
\tau(x,y)=0, \; \text{if } x\not\leq y, & \quad  \tau(x,y)>0 \Leftrightarrow x\ll y.
\end{split}
\end{equation}
\end{definition}
\noindent
\\We will consider $X$ endowed with the metric topology induced by $\sfd$.

We say that $X$ is (resp. \emph{locally}) \emph{causally closed} if $\{x\leq y\}\subset X\times X$ is a closed subset (resp. if every point $x\in X$ has neighbourhood $U$ such that $\{x\leq y\}\cap \bar{U}\times \bar{U}$ is closed in $\bar{U}\times \bar{U}$).


Throughout this appendix, $I\subset \R$ will denote an arbitrary interval.  A non-constant curve $\gamma:I\to X$ is called (future-directed) \emph{time-like} (resp. \emph{causal}) if $\gamma$ is locally Lipschitz continuous (with respect to $\sfd$) and if for all $t_{1}, t_{2}\in I$, with $t_{1}<t_{2}$, it holds $\gamma_{t_{1}}\ll \gamma_{t_{2}}$ (resp. $\gamma_{t_{1}}\leq \gamma_{t_{2}}$).


The length of a causal curve is defined via the time separation function, in analogy to the theory of length  metric spaces: 
for $\gamma:[a,b]\to X$ future-directed causal we set 
\begin{equation*}
{\rm L}_{\tau}(\gamma):=\inf\left\{ \sum_{i=0}^{N-1} \tau(\gamma_{t_{i}}, \gamma_{t_{i+1}})  \,:\, a=t_{0}<t_{1}<\ldots<t_{N}=b, \; N\in \N \right\}.
\end{equation*}

A future-directed causal curve $\gamma:[a,b]\to X$ is \emph{maximal} (also called \emph{geodesic}) if  it realises the
time separation, i.e. if ${\rm L}_{\tau}(\gamma)=\tau(\gamma_{a}, \gamma_{b})$.
\\A Lorentzian pre-length space $(X,\sfd, \ll, \leq,\tau)$ is called 
\begin{itemize}
\item \emph{non-totally imprisoning} if for every compact set $K\Subset X$ there is constant $C>0$ such that the $\sfd$-arc-length of all causal curves contained in $K$ is bounded by $C$;
\item  \emph{globally hyperbolic} if it is non-totally imprisoning and  for every $x,y\in X$ the set $J^{+}(x)\cap J^{-}(y)$ is compact in $X$;
\item \emph{$\mathcal K$-globally hyperbolic}  if it is non-totally imprisoning and  for every $K_{1}, K_{2}\Subset X$ compact subsets,  the set $J^{+}(K_{1})\cap J^{-}(K_{2})$ is compact in $X$;
\item   \emph{geodesic} if for all $x,y\in X$ with $x \leq y$ there is a future-directed causal curve $\gamma$ from $x$ to $y$ with $\tau(x,y)= L_{\tau}(\gamma)$, i.e. a (maximizing) geodesic from $x$ to $y$.
\end{itemize}
For a globally hyperbolic  Lorentzian geodesic (actually length would suffice) space $(X,\sfd, \ll, \leq,\tau)$, the time-separation function $\tau$ is finite and continuous, see \cite[Theorem 3.28]{KS}. Moreover, any globally hyperbolic  Lorentzian length space (for the definition of Lorenzian length space see \cite[Definition 3.22]{KS}, we omit it for brevity since we will not use it) is geodesic  \cite[Theorem 3.30]{KS}.
\\
A \emph{measured Lorentzian pre-length space} $(X,\sfd, \mm, \ll, \leq, \tau)$ is a Lorentzian pre-length space endowed with a Radon non-negative measure $\mm$ with $\supp \, \mm=X$.

\subsection*{Examples entering the class of Lorentzian synthetic  spaces}
In this short section we briefly recall some notable examples entering the aforementioned framework of  Lorentzian synthetic spaces.
\medskip

\noindent 
\textit{Spacetimes with a continuous Lorentzian metric.} Let $M$ be a smooth manifold endowed with a continuous Lorentzian metric $g$. Assume that $(M,g)$ is time-oriented, i.e. there is a continuous time-like
vector field.  Observe that, for $C^0$-metrics, the natural class of differentiability of the underlying manifolds is $C^{1}$; now, $C^{1}$ manifolds always
admit a $C^{\infty}$ subatlas, and one can pick such sub-atlas whenever convenient.  
 
A causal (respectively time-like) curve in $M$ is by definition a locally Lipschitz
curve whose tangent vector is causal (resp. time-like) almost everywhere. One could also start from absolutely continuous (AC for short) curves, but since causal  AC
curves always admit a Lipschitz re-parametrisation  \cite[Sec. 2.1, Rem. 2.3]{Min}, we do not loose in generality with the above convention.  
Set  ${\rm L}_{g}(\gamma)$ to be the $g$-length of a causal curve $\gamma:I\to M$, i.e. 
$${\rm L}_{g}(\gamma):=\int_{I}\sqrt{-g(\dot \gamma, \dot \gamma)}\, dt.$$
The time separation function $\tau:M\times M\to [0,\infty]$ is then defined in the classical way, i.e. 
$$\tau(x,y):=\sup\{{\rm L}_{g}(\gamma): \gamma \text{ is future directed causal from $x$ to $y$}\}, \quad  \text{if $x\leq y$},$$
and $\tau(x,y)=0$ otherwise.  The reverse triangle inequality \eqref{eq:deftau} follows directly from the definition. Moreover, every  ${\rm L}_{g}$-maximal curve $\gamma$ is also ${\rm L_{\tau}}$-maximal, and  ${\rm L}_{g}(\gamma) = {\rm L}_{\tau}(\gamma)$ (see for instance \cite[Remark 5.1]{KS}). In order to have an underlying metric structure, we also fix a complete Riemannian
metric $h$ on $M$ and denote by $\sfd^{h}$ the associated distance function.

For a  spacetime with a Lorentzian $C^{0}$-metric: 
\begin{itemize}
\item Global hyperbolicity implies causal closedness and  $\cK$-global hyperbolicity  \cite[Proposition 3.3 and Corollary 3.4]{SaC0}.
\item Recall that a Cauchy hypersurface is a subset  which is met exactly  once by every inextendible causal curve. Every Cauchy hypersurface is a closed acasual topological hypersurface \cite[Proposition 5.2]{SaC0}.   Global hyperbolicity is equivalent to the existence of a Cauchy hypersurface \cite[Theorem 5.7, Theorem 5.9]{SaC0}.
\item By \cite[Proposition 5.8]{KS},   if $g$ is a causally plain (or, more strongly, locally Lipschitz) Lorentzian $C^{0}$-metric on $M$ then the associated synthetic structure is a pre-length Lorentzian space.
More strongly, from   \cite[Theorem 3.30 and Theorem 5.12]{KS}, if $g$ is a   globally hyperbolic and causally plain Lorentzian $C^{0}$-metric on $M$ then the associated synthetic structure is a  causally closed,  $\cK$-globally hyperbolic Lorentzian geodesic space.
\end{itemize}
\medskip

\noindent 
\textit{Closed cone structures}.  Closed cone structures provide
a rich source of examples of Lorentzian pre-length and length spaces, which can be seen as the synthetic-Lorentzian analogue of Finsler manifolds (see \cite[Section 5.2]{KS} for more details). See Minguzzi's review paper \cite{Min} for a comprehensive analysis of causality theory in the framework of closed cone structures, including embedding and singularity theorems.
\medskip

\noindent \textit{Some examples towards  quantum gravity}. The framework of Lorentzian pre-length spaces allows to handle situations
where one may not have the structure of a manifold or a Lorentz(-Finsler) metric.  A remarkable example of such a situation is given by certain approaches to quantum gravity, see for instance \cite{MP} where it is shown that it is possible to reconstruct a globally hyperbolic spacetime and the causality relation from  a countable dense set of events,  in a purely order theoretic manner. 
\\Two approaches to quantum gravity,  particularly linked to Lorentzian pre-length spaces,  are the theory of \emph{causal Fermion systems} \cite{FinsterPrimer, FinsterBook} and the  \emph{theory of causal sets} \cite{BLMS}. The basic idea in both cases is that the structure of spacetime needs to be adjusted
on a microscopic scale to include quantum effects. This leads to non-smoothness of
the underlying geometry, and  the classical  structure of Lorentzian manifold emerges only in the macroscopic regime.   For the connection to the theory of Lorentzian (pre-)length spaces we refer to \cite[Section 5.3]{KS}, \cite[Section 5.1]{FinsterPrimer}.

\subsection{The $p$-Lorentz-Wasserstein distance on a Lorentzian pre-length space.}
We denote with $\mathcal P(X)$ (resp.  $\mathcal{P}_{c}(X)$, or  $\mathcal{P}_{ac}(X)$)  the  collection of all Borel probability
measures (resp. with compact support, or  absolutely continuous with respect to $\mm$).
\\Given a probability measure $\mu \in \mathcal{P}(X)$ we define 
its relative entropy by
\begin{equation*}
\Ent(\mu|\mm) = \int_{X} \rho \log(\rho) \, d\mm,
\end{equation*}
if $\mu = \rho \, \mm\in \mathcal{P}_{ac}(X)$ and $(\rho\log(\rho))_{+}$ is  $\mm$-integrable. 
Otherwise we set $\Ent(\mu|\mm) = +\infty$. We set ${\rm Dom}(\Ent(\cdot|\mm))$ the finiteness domain of $\Ent(\cdot|\mm)$.
\\

We next recall some basics of optimal transport theory in Lorentzian pre-length spaces, following the approach and the notation in \cite{CaMoLor}.
\\Given $\mu,\nu\in \mathcal{P}(X)$, denote
\begin{align*}
 \Pi(\mu,\nu)&:=\{\pi\in  \mathcal{P}(X\times X) \,:\, (P_{1})_{\sharp}\pi=\mu, \, (P_{2})_{\sharp}\pi=\nu \}, \nonumber \\
 \Pi_{\leq}(\mu,\nu)&:=\{\pi\in  \Pi(\mu,\nu) \,:\,  \pi(X^{2}_{\leq})=1 \}, 
\nonumber \\
  \Pi_{\ll}(\mu,\nu)&:=\{\pi\in  \Pi(\mu,\nu) \,:\,  \pi(X^{2}_{\ll})=1 \} 
\end{align*}
where $X^{2}_{\leq}:=\{(x,y) \in X^{2}\,:\, x\leq y \}$ and   $X^{2}_{\ll}:=\{(x,y) \in X^{2}\,:\, x\ll y \}$.


\begin{definition}[The $p$-Lorentz-Wasserstein distance]\label{def:Wp}
Let  $(X,\sfd, \ll, \leq, \tau)$ be a Lorentzian pre-length space and let $p\in (0,1]$. Given $\mu,\nu\in \mathcal{P}(X)$, the \emph{$p$-Lorentz-Wasserstein distance} is defined by
\begin{equation}\label{eq:defWp}
\ell_{p}(\mu,\nu):= \sup_{\pi \in \Pi_{\leq}(\mu,\nu)} \left(  \int_{X\times X}  \tau(x,y)^{p} \, d\pi(x,y)\right)^{1/p}.
\end{equation}
When $\Pi_{\leq}(\mu,\nu)=\emptyset$ we set $\ell_{p}(\mu,\nu):=-\infty$.
\end{definition}
Note that Definition \ref{def:Wp}  extends to Lorentzian pre-length spaces the corresponding notion given in the smooth Lorentzian setting in \cite{EM17}  (see also \cite{McCann18}, and \cite{suhr} for $p=1$); when $\Pi_{\leq}(\mu,\nu)=\emptyset$ we adopt the convention of McCann  \cite{McCann18} (note that \cite{EM17} set $\ell_{p}(\mu,\nu)=0$ in this case).
\\

A key property of $\ell_{p}$ is the reverse triangle inequality. 
This was proved in the smooth Lorentzian setting by 
Eckstein-Miller \cite[Theorem 13]{EM17} and in the present synthetic setting in \cite[Proposition 2.5]{CaMoLor}.
Such a property is the natural Lorentzian 
analogue of the fact that the Kantorovich-Rubinstein-Wasserstein distances 
$W_{p}$, $p\geq 1$, in the metric space setting satisfy 
the usual triangle inequality (see for instance \cite[Section 6]{Vil}).

\begin{proposition}[$\ell_{p}$ satisfies the reverse triangle inequality]\label{prop:RTIellp}
Let  $(X,\sfd, \ll, \leq, \tau)$ be a Lorentzian pre-length space and let $p\in (0,1]$. Then $\ell_{p}$ satisfies the reverse triangle inequality:
\begin{equation}\label{eq:RTInellq}
\ell_{p}(\mu_{0},\mu_{1})+ \ell_{p}(\mu_{1},\mu_{2})
\leq \ell_{p}(\mu_{0}, \mu_{2}), 
\quad \forall \mu_{0},\mu_{1},\mu_{2}\in \mathcal{P}(X),
\end{equation}
where we adopt the convention that $\infty-\infty=-\infty$ to interpret the left hand side of \eqref{eq:RTInellq}.
\end{proposition}

A coupling  $\pi\in  \Pi_{\leq}(\mu,\nu)$ maximising in \eqref{eq:defWp} is called \emph{$\ell_{p}$-optimal}. The set of \emph{$\ell_{p}$-optimal} couplings from $\mu$ to $\nu$ is denoted by $  \Pi_{\leq}^{p\text{-opt}}(\mu,\nu)$. We also set 
$ \Pi_{\ll}^{p\text{-opt}}(\mu,\nu):=\Pi_{\ll}(\mu,\nu)\cap  \Pi_{\leq}^{p\text{-opt}}(\mu,\nu)$.

We say that $(\mu_{s})_{s\in [0,1]}\subset \mathcal{P}(X)$ is an $\ell_{p}$-geodesic if and only if 
\begin{equation}\label{eq:defgeodellq}
\ell_{p}(\mu_{s}, \mu_{t})=(t-s)\,  \ell_{p}(\mu_{0}, \mu_{1}) \in (0,\infty), \quad \text{for all }0\leq s< t\leq 1.
\end{equation}
Note that, with the convention above, $\ell_{p}$-geodesics are implicitly future-directed and time-like.

The next lemma follows by the classical theory of optimal transport, since the assumptions allow to localise the argument on $\supp\, \mu_{0}\times \supp \, \mu_{1}\subset X^{2}_{\leq}$ where the optimal transport problem becomes standard.
\begin{lemma}\label{lem:OptCycMon}
Assume $\supp\, \mu \times \supp \, \nu \subset X^{2}_{\leq}$. Then
\begin{enumerate}
\item Every $\ell_{p}$-optimal coupling $\pi\in \Pi_{\leq}^{p\text{-opt}}(\mu, \nu)$ is also $\tau^{p}$-cyclically monotone; 
\item If $\pi\in \Pi_{\leq}(\mu, \nu)$ is  $\tau^{p}$-cyclically monotone then  $\pi\in \Pi_{\leq}^{p\text{-opt}}(\mu, \nu)$  is $\ell_{p}$-optimal. 
\end{enumerate}
\end{lemma}

\subsection{Synthetic time-like Ricci upper bounds.}

Inspired by Sturm's approach to Riemannian/metric Ricci curvature upper bounds \cite{StUB} and by Theorem \ref{thm:RiccileqT} (see also Remark \ref{rem:RiccileqT}), we propose the following synthetic notion of time-like Ricci curvature bounded above. We denote the metric ball in $(X,\sfd)$ with center $x$ and radius $r$ by $B^{\sfd}_{r}(x)$.

 \begin{definition}[Synthetic time-like Ricci curvature bounded above]\label{def:RUB}
Fix $p\in (0,1)$, $K\in \R$. We say that  a  measured  Lorentzian  pre-length space $(X,\sfd,\mm, \ll, \leq, \tau)$ has time-like Ricci curvature bounded above by $K$ in a synthetic sense if there exists $r_{0}>0$ and  a function $\omega:[0,r_{0})\to [0,\infty)$ with $\lim_{r\downarrow 0} \omega(r)=0$ such that for every $r\in [0,r_{0})$ the following holds. 
\begin{itemize}
\item For every $x,y\in X$ with $\sfd(x,y)=r>0$ and such that $B^{\sfd}_{r^4}(x)\times B^{\sfd}_{r^2}(y)\subset X^{2}_{\ll}$,
\item for every $\mu_{0}\in {\rm Dom}(\Ent(\cdot|\mm)$ with $\supp\, \mu_{0}\subset B^{\sfd}_{r^4}(x)$,
\end{itemize}
there exists an $\ell_{p}$-geodesic $(\mu_{t})_{t\in [-1,1]}$ satisfying
\begin{itemize}
\item  $\supp\, \mu_{1}\subset B^{\sfd}_{r^2}(y)$, 
\item $\supp\, \mu_{-1}\times \supp \, \mu_{1}\subset X^{2}_{\leq},$
\item  $\bigcup_{t\in [-1, 1]} \supp \, \mu_{t} \subset B_{10 r_{0}}^{\sfd} (x) $,
\item $
\Ent(\mu_{-1}|\mm)-2 \Ent(\mu_{0}|\mm)+\Ent(\mu_{1}|\mm) \leq (K+\omega(r)) \, r^{2}.
$
\end{itemize}
\end{definition}

A key property of the above notion of time-like Ricci bounded above is the stability under weak convergence of measured  Lorentzian  pre-length spaces. The stability of Riemannian/metric Ricci upper bounds via optimal transport was proved by Sturm \cite{StUB}. The notion of convergence we use below is a slight reinforcement (we ask that the immersion maps are isometries with respect to the metric structures instead of merely topological embedding maps) of the weak convergence used in \cite{CaMoLor} to show stability of synthetic time-like Ricci \emph{lower} bounds; in any case it is a  natural adaptation to the Lorentzian setting of the mGH convergence used for metric measure spaces (see for instance \cite[Section 3]{GMS2013} for an overview of equivalent formulations of convergence for metric measure spaces).

\begin{theorem}[Stability of time-like Ricci curvature upper bounds]\label{thm:stabUB}

Let  $\{(X_{j},\sfd_{j}, \mm_{j}, \ll_{j}, \leq_{j}, \tau_{j})\}_{j\in \N\cup\{\infty\}}$ be a sequence of  measured Lorentzian geodesic spaces  satisfying the following properties :
\begin{enumerate}
\item There exists a  locally causally closed, globally hyperbolic  Lorentzian geodesic space $(X, \sfd,  \ll, \leq,  \tau)$ such that each $(X_{j},\sfd_{j}, \mm_{j}, \ll_{j}, \leq_{j}, \tau_{j})$, $j\in \N\cup\{\infty\}$, is 
 isomorphically embedded in it, i.e. there exist inclusion  maps $\iota_{j}: X_{j}\hookrightarrow X$  such that  for every $x^{1}_{j}, x^{2}_{j}\in X_{j}$,  for every $j\in \N\cup\{\infty\}$, the following holds:
\begin{itemize}
\item ${\sfd}  (\iota_{j}(x^{1}_{j}), \iota_{j}(x^{2}_{j}))= \sfd_{j} (x^{1}_{j}, x^{2}_{j})$;
\item  $x^{1}_{j} \leq_{j} x^{2}_{j}$ 
if and only if $\iota_{j}(x^{1}_{j}) \leq  \iota_{j}(x^{2}_{j})$; 
\item $\tau  (\iota_{j}(x^{1}_{j}), \iota_{j}(x^{2}_{j}))= \tau_{j} (x^{1}_{j}, x^{2}_{j})$.
\end{itemize}

\item The measures $(\iota_{j})_{\sharp} \mm_{j}$  converge to $(\iota_{\infty})_{\sharp} \mm_{\infty}$ weakly in duality with $C_{c}( X)$ in $ X$, i.e.
\begin{equation}\label{eq:weakconv}
\int \varphi\; d(\iota_{j})_{\sharp} \mm_{j} \to \int \varphi \; d(\iota_{\infty})_{\sharp} \mm_{\infty} \quad \forall \varphi \in C_{c}( X),
\end{equation} 
where  $C_{c}( X)$ denotes the set of continuous functions with compact support.

\item Volume non-collapsing: there exists a function $v:(0,\infty)\to (0,\infty)$ such that for every $x_{j}\in X_{j}$ it holds $\mm_{j}(B^{\sfd_{j}}_{r}(x_{j})) \geq v(r)>0$.
\item There exists a function $\omega:[0,\infty)\to [0,\infty)$ with $\lim_{r\downarrow 0} \omega(r)=0$ and there exist $p\in (0,1), K\in \R$ such that  $(X_{j},\sfd_{j}, \mm_{j}, \ll_{j}, \leq_{j}, \tau_{j})$ has time-like Ricci curvature bounded above by $K$ with respect to $p\in (0,1)$, $r_{0}>0$ and with remainder function $\omega$ in the synthetic sense of Definition \ref{def:RUB}.
\end{enumerate}
Then also the limit space 
$(X_{\infty},\sfd_{\infty}, \mm_{\infty}, \ll_{\infty}, \leq_{\infty}, \tau_{\infty})$  
 has time-like Ricci curvature bounded above by $K$ with respect to $p\in (0,1), r_{0}+1$ and with remainder function $\omega$ in the synthetic sense of Definition \ref{def:RUB}.

\end{theorem}

\begin{proof}
For simplicity of notation, we will identify $X_{j}$ with its isomorphic image $\iota_{j}(X_{j})\subset  X$ and the measure $\mm_{j}$ with $(\iota_{j})_{\sharp} \mm_{j}$, for each $j\in \N\cup \{\infty\}$. 

We will write $B^{\sfd_{j}}_{r}(x)$ for the ball in $(X_{j}, \sfd_{j})$ centred at $x\in X_{j}$. Notice that $B^{\sfd_{j}}_{r}(x)= B^{{\sfd}}_{r}(x)\cap X_{j}$.

Fix  arbitrary $x_{\infty},y_{\infty}\in X_{\infty}$ with $\sfd_{\infty}(x_{\infty},y_{\infty})=r>0$ and $B^{\sfd_{\infty}}_{r^{4}} (x_{\infty}) \times B^{\sfd_{\infty}}_{r^{2}} (y_{\infty}) \subset (X_{\infty}^{2})_{\ll}$. Since $(X^{2}_{\infty})_{\ll}\subset X_{\infty}\times X_{\infty}$ is an open subset, there exist $\varepsilon>0$ such that
\begin{equation}\label{eq:BrepsInf}
B^{\sfd_{\infty}}_{(r+\varepsilon)^{4}} (x_{\infty}) \times B^{\sfd_{\infty}}_{(r+\varepsilon)^{2}} (y_{\infty}) \subset (X^{2}_{\infty})_{\ll}\;.
\end{equation}
Fix  $\mu_{0}^{\infty} \in \rm{Dom}(\Ent(\cdot|\mm_{\infty}))$ with $\supp \, \mu_{0}^{\infty}\subset B^{\sfd_{\infty}}_{r^{4}} (x_{\infty})$.
\\

\textbf{Step 1}. We claim that, for $j\in \N$ sufficiently large, there exist $x_{j}, y_{j}\in X_{j}$ such that 
\begin{enumerate}
\item $x_{j}\to x_{\infty}$ and $y_{j}\to y_{\infty}$ in $( X, {\sfd})$;
\item $\sfd_{j}(x_{j}, y_{j})=:r_{j}\to r$;
\item $B^{\sfd_{j}}_{r_{j}^{4}} (x_{j}) \times B^{\sfd_{j}}_{r_{j}^{2}} (y_{j}) \subset (X_{j}^{2})_{\ll}$ .
\end{enumerate}
Since by assumption $\mm_{j} \rightharpoonup \mm_{\infty}$ weakly as measures in $({X}, {\sfd})$, by the lower-semicontinuity over open subset we have that for every $r>0$ it holds 
$$
0<\mm_{\infty}(B^{\sfd_{\infty}}_{r}(x_{\infty}))= \mm_{\infty}(B^{{\sfd}}_{r}(x_{\infty}))\leq \liminf_{j\to \infty} \mm_{j}(B^{{\sfd}}_{r}(x_{\infty})).
$$
In particular, for every $r>0$ there exists $J(r)>0$ such that $B^{{\sfd}}_{r}(x_{\infty}) \cap X_{j}\neq \emptyset$ for all $j\geq J(r)$.  By a diagonal argument we obtain that there exist $x_{j}\in X_{j}$ with $x_{j}\to x_{\infty}$ in $( X, {\sfd})$. The analogous argument gives a sequence $y_{j}\in X_{j}$ with $y_{j}\to y_{\infty}$.
\\Using that the inclusion maps $\iota_{j}$ are isometric, we also get
$$
r_{j}:= \sfd_{j}(x_{j}, y_{j})=  {\sfd}(x_{j}, y_{j})\to  {\sfd}(x_{\infty}, y_{\infty})=r.
$$ 
We are left to show the third claim.
For every $j\in \N$, let  $x_{j}'\in B^{\sfd_{j}}_{r_{j}^{4}} (x_{j})$ and  $y_{j}'\in B^{\sfd_{j}}_{r_{j}^{2}} (y_{j})$ be arbitrary.
\\Combining the volume non-collapsing assumption with the weak convergence $\mm_{j}  \rightharpoonup \mm_{\infty}$ and the properness of the metrics, we obtain that there exist 
$x_{\infty}'\in B^{ \sfd_{\infty}}_{(r+\varepsilon)^{4}} (x_{\infty})$ and  $y_{\infty}'\in B^{\sfd_{\infty}}_{(r+\varepsilon)^{2}} (y_{\infty})$ such that 
$x_{j}'\to x_{\infty}'$ and $y_{j}'\to y_{\infty}'$ up to a subsequence.
 \\Recalling \eqref{eq:BrepsInf} and that $\iota_{\infty}$ is an isomorphic embedding, we infer that $(x'_{\infty}, y'_{\infty})\in {X}^{2}_{\ll}$. Since ${X}^{2}_{\ll}\subset X^{2}$ is an open subset, we conclude that for $j$ sufficiently large it holds
 $$
 (x_{j}', y_{j}')\subset {X}^{2}_{\ll} \cap X_{j}^{2}= (X_{j}^{2})_{\ll}
 $$
 as desired.
 \\
 
 \textbf{Step 2}. We claim that, for every $j\in \N$, there exists $\mu_{0}^{j}\in \rm{Dom}(\Ent(\ \cdot\ |\mm_{j}))$ with $\supp \, \mu_{0}^{j}\subset B^{\sfd_{j}}_{r_{j}^{4}} (x_{j})$ such that $\mu_{0}^{j}\to \mu_{0}^{\infty}$ narrowly and
 \begin{equation}\label{eq:claimStep2}
 \Ent(\mu_{0}^{\infty}|\mm_{\infty})\geq \limsup_{j\to \infty} \Ent(\mu_{0}^{j}|\mm_{j}).
 \end{equation}
 
Since $\mm_{j} \rightharpoonup \mm_{\infty}$ weakly,  there exists $R\in (2r^{4}, 3 r^{4})$ such that 
\begin{equation}\label{eqzjconv}
z_{j}^{-1}:=\mm_{j}(B^{ \sfd}_{R}(x_{\infty})) \rightarrow \mm_{\infty}(B^{ \sfd}_{R}(x_{\infty}))=:z_{\infty}^{-1}.
\end{equation}
For $j\in \N\cup \{\infty\}$,  denote  $\tilde{\mm}_{j}:= z_{j} \, \mm_{j}\llcorner B^{ \sfd}_{R}(x_{\infty})$ and observe that $\tilde{\mm}_{j}  \rightharpoonup \tilde{\mm}_{\infty}$ weakly and thus (since now the supports are uniformly bounded in ${X}$) in $W_{q}^{( X,  \sfd)}$ for some (or equivalently every) $q\in [1,\infty)$. In particular $\tilde{\mm}_{j}\to \tilde{\mm}_{\infty}$ in $W_{2}^{( X,  \sfd)}$. Let 
\begin{equation}\label{eq:defggammaj}
\text{${\eta}_{j}\in \Pi(\tilde{\mm}_{\infty}, \tilde{\mm}_{j})$ be a  $W_{2}^{( X,  \sfd)}$-optimal coupling.} 
\end{equation}

 Let us first consider the case $\mu_{0}^{\infty}=\rho\, \tilde{\mm}_{\infty}$ with $\rho \in L^{\infty}(\tilde{\mm}_{\infty})$.
 \\
Define $\eta_j'\in {\mathcal P}( X)$ as $d\eta_j'(x,y):=\rho(x) \, d\eta_j(x,y)$ and $\tilde{\mu}_{0}^{j}:=(p_{2})_\sharp \eta_j'\in {\mathcal P}(X_{j})$.
By construction, $\eta_j'\ll\eta_j$, hence $\tilde{\mu}_{0}^{j} \ll (p_{2})_\sharp\eta_j=\tilde\mm_j$.
Let $\tilde{\mu}_{0}^{j}=\tilde{\rho}_{0}^{j} \tilde\mm_j$. It is readily checked from the definition that it holds $\tilde{\rho}_{0}^{j}(y)=\int \rho(x)\, d(\eta_j)_{y}(x)$, where $\{(\eta_{j})_y\}$ is the disintegration of $\eta_j$ w.r.t. the projection on the second marginal. By Jensen's inequality applied to the convex function $u(s):=s\log(s)$, we have:
\begin{equation}\label{eq:Enttildemuileq}
\begin{split}
\Ent(\tilde{\mu}_{0}^{j}|\tilde{\mm}_{j})&=\int u(\tilde{\rho}_{0}^{j})\, d\tilde\mm_j= \int u\left(\int \rho(x)\, d(\eta_j)_{y}(x)\right)\, d\tilde\mm_j(y)\\
&\leq \int \int u(\rho(x))\,d(\eta_j)_{y}(x)\,d\tilde\mm_j(y)=\int \int u(\rho(x))\,d\eta_{j}(x,y)\\
&=\int u(\rho)\,d(p_{1})_\sharp\eta_{j}=\int u(\rho)\,d\tilde\mm_\infty=\Ent(\mu_{0}^{\infty}|\tilde{\mm}_{\infty}).
\end{split}
\end{equation}
Since by construction we have $\gamma_j'\in\Pi(\mu_{0}^{\infty}, \tilde{\mu}_{0}^{j})$, it holds
\[
\begin{split}
W_2^2(\mu_{0}^{\infty}, \tilde{\mu}_{0}^{j})&\leq \int\sfd^2(x,y)\,d\eta_j'(x,y)= \int\rho(x)\, \sfd^2(x,y)\,d\eta_j(x,y)\\
&\leq  \| \rho\|_{L^{\infty}(\tilde{\mm}_{\infty})} \,  W_2^2(\tilde\mm_\infty,\tilde\mm_j),
\end{split}
\]
and therefore $W_2(\mu_{0}^{\infty}, \tilde{\mu}_{0}^{j})\to 0$.  In particular, $\tilde{\mu}_{0}^{j}\to \mu_{0}^{\infty}$ narrowly in $( X,  \sfd)$.
\\ Using that $x_{j}\to x_{\infty}, r_{j}\to r$ and  $\supp \, \mu_{0}^{\infty}\subset B^{\sfd_{\infty}}_{r^{4}} (x_{\infty})$,  it follows that $c_{j}^{-1}:=  \tilde{\mu}_{0}^{j} \big(B^{\sfd_{j}}_{(r_{j}')^{4}}(x_{j}) \big)\to 1$, where we set $r_{j}':=r_{j}-\frac{1}{j}$. 
Define
 ${\mu}_{0}^{j}:=c_j\, \tilde{\mu}_{0}^{j}\llcorner B^{\sfd_{j}}_{(r_{j}')^{4}}(x_{j})$.
 \\Clearly,  $\mu_{0}^{j}\in \rm{Dom}(\Ent(\ \cdot\ |\mm_{j}))$,  $\supp \, \mu_{0}^{j}\subset B^{\sfd_{j}}_{r_{j}^{4}} (x_{j})$,  $W_2(\mu_{0}^{\infty}, {\mu}_{0}^{j})\to 0$ and 
 \begin{equation}\label{eq:Entmujmujtilde}
\lim_{j\to \infty} \big|  \Ent({\mu}_{0}^{j}|\tilde{\mm}_{j})- \Ent(\tilde{\mu}_{0}^{j}|\tilde{\mm}_{j}) \big|=0.
 \end{equation}
 The combination of \eqref{eqzjconv}, \eqref{eq:Enttildemuileq} and \eqref{eq:Entmujmujtilde} gives the desired \eqref{eq:claimStep2}.

If $\rho$ is not bounded, for $k\in\N$ define $\rho^k:=C_k\min\{\rho,k\}$, $C_k$ being such that  $\mu^k_{\infty}:=\rho^k\, \tilde\mm_\infty\in {\mathcal P} (X_{\infty})$. Clearly, it holds
\[
\limsup_{k\to\infty}\Ent(\mu^k_{\infty}|\mm_{\infty})\leq \Ent(\mu_\infty|\mm_{\infty}),\qquad \lim_{k\to\infty}W_2(\mu^k_{\infty},\mu_\infty)=0.
\]
Then apply the previous procedure to $\mu^k_{\infty}$ and conclude with a diagonal argument.
\\
 
  \textbf{Step 3}. Conclusion.
\\By assumption, for  every $j\in \N$  the space   $(X_{j},\sfd_{j}, \mm_{j}, \ll_{j}, \leq_{j}, \tau_{j})$ has time-like Ricci curvature bounded above by $K$ with respect to $p\in (0,1), r_{0}>0$ and with remainder function $\omega$. Thus there exists an $\ell_{p}$-geodesic $(\mu_{t}^{j})_{t\in [-1,1]}$ with  $\supp\, \mu_{1}^{j}\subset B^{\sfd_{j}}_{r_{j}^2}(y_{j})$,   $\bigcup_{t\in [-1, 1]} \supp \, \mu_{t}^{j} \subset B_{10 r_{0}}^{\sfd_{j}} (x_{j}) $ such that
\begin{equation}\label{eq:RUBKj}
\Ent(\mu_{-1}|\mm_{j})-2 \Ent(\mu_{0}|\mm_{j})+\Ent(\mu_{1}|\mm_{j}) \leq (K+\omega(r)) \, r^{2}.
\end{equation}
Since by construction $x_{j}\to x_{\infty}$, there exists $J>0$ such that 
\begin{equation}\label{eq:suppmujt}  
\bigcup_{j\geq J} \bigcup_{t\in [-1, 1]} \supp \, \mu_{t}^{j} \subset B_{11 r_{0}}^{{\sfd}} (x_{\infty}).
\end{equation}
\\Let $\eta_{j}\in {\mathcal P} (C[-1,1], ({X},  \sfd))$ be the dynamical plan associated to the geodesic $(\mu_{t}^{j})_{t\in [-1,1]}$, i.e. such that $(\ee_{t})_{\sharp} \eta_{j}=\mu^{j}_{t}$ for every $t\in [0,1]$, $j\in \N$.
From the non-total imprisoning property of $ X$, we infer that there exists $C>0$ such that
$$
\sup\{ {\rm L}_{{\sfd}}(\gamma) : \gamma\in \supp \, \eta_{j}, \, j\in \N  \} \leq C <\infty.
$$
From \cite[Theorem 4]{Lisini}, we infer that $(\mu^{j}_{t})_{t\in [-1, 1]}$ is an absolutely continuous curve in the $W_{1}$-Kantorovich Wasserstein space $(\mathcal{P}( X), W_{1})$ w.r.t. $ \sfd$, with bounded length:
$$
 {\rm L}_{W_{1}} ( (\mu^{j}_{t})_{t\in [-1, 1]} ) \leq \int  {\rm L}_{{\sfd}}(\gamma) \, d\eta_{j}(\gamma) \leq C<\infty, \quad \forall j\in \N.
$$
By the metric Arzel\'a-Ascoli Theorem we deduce that there exists a limit continuous curve $(\mu^{\infty}_{t})_{t\in [-1,1]} \subset \mathcal{P}(X_{\infty}) \cap \mathcal{P}({X}, W_{1}^{({X},{\sfd})})$ such that (up to a sub-sequence)
$ W_{1}^{({X},{\sfd})}\left(\mu^{j}_{t}, \mu^{\infty}_{t} \right)\to 0 $ and thus $\mu^{j}_{t} \to \mu^{\infty}_{t}$ narrowly in $ X$, as $j\to \infty$ for every $t\in [0,1]$. 
From \eqref{eq:suppmujt}  we have
$$
\bigcup_{t\in [-1, 1]} \supp \, \mu_{t}^{\infty} \subset {B}_{11 r_{0}}^{{\sfd}} (x_{\infty})\subset B_{10 (r_{0}+1)}^{{\sfd}} (x_{\infty}).
$$
By assumption, we also know that for every $j\in \N$ it holds 
\begin{equation}\label{eq:suppmujXleq}
\supp\, \mu_{-1}^{j}\times \supp \, \mu_{1}^{j}\subset (X_{j}^{2})_{\leq}=  {X}^{2}_{\leq} \cap X_{j}^{2}, \quad \forall j\in \N.
\end{equation}
Since by assumption ${X}$ is locally causally closed, we infer that 
\begin{equation}\label{eq:suppmuinftyXleq}
\supp\, \mu_{-1}^{\infty}\times \supp \, \mu_{1}^{\infty}\subset (X_{\infty}^{2})_{\leq}=  X^{2}_{\leq} \cap X_{\infty}^{2}.
\end{equation}
For every $j\in \N$, let $\pi_{j}\in \Pi_{\leq}^{p\text{-opt}}(\mu_{-1}^{j}, \mu_{1}^{j})$. Thanks to Lemma \ref{lem:OptCycMon} we know that $\pi_{j}$ is ${\tau}^{p}$-cyclically monotone. 
\\From \eqref{eq:suppmujt}, we have that  $\supp\,  \pi_{j}\subset  B_{11 r_{0}}^{{\sfd}} (x_{\infty})^{2}$ for $j\geq J$ and thus (by Prokhorov Theorem) there exists $\pi_{\infty}\in \mathcal{P}( X\times  X)$ such that $\pi_{j}\to \pi_{\infty} $ narrowly.
 Using the continuity of the projection maps, it is readily seen that $\pi_{\infty}\in \Pi(\mu_{-1}^{\infty}, \mu_{1}^{\infty})$ and thus $\pi_{\infty}\in \Pi_{\leq}(\mu_{-1}^{\infty}, \mu_{1}^{\infty})$, thanks to \eqref{eq:suppmuinftyXleq}. 
 \\Since by global hyperbolicity $ \tau^{p}$ is continuous, the $ \tau^{p}$-cyclical monotonicity is preserved under narrow convergence. We infer that $\pi_{\infty}$ is  $ \tau^{p}$-cyclically monotone and thus $\ell_{p}$-optimal thanks to \eqref{eq:suppmuinftyXleq} and Lemma \ref{lem:OptCycMon}. Therefore:
 \begin{align}
 \ell_{p} (\mu_{-1}^{j}, \mu_{1}^{j})= & \left( \int {\tau}(x,y)^{p}\, d\pi_{j}(x,y) \right)^{1/p} \nonumber \\
 &\longrightarrow    \left(  \int {\tau}(x,y)^{p}\, d\pi_{\infty}(x,y)  \right)^{1/p} =  \ell_{p} (\mu_{-1}^{\infty}, \mu_{1}^{\infty}). \label{eq:ellpjtoinf}
 \end{align}
For intermediate $t\in (-1,1)$,  using that $(\mu^{j}_{t})_{t\in [-1,1]}$ is an $\ell_{p}$-geodesic, we have
\begin{equation}\label{eq:ellqnuinf}
\ell_{p}(\mu^{\infty}_{-1},\mu^{\infty}_{t})\geq  \lim_{j\to \infty}  \ell_{p}(\mu^{j}_{-1}, \mu^{j}_{t})=  (t+1) \;\lim_{j\to \infty}  \ell_{p}(\mu^{j}_{-1}, \mu^{j}_{1}) =  (t+1)\,\ell_{p}(\mu^{\infty}_{-1},\mu^{\infty}_{1}).
\end{equation}
By reverse triangle inequality \eqref{eq:RTInellq}, we get that the curve   $(\mu^{\infty}_{t})_{t\in [-1,1]}$ is an $\ell_{p}$-geodesic from $\mu^{\infty}_{-1}$ to $\mu^{\infty}_{1}$. 
\\The joint lower semicontinuity of $\Ent(\cdot|\cdot)$ under narrow convergence of probability measures (this is a well-known fact, for a proof see for instance \cite[Lemma 9.4.3]{AGSBook}) yields:
\begin{equation*}
\Ent(\mu^{\infty}_{t}|\tilde{\mm}_{\infty}) \leq \liminf_{j\in \N} \Ent(\mu^{j}_{t}|\tilde{\mm}_{j}), \quad \forall t\in [0,1],
\end{equation*}
and thus, using \eqref{eqzjconv}:
\begin{equation}\label{eq:Gammaliminfjt}
\Ent(\mu^{\infty}_{t}| {\mm}_{\infty}) \leq \liminf_{j\in \N} \Ent(\mu^{j}_{t}| {\mm}_{j}), \quad \forall t\in [0,1].
\end{equation}
The combination of \eqref{eq:claimStep2}, \eqref{eq:RUBKj} and  \eqref{eq:Gammaliminfjt} gives that 
\begin{equation*}
\Ent(\mu_{-1}^{\infty}|\mm_{\infty})-2 \Ent(\mu_{0}^{\infty}|\mm_{\infty})+\Ent(\mu_{1}^{\infty}|\mm_{\infty}) \leq (K+\omega(r)) \, r^{2}.
\end{equation*}
as desired.

\end{proof}

\subsection{Synthetic time-like Ricci lower bounds.}
In this subsection we briefly recall some basics of the synthetic theory of time-like Ricci lower bounds developed in \cite{CaMoLor}.

\begin{definition}[Time-like $p$-dualisable measures]\label{def:Tpdual}
 Let  $(X,\sfd, \ll, \leq, \tau)$ be a Lorentzian pre-length space and fix $p\in (0,1]$. We say that $(\mu,\nu)\in \mathcal{P}(X)^{2}$ is \emph{time-like $p$-dualisable (by $\pi\in \Pi_{\ll}(\mu,\nu)$)}  if 
 \begin{enumerate}
\item  $\ell_{p}(\mu,\nu)\in (0,\infty)$;
\item  $\pi\in  \Pi_{\leq}^{p\text{-opt}}(\mu,\nu)$ and $\pi(X^{2}_{\ll})=1$;
\item there exist measurable functions $a,b:X\to \R$, with $a\oplus b \in L^{1}(\mu\otimes \nu)$ such that  $\ell^{p}\leq a\oplus b$ on $\supp \, \mu \times  \supp \, \nu $.
\end{enumerate}
The pair $(\mu,\nu) \in (\mathcal{P}(X))^{2}$ is said to be \emph{strongly time-like $p$-dualisable} if in addition there exists a measurable $\ell^{p}$-cyclically monotone set $\Gamma\subset X^{2}_{\ll} \cap (\supp \, \mu \times \supp \,\nu)$ such that a coupling  $\pi\in \Pi_{\leq}(\mu,\nu)$ is $\ell_{p}$-optimal if  and only if $\pi$ is concentrated on $\Gamma$, i.e. $\pi(\Gamma)=1$.
 \end{definition} 

The motivation for considering (strongly) time-like $p$-dualisable pairs of measures is twofold: firstly the $p$-optimal coupling $d\pi(x,y)$ matches events described by $d\mu(x)$  with events described by $d\nu(y)$ so that $x\ll y$; secondly
 Kantorovich duality holds (see \cite[Proposition 2.7]{suhr} in smooth Lorentzian setting and in case $p=1$ and \cite[Section 2.4]{CaMoLor} for the non-smooth setting and general $p\in (0,1]$).
 
 Motivated by the work of McCann \cite{McCann18} and by Theorem \ref{thm:RiccigeqT}, in \cite{CaMoLor} the following synthetic notion of time-like Ricci curvature bounded below by $K\in \R$ and dimension bounded above by $N\in (0,\infty] $ was proposed:
   
 \begin{definition}[$\mathsf{TCD}^{e}_{p}(K,N)$ and  $\mathsf{wTCD}^{e}_{p}(K,N)$ conditions]\label{def:TCD(KN)}
Fix $p\in (0,1)$, $K\in \R$, $N\in (0,\infty]$. We say that measured  Lorentzian  pre-length space $(X,\sfd,\mm, \ll, \leq, \tau)$ satisfies  $\mathsf{TCD}^{e}_{p}(K,N)$ 
(resp. $\mathsf{wTCD}^{e}_{p}(K,N)$)  if the following holds.
For any couple $(\mu_{0},\mu_{1})\in ({\rm Dom}(\Ent(\cdot|\mm)))^{2}$ which is   time-like $p$-dualisable   
(resp. $(\mu_{0},\mu_{1})\in [{\rm Dom}(\Ent(\cdot|\mm))\cap  \mathcal{P}_{c}(X)]^{2}$  
which is   strongly time-like $p$-dualisable) by some 
$\pi\in \Pi^{p\text{-opt}}_{\ll}(\mu_{0},\mu_{1})$,  
there exists an  $\ell_{p}$-geodesic $(\mu_{t})_{t\in [0,1]}$ such that  
the function $[0,1]\ni t\mapsto e(t) : = \Ent(\mu_{t}|\vol_{g})$ is 
semi-convex (and thus in particular it is locally Lipschitz in $(0,1)$) and it satisfies
\begin{equation}\label{eq:conveKN}
e''(t) - \frac{1}{N} e'(t)^{2 } \geq K \int_{X\times X} \tau(x,y)^{2} \, d\pi(x,y),
\end{equation}
in the distributional sense on $[0,1]$, where we adopt the convention that the second adding term in the left hand side is zero if $N=\infty$.
\end{definition}

The following stability result for synthetic time-like Ricci lower bounds was proved in  \cite[Theorem 3.14]{CaMoLor}.

 \begin{theorem}[Weak stability of  $\TCD^{e}_{p}(K,N)$]\label{thm:StabTCD}
Let  $\{(X_{j},\sfd_{j}, \mm_{j}, \ll_{j}, \leq_{j}, \tau_{j})\}_{j\in \N\cup\{\infty\}}$ be a sequence of  measured Lorentzian geodesic spaces  satisfying the following properties:
\begin{enumerate}
\item There exists a locally causally closed, $\cK$-globally hyperbolic  Lorentzian geodesic space $(X,  \sfd,   \ll,  \leq,  \tau)$ such that each $(X_{j},\sfd_{j}, \mm_{j}, \ll_{j}, \leq_{j}, \tau_{j})$, $j\in \N\cup\{\infty\}$, is isomorphically embedded in it (as in ${\rm (1)}$ of Theorem \ref{thm:stabUB}).
\item The measures $(\iota_{j})_{\sharp} \mm_{j}$  converge to $(\iota_{\infty})_{\sharp} \mm_{\infty}$ weakly in duality with $C_{c}( X)$ in $ X$, i.e. \eqref{eq:weakconv} holds.
\item There exist $p\in (0,1), K\in \R, N\in (0,\infty]$ such that   $(X_{j},\sfd_{j}, \mm_{j}, \ll_{j}, \leq_{j}, \tau_{j})$  satisfies  $\TCD^{e}_{p}(K,N)$, for each $j\in \N$.
\end{enumerate}
Then  the limit space $(X_{\infty},\sfd_{\infty}, \mm_{\infty}, \ll_{\infty}, \leq_{\infty}, \tau_{\infty})$  satisfies the $\wTCD^{e}_{p}(K,N)$ condition.
\end{theorem}

 \subsection{Synthetic vacuum Einstein's equations,}
 
 Combining the synthetic upper and lower bounds on the time-like Ricci curvature, i.e. Definitions \ref{def:RUB} and \ref{def:TCD(KN)}, it is natural to propose the following synthetic version for the  vacuum Einstein's equations (with possibly non-zero cosmological constant):
 
  \begin{definition}[Synthetic vacuum Einstein's equations]\label{def:SVEE}
  Fix $p\in (0,1)$, $\Lambda\in \R$, $N\in (0,\infty]$. We say that the measured  Lorentzian  pre-length space $(X,\sfd,\mm, \ll, \leq, \tau)$ satisfies the (resp. weak) synthetic formulation of the vacuum Einstein equations $\Ric\equiv \Lambda$ with cosmological constant $\Lambda\in \R$ and has synthetic dimension $\leq N$ if 
  \begin{itemize}
  \item  $(X,\sfd,\mm, \ll, \leq, \tau)$ satisfies the $\TCD^{e}_{p}(\Lambda,N)$ (resp. $\wTCD^{e}_{p}(\Lambda,N)$) condition;
  \item There exists $r_{0}>0$ and a function  $\omega:[0,r_{0})\to [0,\infty)$ with $\lim_{r\downarrow 0} \omega(r)=0$ such that $(X,\sfd,\mm, \ll, \leq, \tau)$ has time-like Ricci curvature bounded above by $\Lambda$, with respect to  $p\in (0,1), r_{0}$ and $\omega$.
    \end{itemize} 
  
  \end{definition}
 
 The combination of the stability of time-like Ricci lower and upper bounds (i.e. Theorem \ref{thm:stabUB} and Theorem \ref{thm:StabTCD}) gives the stability of the synthetic vacuum Einstein's equations under the aforementioned natural Lorentzian variant of measured Gromov-Hausdorff convergence. 
 
  \begin{theorem}[Weak stability of synthetic vacuum Einstein's equations]\label{thm:StabVEE}
Let  $\{(X_{j},\sfd_{j}, \mm_{j}, \ll_{j}, \leq_{j}, \tau_{j})\}_{j\in \N\cup\{\infty\}}$ be a sequence of  measured Lorentzian geodesic spaces  satisfying the following properties:
\begin{enumerate}
\item There exists a locally causally closed, $\cK$-globally hyperbolic  Lorentzian geodesic space $(X,  \sfd,   \ll,  \leq,  \tau)$ such that each $(X_{j},\sfd_{j}, \mm_{j}, \ll_{j}, \leq_{j}, \tau_{j})$, $j\in \N\cup\{\infty\}$, is isomorphically embedded in it (as in ${\rm (1)}$ of Theorem \ref{thm:stabUB}).
\item The measures $(\iota_{j})_{\sharp} \mm_{j}$  converge to $(\iota_{\infty})_{\sharp} \mm_{\infty}$ weakly in duality with $C_{c}( X)$ in $ X$, i.e. \eqref{eq:weakconv} holds.
\item  Volume non-collapsing: there exists a function $v:(0,\infty)\to (0,\infty)$ such that for every $x_{j}\in X_{j}$ it holds $\mm_{j}(B^{\sfd_{j}}_{r}(x_{j})) \geq v(r)>0$.
\item There exist $p\in (0,1), \Lambda \in \R, N\in (0,\infty], r_{0}>0$ and $\omega:[0,r_{0})\to [0,\infty)$ with $\lim_{r\downarrow 0} \omega(r)=0$ such that   $(X_{j},\sfd_{j}, \mm_{j}, \ll_{j}, \leq_{j}, \tau_{j})$  satisfies the synthetic formulation of the vacuum Einstein equations $\Ric\equiv \Lambda$ with cosmological constant $\Lambda\in \R$, with synthetic dimension $\leq N$ with respect to $p\in (0,1), r_{0}$ and $\omega$ as in Definition \ref{def:SVEE}.
\end{enumerate}
Then  the limit space $(X_{\infty},\sfd_{\infty}, \mm_{\infty}, \ll_{\infty}, \leq_{\infty}, \tau_{\infty})$  satisfies the weak synthetic formulation of the vacuum Einstein equations $\Ric\equiv \Lambda$ with cosmological constant $\Lambda\in \R$, with synthetic dimension $\leq N$ with respect to $p\in (0,1), r_{0}+1$ and $\omega$ as in Definition \ref{def:SVEE}.
\end{theorem}

\begin{remark}
By \cite[Theorem 3.1]{CaMoLor} after \cite{McCann18} (see also Corollary \ref{thm:RicciLowerBound}) and by Theorem \ref{thm:RiccileqT} (see also Remark \ref{rem:RiccileqT}), if $(X,\sfd,\mm, \ll, \leq, \tau)$ is a (for simplicity say a compact subset in a) smooth Lorentzian manifold, then $(X,\sfd,\mm, \ll, \leq, \tau)$ satisfies the Einstein's equations $\Ric\equiv \Lambda$ in the smooth classical sense if and only if $(X,\sfd,\mm, \ll, \leq, \tau)$ satisfies the Einstein's equations in the synthetic sense of Definition   \ref{def:SVEE}.  
Therefore, Theorem \ref{thm:StabVEE} gives that the corresponding limits of smooth solutions to Einstein's equation  $\Ric\equiv \Lambda$  satisfy the weak synthetic Einstein's equations $\Ric\equiv \Lambda$ in the sense of Definition   \ref{def:SVEE}.  In other terms, the vacuum Einstein's equations are stable under the conditions (and with respect to the notion of convergence) of   Theorem \ref{thm:StabVEE}.
\\

Let us mention that the stability of the Einstein's equations under various notions of (weak) convergence is a topic of high interest in General Relativity.
\\Classically, the problem is phrased in terms of convergence of a sequence of Lorentzian metrics $g_j$ converging to a limit Lorentzian metric $g_{\infty}$, \emph{on a fixed underlying manifold}.
\\It is well known that, if $g_j$ are solutions of the vacuum Einstein equations,  $g_j\to g_\infty$ in $C^{0}_{loc}$ and the derivatives of $g_j$ converge in $L^2_{loc}$, then the limit $g_{\infty}$ satisfies the vacuum Einstein equations as well.
\\However, if the $g_j\to g_\infty$ in $C^{0}_{loc}$ and the derivatives of $g_j$ converge \emph{weakly} in $L^2_{loc}$, explicit examples are known (see \cite{Burn,GrWa} for examples  in symmetry classes) where the limit $g_{\infty}$ may satisfy the Einstein equations with a \emph{non-vanishing} stress energy momentum tensor.  Burnett \cite{Burn} conjectured that, if there exist $C>0$ and $\lambda_{j}\to 0$ such that
$$
|g_{j}- g_{\infty}| \leq \lambda_{j}, \quad |\partial g_{j}|\leq C, \quad |\partial^{2} g_{j}| \leq C \lambda_{j}^{-1},
$$
then $g_{\infty}$ is isometric to a solution to the Einstein-massless Vlasov system for some appropriate choice of Vlasov field. Such a conjecture remains open,  although there has been recent progress \cite{HuLuk1, HuLuk2} when $g_{j}$ are assumed to be $\mathbb{U}(1)$-symmetric. We also mention the recent work \cite{LukRod} where concentrations (at the level of $\partial g_{j}$) are allowed in addition to oscillations.
\\

Theorem \ref{thm:StabVEE} gives a new point of view on the stability of the vacuum Einstein's equations. Indeed, while in the aforementioned results the  metrics $g_{j}$ are converging \emph{on a fixed underlying manifold}, in Theorem \ref{thm:StabVEE} also the underlying space $X$ may vary (along the sequence and at the limit), allowing change in topology in the limit, as one may expect in case of formation of singularities. Moreover, the notion of convergence is quite different in spirit: while in the aforementioned results $g_{j}\to g_{\infty}$ in a suitable \emph{functional analytic sense}, in Theorem \ref{thm:StabVEE} the spaces are converging in a \emph{more geometric sense}  (inspired by the measured Gromov-Haudorff convergence).
\end{remark}

\end{appendix}

\end{document}